\documentclass[letterpaper, 10 pt, journal]{IEEEtran}

\IEEEoverridecommandlockouts

\usepackage{graphicx} 
\usepackage{epsfig} 
\usepackage{amsmath} 
\usepackage{amssymb} 
\usepackage{amsthm}
\usepackage{multirow}
\usepackage{upgreek}
\usepackage{lettrine}
\usepackage{subfigure}
\usepackage{mathtools,amssymb,lipsum}
\usepackage{cuted}
\usepackage{mathtools}
\usepackage{cite}
\usepackage{caption}
\usepackage{enumitem}
\usepackage{setspace}
\usepackage[hidelinks]{hyperref}
\usepackage{pgf}

\newtheorem{thm}{Theorem}

\newtheorem{cor}{Corollary}
\newtheorem{lem}{Lemma}
\newtheorem{assumption}{Assumption}
\newtheorem{definition}{Definition}
\newtheorem{remark}{Remark}
\newtheorem{example}{Example}
\newtheorem{problem}{Problem}

\definecolor{LB}{RGB}{65,105,225}
\definecolor{LO}{RGB}{210,127,0}
\definecolor{LP}{cmyk}{0, 0.7808, 0.4429, 0.1412}
\definecolor{LG}{rgb}{0.13, 0.55, 0.13}
\definecolor{LR}{rgb}{0.77, 0.01, 0.2}
\definecolor{LV}{rgb}{0.5, 0.0, 0.5}

\title{Phase-Coordinated Circular Formation Control of Unicycle Agents Under Non-Concentric Boundary Constraints}

\author{Shubham Singh and Anoop Jain,~\IEEEmembership{Senior Member,~IEEE}
\thanks{The authors are with the Department of Electrical Engineering, Indian Institute of Technology Jodhpur, 342030, India (e-mail:  {shubham.1@iitj.ac.in}; {anoopj@iitj.ac.in}).}}

\begin{document}
\maketitle


\begin{abstract}
This paper addresses the problem of collective circular motion control for unicycle agents, with the objective of achieving phase coordination of their velocity vectors while ensuring that their trajectories remain confined within a prescribed non-concentric circular boundary. To handle such nonuniform spatial constraints, we build upon our earlier work and extend the use of the M\"{o}bius transformation to a multi-agent framework. The M\"{o}bius transformation maps two nonconcentric circles to concentric ones, thereby converting spatially nonuniform constraints into uniform constraints in the transformed plane. Leveraging this property, we introduce a phase-shifted order parameter, along with the notions of M\"{o}bius phase-shift-coupled synchronization and balancing, to characterize the phase-coordinated patterns of interest. In particular, synchronization corresponds to agents sharing a common velocity direction, whereas balancing corresponds to configurations in which the agents are distributed along the desired orbit such that the vector sum of their velocities is zero. We establish an equivalence between the unicycle dynamics in the original and transformed planes under the M\"{o}bius transformation and its inverse, and show that synchronization is preserved across both planes, while balancing is generally not. Distributed control laws are then designed in the transformed plane using barrier Lyapunov functions, under the assumption of an undirected and connected interaction topology. These controllers are subsequently mapped back to the original plane to obtain the linear acceleration and turn-rate inputs applied to the physical agents. Both simulations and experimental results are provided. 
\end{abstract}



\section{Introduction}
 
\subsection{Motivation and Background}
\lettrine{N}{onuniform} motion constraints frequently arise in practical scenarios such as autonomous vehicles operating in complex environments or drone fleets performing surveillance in hazardous, confined regions \cite{zhao2017general,molnar2025collision}. In such settings, robots (or more broadly, agents) must operate within spatially dependent boundaries to ensure their physical safety. These nonuniform constraints pose significant challenges in the design of distributed controllers for multi-agent systems executing cooperative tasks. This paper addresses this challenge by developing control laws that stabilize multi-agent systems in a \emph{phase-coordinated} manner around a common circular orbit, while ensuring that their trajectories remain bounded within a prescribed, non-concentric external circular boundary. This problem falls within the broader domain of geofencing control \cite{hermand2018constrained}, which seeks not only to restrict agents motion to a designated region but also to preserve the underlying communication topology, thereby preventing loss of connectivity due to sensing limitations. In addition, the proposed distributed controller explicitly accounts for the nonholonomic motion constraints inherent in robotic platforms such as unmanned ground vehicles (UGVs). The notion of phase coordination considered in this work is characterized through the collective behavior of agents' unit velocity vectors (or phasors), closely related to the \emph{order parameter} studied in coupled oscillator networks \cite{strogatz2000kuramoto}, from which the terminology ``phase" is adopted.

This paper presents an extension to our earlier work \cite{singh2025mobius}, where the problem of constrained circular motion control was solved for a single robot modeled using unicycle kinematics with only a turn-rate controller. However, such a control mechanism may not be sufficient for achieving sophisticated collective goals in a multi-agent framework subject to motion constraints. In these scenarios, it is necessary to exploit the full control authority of unicycle dynamics, i.e., by actuating both the linear and angular velocities \cite{brinon2014cooperative,panagou2015distributed,han2019robust,zhang2022tracking}. Motivated by this, we consider unicycle agents with controllable linear and angular speeds, and develop a framework for enforcing non-concentric boundary constraints in a multi-agent setting.

\begin{figure*}[t!]
	\centering
	\subfigure{\includegraphics[width=12cm]{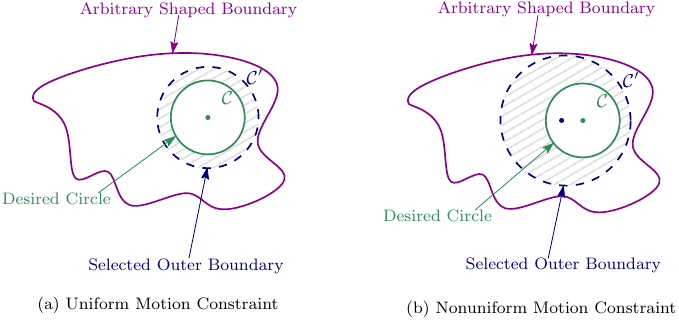}}
	\caption{For any arbitrarily shaped geofencing boundary, non-concentric circles allow greater flexibility in selecting the outer boundary, thereby allowing a wider range of feasible initial conditions compared to concentric circles.}
	\label{fig_arbitary_boundary}
\end{figure*}

\subsection{Related Work}
Collective circular motion control of multi-agent systems has been widely studied across diverse application domains, with research efforts primarily focusing on two aspects: (i) designing control laws to generate desired motion patterns along a common circular path, and (ii) identifying the minimum information requirements for distributed implementation. In the first category, \cite{brinon2015distributed} addressed source-seeking by stabilizing unicycle agents on a common circular orbit with uniform angular spacing, thereby eliminating measurement errors in signal intensity. In target-tracking applications, there exists works dealing with both stationary and moving targets. For stationary targets, distributed circular motion control has been developed under ring-coupled interaction graphs \cite{zheng2015distributed}, directed sensing graphs \cite{yu2022decentralized}, and with prescribed inter-agent spacing without direct distance measurements \cite{yu2018distributed}. For moving targets with time-varying velocity, \cite{yu2016cooperative} proposed evenly spaced circular formations using only local measurements expressed in the Frenet–Serret frame. A framework based on affine transformations to circular formation was proposed in \cite{brinon2014cooperative} to stabilize complex time-varying formations, useful in navigating narrow passages with a moving reference center. For oceanographic data collection, \cite{leonard2007collective} proposed sensor networks exploiting phase patterns of agent motion around a common circular path. These phase patterns are characterized by two fundamental formation primitives, namely synchronization and balancing, introduced in \cite{sepulchre2007stabilization}. These concepts have been generalized to nonidentical constant speeds, where agents evolve on individual circular trajectories while preserving collective phase patterns \cite{seyboth2014collective,sun2018circular}. In the second category, attention has been given to minimizing sensing and communication requirements for controller implementation. For instance, \cite{chen2013remark} proposed a framework without reliance on reference beacons, \cite{cao2016collective} developed range-based control strategies, \cite{summers2009coordinated} introduced Lyapunov-based guidance vector fields, and \cite{sun2018circular} studied displacement- and distance-based approaches. The underlying interaction topology has also been shown to critically influence stability and convergence properties; see, e.g., \cite{zheng2015distributed,yu2022decentralized,yu2018distributed,sepulchre2008stabilization}.

While the aforementioned works provide elegant control frameworks for collective circular motion, their applicability to safety-critical scenarios remains limited. In applications such as border surveillance, it is imperative that all agents remain confined within a predefined region at all times to ensure physical safety. Some relevant works in this direction are \cite{jain2019trajectory} and \cite{hegde2023synchronization} where the unicycle agents were stabilized to a circular \cite{jain2019trajectory} or a general polar curve \cite{hegde2023synchronization} under phase synchronization and balancing, while ensuring that their trajectories remained bounded within concentric safety boundaries. However, in practice, the safety region is not necessarily concentric with the desired path, leading to spatially nonuniform motion constraints that must be explicitly addressed in controller design. This paper contributes to this problem by enabling collective stabilization on a desired circular orbit, while simultaneously realizing different phase patterns in the agents' motion under such non-concentric boundary constraints. Such phase patterns are particularly relevant in applications requiring prioritized monitoring of specific regions, such as border tracking and surveillance.

\subsection{Methodology}
To account for system and operational constraints, control design approaches in the literature can broadly be classified into two categories: prescribed performance control (PPC) \cite{kanakis2019guaranteeing} and barrier function–based methods \cite{ames2019control,tee2009barrier,panagou2015distributed}. In the PPC framework, constraints are typically encoded via time-varying performance functions; however, such formulations are not directly suitable for the present problem, since the constraints exhibit \emph{spatial dependence} and are thus difficult to capture using purely time-varying functions. Barrier function-based approaches, on the other hand, encode constraints through either control barrier functions (CBFs) \cite{ames2019control} or barrier Lyapunov functions (BLFs) \cite{tee2009barrier,panagou2015distributed}. CBF-based controllers are commonly implemented via convex quadratic programs (QPs), which can often be solved efficiently in real time; nevertheless, such an implementation requires each agent to solve a QP at every control instant. BLF-based approaches, in contrast, extend control Lyapunov function (CLF) techniques to constrained systems and can facilitate the derivation of explicit control laws. Accordingly, we adopt the BLF framework in this work to obtain a closed-form control law for each agent, thereby avoiding online optimization altogether. Furthermore, a direct formulation in the original plane would require the online evaluation of the \emph{rate of change of radial separation} between the non-concentric boundaries, which can be computationally demanding.

To address these challenges, we employ the M\"{o}bius transformation \cite{priestley2003introduction,needham2023visual,turyn2014advanced}, which maps two nonconcentric circles to concentric ones, thereby converting spatially nonuniform constraints into uniform constraints in the transformed plane. Please note that the transformation is not restricted to position mapping; rather, it operates at the motion level, capturing both the kinematics and control inputs of the unicycle agents. Within this framework, we first design distributed control laws in the transformed plane and subsequently map them back to the original plane to obtain the controllers applied to the physical agents. It is worth mentioning that our framework can be extended to \emph{arbitrary-shaped} geofences by introducing a safety circle $\mathcal{C}'$ that is tangent to the irregular boundary at one or more points while enclosing the desired circular path $\mathcal{C}$. Moreover, such consideration of non-concentric boundaries provides greater flexibility in selecting initial conditions and maneuvering multiple nonholonomic agents (see Fig.~\ref{fig_arbitary_boundary}), as compared with the concentric constraints considered in \cite{jain2019trajectory}.

\subsection{Main Contributions}  
We introduce the notion of a phase-shifted order parameter that incorporates a position-dependent phase shift in the agents' velocity vectors, in contrast to the classical order parameter. Since this phase shift arises naturally from the M\"{o}bius transformation applied at each agent's position, we characterize two new motion patterns, referred to as M\"{o}bius phase-shift coupled synchronization and M\"{o}bius phase-shift coupled balancing, in the agents' velocity vectors while stabilizing them on a common circular orbit within non-concentric boundaries. To address the nonuniform trajectory constraints, we build upon the M\"{o}bius transformation framework in \cite{singh2025mobius} and show how a unicycle model with linear acceleration and turn-rate controllers can be mapped to an equivalent unicycle model with the same control structure in the transformed plane. We formally establish this equivalence and analyze the implications of the M\"{o}bius and inverse M\"{o}bius mappings. We show that synchronization is preserved across both planes under the M\"{o}bius transformation, whereas balancing is not. Leveraging this equivalence, we design distributed controllers in the transformed plane using a barrier Lyapunov function approach, and subsequently recover the corresponding controllers in the original plane via explicit parameter mappings across the two planes. We prove that, under an undirected and connected communication graph (resp., circulant topology), the agents asymptotically converge to the desired circular path while achieving M\"{o}bius phase-shift coupled synchronization (resp., balancing), with trajectories guaranteed to remain within the prescribed non-concentric boundary. As part of the design framework, we further address the selection of two critical design parameters that allow the unicycle dynamics in both planes to be expressed in polar form, thereby enabling analysis through the M\"{o}bius transformation formulation. We also derive tighter analytical bounds on key post-design signals, including positional errors, agent positions, and relative velocity vectors, in both planes. Finally, the effectiveness of the proposed controllers is demonstrated through numerical simulations and experimental implementation on Khepera IV robots. 


\section{Preliminaries}\label{section_2_preliminaries}

\subsection{Notations and Graph Theory}
	The sets of real, natural, complex, nonnegative real, and integer numbers are denoted by $\mathbb{R}$, $\mathbb{N}$, $\mathbb{C}$, $\mathbb{R}_{+}$, and $\mathbb{Z}$, respectively. The imaginary unit is $i \triangleq \sqrt{-1}$. For any $z \in \mathbb{C}$, $\Re(z)$ and $\Im(z)$ denote its real and imaginary parts, and $|\bullet|$ denotes the modulus of a real or complex number. The unit circle in the complex plane is denoted by $\mathbb{S}^1 \subset \mathbb{C}$, and the $N$-torus is $\mathbb{T}^N \triangleq \mathbb{S}^1 \times \cdots \times \mathbb{S}^1$ ($N$ times), where $\times$ denotes the Cartesian product. For ${\rm e}^{i\Upsilon_j}, {\rm e}^{i\Upsilon_k} \in \mathbb{S}^1$, we write $\Upsilon_j \equiv \Upsilon_k \pmod{2\pi}$ if $\Upsilon_j - \Upsilon_k = 2k_0\pi$ for some $k_0 \in \mathbb{Z}$. The inner product of $z_1, z_2 \in \mathbb{C}$ is defined as $\langle z_1, z_2 \rangle \triangleq \Re(\bar{z}_1 z_2)$, where $\bar{z}_1$ is the complex conjugate of $z_1$. For vectors $\pmb{w}_1, \pmb{w}_2 \in \mathbb{C}^N$, the inner product is $\langle \pmb{w}_1, \pmb{w}_2 \rangle \triangleq \Re(\pmb{w}_1^\ast \pmb{w}_2)$, with $\pmb{w}_1^\ast$ denoting the conjugate transpose. The superscript $\top$ denotes transpose of a real or complex vector/matrix. The vectors of all zeros and ones are $\pmb{0}_N = [0,\dots,0]^\top \in \mathbb{R}^N$ and $\pmb{1}_N = [1,\dots,1]^\top \in \mathbb{R}^N$, respectively. For $\pmb{\Upsilon} = [\Upsilon_1,\dots,\Upsilon_N]^\top \in \mathbb{R}^N$, we define ${\rm e}^{i\pmb{\Upsilon}} \triangleq [{\rm e}^{i\Upsilon_1},\dots,{\rm e}^{i\Upsilon_N}]^\top \in \mathbb{T}^N$. For a function $V(\pmb{x}): \mathbb{R}^N \to \mathbb{R}_{+}$ with $\pmb{x} = [x_1,\dots,x_N]^\top \in \mathbb{R}^N$, its gradient is $\nabla V_{\pmb{x}} \triangleq [\partial V/\partial x_1,\dots,\partial V/\partial x_N]^\top$. The space of continuously differentiable functions of order $N$ is denoted by $\mathcal{C}^N$.

	A graph is a pair $\mathcal{G} = (\mathbb{V}, \mathbb{E})$, consisting of a finite set of vertices $\mathbb{V}$, and a finite set of edges $\mathbb{E} \subseteq \mathbb{V}\times \mathbb{V}$. An edge $\mathfrak{e} = (j, k) \in \mathbb{E}$ between nodes $j$ and $k$ indicates that $j \in \mathcal{V}$ is the head and $k \in \mathcal{V}$ is the tail. An undirected edge is denoted by $\mathfrak{e} = \{j, k\} \in \mathbb{E}$ and indicates that the information can be shared from node $j$ to node $k$ and vice-versa. The graph $\mathcal{G}$ is called an undirected graph if it consists of only undirected edges. The set of neighbors of node $j$ is denoted by $\mathcal{N}_j \triangleq \{k \mid (j, k) \in \mathbb{E}\}$. The incidence matrix $\mathcal{B}\in\mathbb{R}^{|\mathbb{V}|\times|\mathbb{E}|}$ and Laplacian matrix $\mathcal{L}\in\mathbb{R}^{|\mathbb{V}|\times|\mathbb{V}|}$ are defined in the standard way \cite{mesbahi2010graph}, with $\mathcal{L}=\mathcal{B}\mathcal{B}^\top$ for undirected graphs. For $\pmb{z}\in\mathbb{C}^N$, the Laplacian quadratic form $Q_{\mathcal{L}}(\pmb{z})=\langle \pmb{z},\mathcal{L}\pmb{z}\rangle$ is positive semi-definite, and $Q_{\mathcal{L}}(\pmb{z})=0$ iff $\pmb{z}=z_0\pmb{1}_N$ for some $z_0\in\mathbb{C}$. A graph is circulant if and only if its Laplacian is a circulant matrix \cite[pp.~29$-$30]{mesbahi2010graph}. For instance, a ring topology is a special case of a circulant graph, where each node has degree two and is connected to its immediate neighbors. Throughout this paper, agents are assumed to interact over an undirected and limited communication topology that remains connected (i.e., $\mathcal{G}$ contains a spanning tree).
	
\subsection{A Review of M\"{o}bius Transformation}	
In this subsection, we provide a brief overview of the M\"{o}bius Transformation and its properties in mapping two non-concentric circles to the concentric circles from \cite{singh2025mobius}, which plays a key role in the subsequent developments in this paper. 

A M\"{o}bius transformation is a mapping of the form \cite[Chapter 2, pg. 23]{priestley2003introduction}, \cite[Chapter 3, pg. 137]{needham2023visual}:
\begin{equation}\label{mt}
	z \mapsto w = f(z) \triangleq \frac{az + b}{cz + d}, \quad a, b, c, d \in \mathbb{C}, ~ ad - bc \neq 0,
\end{equation}
where $ad-bc$ is referred to as the determinant of $f$. The condition $ad - bc \neq 0$ guarantees that $f$ is neither undefined (i.e., at least one of $c$ and $d$ is nonzero) nor one that is identically a constant, as can be seen in the case when $c \neq 0$ and $ad = bc \implies a/c = b/d = k$ for some $k \in \mathbb{C}$, \eqref{mt} maps to a constant, i.e., $f(z) = (ckz + dk)/(cz + d) = k$. The special case of $c = 0$ implies that neither $a$ nor $d$ is zero, and hence, $f(z)$ is well-defined. Clearly, the domain of $f$ is $\mathbb{C} \setminus \{-d/c\}$ if $c \neq 0$, and $\mathbb{C}$ if $c = 0$.  For all $z \in \mathbb{C} \setminus \{-d/c\}$, the derivative of $f$ is given by \cite[Chapter 8, pg. 95]{priestley2003introduction}:
\begin{equation}\label{mobius_transform_derivative}
	f'(z) = (ad - bc)/(cz + d)^2 \neq 0,	
\end{equation}
since $ad - bc \neq 0$. Hence, $f$ is conformal in $\mathbb{C} \setminus \{-d/c\}$, where $c \neq 0$. In the special case of $c = 0$, $f(z) = (az+b)/d$ is a linear mapping that is conformal everywhere. In the subsequent discussion, we do not explicitly mention the special case of $c =0$. Further, solving for $z$ in \eqref{mt}, the inverse of $f$ is obtained as \cite[Chapter 2, pg. 23]{priestley2003introduction}:
\begin{equation}\label{imt}
	w \mapsto z = f^{-1}(w) = (dw-b)/(a-cw),
\end{equation}
which has the same determinant $ad - bc \neq 0$, and hence, is also a M\"{o}bius transformation on the domain $\mathbb{C} \setminus \{a/c\}$. Note that $f(z) : \mathbb{C} \setminus \{-d/c\} \mapsto \mathbb{C} \setminus \{a/c\}$ defines a bijective and holomorphic map. For simplicity, \eqref{mt} is often regarded as a mapping from $\mathbb{C}_{\infty}$ to $\mathbb{C}_{\infty}$ such that $f(\infty) = a/c$ and $f(-d/c) = \infty$ if $c \neq 0$, and $f(\infty) = \infty$ if $c = 0$, in the extended complex plane $\mathbb{C}_{\infty} = \mathbb{C} \cup \{\infty\}$ (i.e., the complex plane augmented by the point at infinity). With the aid of $\mathbb{C}_{\infty}$, it can be deduced that $f: \mathbb{C}_{\infty} \to \mathbb{C}_{\infty}$ is bijective and conformal everywhere in its domain \cite[Chapter 2, pg. 23]{priestley2003introduction}. 

The M\"{o}bius transformation \eqref{mt} shares some remarkable geometric properties \cite[Ch. 3, p. 168]{needham2023visual}: (i) it maps circles (and lines) to circles, and (ii) if two points are symmetric with respect to a circle, then their images under \eqref{mt} are symmetric with respect to the corresponding image circle. The latter is referred to as the \emph{symmetry principle} \cite[Ch. 3, pp. 139--142]{needham2023visual}. By exploiting these properties, a M\"{o}bius transformation can be devised that maps \emph{any} pair non-intersecting, non-concentric circles in the $z$-plane to a pair of concentric circles in the $w$-plane. This property forms the foundation of our subsequent developments, which we formalize in the following lemma. In the subsequent discussion, the $z$-plane denotes the original (physical) domain in which the agents evolve, while the $w$-plane represents the transformed domain under \eqref{mt}.  

\begin{lem}\label{lem_mapping_of_circles}
	Let $\mathcal{C}: |z| = 1$ and $\mathcal{C}': |z - \lambda| = \mu$ be two circles, where $\lambda \neq 0$ and $\mu > 0$ are given real numbers. Suppose $\mathcal{C}$ and $\mathcal{C}'$ have no point in common (i.e., the two circles do not intersect or touch each other). We have the following results:
	\begin{itemize}[leftmargin=*]
		\item[(a)] (Mapping of non-concentric circles to concentric circles \cite[Chapter 16, pp. 1234$-$1240]{turyn2014advanced}). There exist real numbers $\alpha$ and $\beta$ such that the M\"{o}bius transformation 
	\begin{equation}\label{mobius_transformation}
		w = f(z) = (z+\alpha)/(z+\beta),
	\end{equation}
	maps both the circles $\mathcal{C}$ and $\mathcal{C}'$ to concentric circles centered at $0$ in the $w$-plane, as long as, it turns out that $z = -\beta$ is on neither $\mathcal{C}$ nor $\mathcal{C}'$. Here, $\beta = 1/\alpha$, where the real number $\alpha$ satisfies
	\begin{equation}\label{mobius_roots}
		\lambda\alpha^2 + (\lambda^2 - \mu^2 + 1)\alpha + \lambda = 0.
	\end{equation}
	Further, the images of the circles in the $w$-plane are given by $f(\mathcal{C}): |w| = |\alpha|$ and $f(\mathcal{C}'): |w| = |(\lambda + \alpha)/\mu|$, respectively.
	\item[(b)] (Mapping of enclosed regions \cite[Theorem~2]{singh2025mobius}). Let $\alpha_{+}$ and $\alpha_{-}$ be the roots of quadratic equation \eqref{mobius_roots}, corresponding to the positive and negative sign of its discriminant, respectively. Further, define the smaller root $\alpha_s$ and the larger root $\alpha_{\ell}$ based on their magnitudes as:
	\begin{align}
		\label{alpha_s}\alpha_s &= \{\alpha_{\pm} \mid |\alpha_s| = \min\{|\alpha_+|, |\alpha_{-}|\}\},\\
		\label{alpha_l}\alpha_{\ell} &= \{\alpha_{\pm} \mid |\alpha_\ell| = \max\{|\alpha_+|, |\alpha_{-}|\}\}.
	\end{align}
Then, the M\"{o}bius transformation \eqref{mobius_transformation} with $\alpha = \alpha_s$ (resp., $\alpha = \alpha_{\ell}$) preserves (resp., reverses) the interior-exterior mapping of the regions enclosed by the circles $\mathcal{C}$ and $\mathcal{C}'$ in the $w$-plane. Moreover, the mapping of the region enclosed between the two circles is always an annulus in the $w$-plane, irrespective of the roots $\alpha_s$ or $\alpha_{\ell}$.	
\end{itemize}
\end{lem} 

\begin{figure}[t!]
	\centering{
		\subfigure[$\lambda > 0$]{\includegraphics[width=3.0cm]{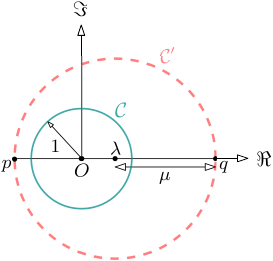}\label{A1}}\hspace*{15pt}
		\subfigure[$\lambda < 0$]{\includegraphics[width=3.0cm]{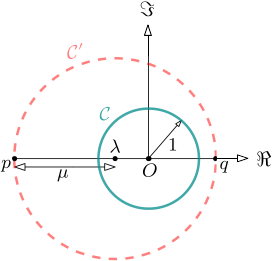}\label{A2}}}
	\caption{Circle $\mathcal{C}'$ encircling circle $\mathcal{C}$ with $\mu > 1 + |\lambda|$.}
	\label{lambda_mu_rel1}
	\vspace*{-15pt}
\end{figure}

\begin{figure}[t!]
	\centering{
	\subfigure[$\lambda > 0$]{\includegraphics[width=3.0cm]{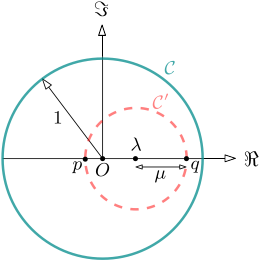}\label{A3}}\hspace*{15pt}
	\subfigure[$\lambda < 0$]{\includegraphics[width=3.0cm]{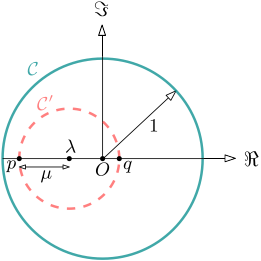}\label{A4}}}
	\caption{Circle $\mathcal{C}$ encircling circle $\mathcal{C}'$ with $\mu < 1 - |\lambda|$.}
	\label{lambda_mu_rel2}
\end{figure}

\begin{cor}[Roots of {\eqref{mobius_roots}} \cite{singh2025mobius}]\label{cor_quadractic_roots}
	Under the conditions in Lemma~\ref{lem_mapping_of_circles}, the roots $\alpha$ of \eqref{mobius_roots} satisfy the following: (a) their product is equal to unity, and (b) the point $z = -\beta$ lies neither on $\mathcal{C}$ nor $\mathcal{C}'$, iff the solutions of \eqref{mobius_roots} are not $\alpha = \pm 1$. 	
\end{cor} 

\begin{remark}
It is worth emphasizing that Lemma~\ref{lem_mapping_of_circles} applies to any arbitrary pair of non-concentric circles through an appropriate coordinate transformation (see \cite[Remark~2]{singh2025mobius} for details). Moreover, depending on the parameters $\lambda$ and $\mu$, two distinct geometric configurations arise for the circles $\mathcal{C}$ and $\mathcal{C}'$: (a) $\mathcal{C}'$ encloses $\mathcal{C}$ (Fig.~\ref{lambda_mu_rel1}), which occurs when $\mu > 1 + |\lambda|$, and (b) $\mathcal{C}$ encloses $\mathcal{C}'$ (Fig.~\ref{lambda_mu_rel2}), which occurs when $\mu < 1 - |\lambda|$. Importantly, the conclusions of Lemma~\ref{lem_mapping_of_circles} are valid in both \cite{singh2025mobius}.	
\end{remark}

Lemma~\ref{lem_mapping_of_circles} is illustrated in Fig.~\ref{fig_mobius_mapping}, which depicts how the three regions in the $z$-plane: (i) the annular region between the two circles, (ii) the interior of the inner circle, and (iii) the exterior of the outer circle, are mapped into the $w$-plane under both the smaller and larger roots of \eqref{mobius_roots}. To further clarify, the following example demonstrates Lemma~\ref{lem_mapping_of_circles} for the case where $\mathcal{C}'$ encloses $\mathcal{C}$ (see also \cite{singh2025mobius} for additional examples).
 
 \begin{example}\label{example_1}
For non-concentric circles $\mathcal{C}: |z| = 1$ and $\mathcal{C}': |z - ({1}/{2})| = \sqrt{{5}/{2}}$, \eqref{mobius_roots} yields $({1}/{2})\alpha^2 - (5/4)\alpha + ({1}/{2}) = 0$, since $\lambda = 1/2$ and $\mu = \sqrt{{5}/{2}}$. On solving this, we get two positive roots: $\alpha_{-} = 1/2$ and $\alpha_{+} = 2$. Selecting the smaller root $\alpha_s = 1/2$, \eqref{mobius_transformation} is given by $w= f_s(z) = (2z+1)/(2z+4)$, which preserves the interior-exterior relation with $\mathcal{C}$ mapping to $f_s(\mathcal{C}): |w| = {1}/{2}$ and $\mathcal{C}'$ to $f_s(\mathcal{C}'): |w| = \sqrt{2/5}$. For the larger root $\alpha_{\ell} = 2$, the transformation $w= f_{\ell}(z) = (2z+4)/(2z+1)$ reverses the interior–exterior relation with $\mathcal{C}$ mapping to $|w| = 2$ and $\mathcal{C}'$ to $|w| = \sqrt{{5}/{2}}$. Thus, the results are in accordance with Lemma~\ref{lem_mapping_of_circles}. 
\end{example}

\begin{figure}[t!]
	\centering{\hspace*{-0.5cm}
		\includegraphics[width=8.5cm]{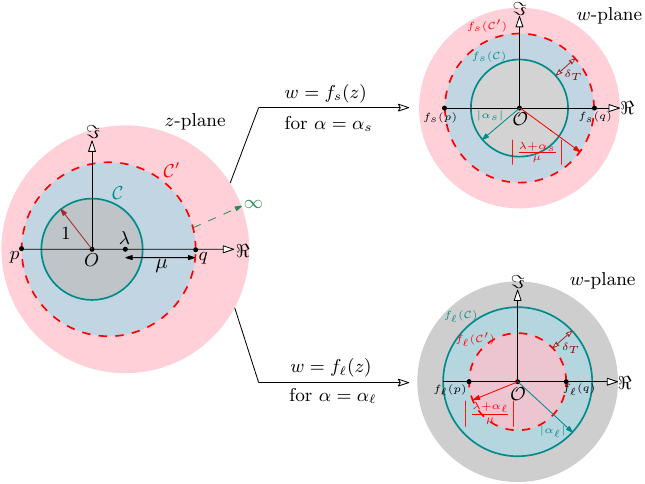}}
	\caption{Mapping of circles $\mathcal{C}$ and $\mathcal{C}'$ along with their enclosed regions under M\"{o}bius transformation \eqref{mobius_transformation}.}
	\label{fig_mobius_mapping}
\end{figure}

\section{System Description, Phase Coordination and Problem Formulation}\label{section_3_system_and_problem}

\subsection{Agent Model}
We consider a group of $N$ agents moving in a planar space $\mathbb{R}^2$. For ease of analysis, we identify $\mathbb{R}^2$ with the complex plane $\mathbb{C}$ via the mapping $[\mathfrak{p}, \mathfrak{q}]^\top \mapsto \mathfrak{p} + i\mathfrak{q}$, and represent the agents' nonholonomic motion using complex notation as:
\begin{equation}\label{model_actual_plane}
	\dot{r}_k = v_k\mathrm{e}^{i\theta_k}, \ \dot{v}_k = u_k, \ \dot{\theta}_k = \omega_k, \ \forall k = 1, \ldots, N,
\end{equation}
where $\dot{r}_k$ is the velocity vector of the $k^\text{th}$ agent, $\theta_k \in \mathbb{R}$ is its heading, $v_k \in \mathbb{R}$ is its speed. Further, $u_k \in \mathbb{R}$ and $\omega_k \in \mathbb{R}$ are the feedback controllers that govern the linear acceleration $\dot{v}_k$ and turn-rate $\dot{\theta}_k$ of the $k^\text{th}$ agent, respectively. Further, let the position of the $k^\text{th}$ agent be represented by
\begin{equation}\label{position_actual_plane}
	r_k = x_k + iy_k \triangleq |r_k|{\rm e}^{i\phi_k},
\end{equation}
where $[x_k, y_k]^\top \in \mathbb{R}^2$ are the Cartesian coordinates and $\phi_k \in \mathbb{R}$ denotes the angle of the position vector $r_k$ from the positive real axis. Please note that \eqref{model_actual_plane} describes the agent's model in the polar form where $v_k = |\dot{r}_k|$ is positive by convention. However, since $v_k$ evolves according to the controller $u_k$ in \eqref{model_actual_plane}, it may, in general, take either positive or negative values. In the subsequent analysis, we assume $v_k > 0$ and later show that there exist controllers $u_k$ and $\omega_k$ such that $v_k(t) > 0$ for all $k$ and $t \geq 0$, provided $v_k(0) > 0, \forall k$ (see Corollary~\ref{cor_positive_speeds}).

\subsection{Characterization of Phase Coordinated Motion}
In this paper, one of our main objectives is to achieve phase coordination among the agents \eqref{model_actual_plane} characterized by the so-called \emph{phase-shifted order parameter}, defined as follows:
\begin{definition}
The phase-shifted order parameter $q$ of the group of $N$ agents in \eqref{model_actual_plane} is defined by
\begin{equation}\label{shifted_order_parameter}
	q \triangleq \frac{1}{N}\sum_{k=1}^{N} {\rm e}^{i\Theta_k} = |q|{\rm e}^{i\bar{\Theta}},
\end{equation}
where $\Theta_k = \theta_k + \chi_k$ with $\chi_k \in \mathbb{R}$ being the phase-shift in the heading angle $\theta_k \in \mathbb{R}$ for each $k$, $|q|$ is the magnitude of $q$ and $\bar{\Theta}$ is the resultant angle of phasors ${\rm e}^{i\Theta_k}, \forall k$.
\end{definition}
 
From \eqref{shifted_order_parameter}, it can be easily observed that $0 \leq |q| \leq 1$. Clearly, $|q| = 1$ if $\Theta_1 = \Theta_2 = \cdots = \Theta_N$, i.e., the shifted phasors ${\rm e}^{i\Theta_k}$ are synchronized. Further, $q = 0$ if $\pmb{1}^\top_N {\rm e}^{i\pmb{\Theta}} = 0$, where ${\rm e}^{i\pmb{\Theta}} = [{\rm e}^{i\Theta_1}, \cdots, {\rm e}^{i\Theta_N}]^\top \in \mathbb{T}^N$, which corresponds to a balanced configuration where the phasors cancel. Depending on $\chi_k$, we have the following two special cases:  
\begin{itemize}[leftmargin=*]
	\item If $\chi_k \equiv 0$ for all $k$, \eqref{shifted_order_parameter} equates to $q = (1/N)\sum_{k=1}^{N} {\rm e}^{i\theta_k}$, which implies that $|q| = 1$ (resp., $q = 0$) corresponds to synchronization (resp., balancing) of the actual heading angles $\theta_k, \forall k$. This scenario has been addressed in literature \cite{sepulchre2007stabilization,sepulchre2008stabilization} where \eqref{shifted_order_parameter} is referred to as \emph{order parameter} \cite{strogatz2000kuramoto}. 
	\item If $\chi_k \equiv \chi$ for all $k$, where $\chi \in \mathbb{R}$ is a constant. In this situation, \eqref{shifted_order_parameter} equates to $q = {\rm e}^{i\chi}(1/N)\sum_{k=1}^{N}{\rm e}^{i\theta_k}$ which again implies the same for $|q|$ as stated above, since $|{\rm e}^{i\chi}| = 1$. Alternatively, \eqref{shifted_order_parameter} is invariant to constant rotation or phase-shift \cite{sepulchre2007stabilization,jain2018collective}. 
\end{itemize}

Unlike these works, in this paper, we address a scenario where $\chi_k$ is no longer a constant; instead, it depends on the position $r_k$ of the $k^\text{th}$ agent. To be specific, we consider that $\chi_k$ is introduced due to the M\"{o}bius transformation $f(r_k)$ at the position $r_k$ of the $k^\text{th}$ agent, and is given by the argument of its positional derivative $df(r_k)/dr_k$. For convenience, we refer to this angle $\chi_k$ as the M\"{o}bius phase-shift, and define it formally as follows:
\begin{definition}[M\"{o}bius phase-shift]\label{def_mobius_phase_shift}
Let $f(r_k)$ be the M\"{o}bius transformation at the position $r_k$ of the $k^\text{th}$ agent in \eqref{model_actual_plane}. The M\"{o}bius phase-shift $\chi_k$ at $r_k$ is defined as: 
\begin{equation}\label{chi_argument_form}
\chi(r_k) \triangleq \arg[{df(r_k)}/{dr_k}].
\end{equation}
\end{definition} 

The rationale for introducing $\chi_k$ in Definition~\ref{def_mobius_phase_shift} will become evident in the next subsection, where we reformulate the problem in a coordinate frame induced by the M\"{o}bius transformation. For notational convenience, we denote $f_k \triangleq f(r_k)$ and $\chi(r_k) \triangleq \chi_k$. In view of this, we now state the following definition, which formalizes our phase coordination objective for the agents \eqref{model_actual_plane} in terms of their heading angles $\theta_k$:

\begin{definition}[M\"{o}bius phase-shift-coupled synchronization and balancing]\label{def_mobius_synchronization_balancing}
Consider the phase-shifted order parameter \eqref{shifted_order_parameter} with $\chi_k$ being the M\"{o}bius phase-shift defined by \eqref{chi_argument_form}. A set of phases $\{\theta_k\}_{k=1}^{N}$ is said to achieve M\"{o}bius phase-shift-coupled synchronization if $|q| = 1$, that is, $\theta_1 + \chi_1 = \theta_2 + \chi_2 = \cdots = \theta_N + \chi_N$. A set of phases $\{\theta_k\}_{k=1}^{N}$ is said to achieve M\"{o}bius phase-shift-coupled balancing if $q = 0$, that is, $\pmb{1}^\top_N {\rm e}^{i(\pmb{\theta} + \pmb{\chi})} = 0$, where ${\rm e}^{i(\pmb{\theta} + \pmb{\chi})} = [{\rm e}^{i(\theta_1 + \chi_1)}, {\rm e}^{i(\theta_2 + \chi_2)}, \cdots, {\rm e}^{i(\theta_N + \chi_N)}]^\top$. 
\end{definition}

\subsection{Problem Description}
Our main control objective aims to achieve the aforementioned phase patterns among the agents around a common circular orbit $\mathcal{C}$, while ensuring that their trajectories remain bounded within a non-concentric circular boundary $\mathcal{C}'$. This situation is depicted in Fig.~\ref{fig_agents_motion} where the $k^\text{th}$ agent is trying to follow the desired circular orbit $\mathcal{C}$ of radius $r_d \in \mathbb{R}_{+}$ and center $c_d \in \mathbb{C}$. To enforce its motion on $\mathcal{C}$, the positional error $e_k$ between the agent's current position (point $\mathcal{Q}$) and the desired position (point $\mathcal{P}$) on $\mathcal{C}$ must be minimized. The point $\mathcal{P}$ corresponds to a point on $\mathcal{C}$ where the tangent is parallel to the agent's current velocity vector $\dot{r}_k$. Another such point on $\mathcal{C}$ is $\mathcal{P}'$, which is diametrically opposite to $\mathcal{P}$, i.e., the line segment joining $\mathcal{P}$ and $\mathcal{P}'$ passes through the center $\mathcal{C}_d$ of $\mathcal{C}$. Either of the tangent directions associated with $\mathcal{P}$ (resp., $\mathcal{P}'$) can be chosen as a convention in the control law characterizing the agent's motion in the anticlockwise (resp., the clockwise) direction on $\mathcal{C}$. Note that the unit tangent vector at $\mathcal{P}$, aligned with the $k^\text{th}$ agent's velocity direction, is ${\rm e}^{i\theta_k}$. A rotation of this vector by $\pi/2$ radians in the clockwise or anticlockwise direction yields the unit vectors $-i \mathrm{e}^{i \theta_k}$ (along $\mathcal{C}_d\mathcal{P}$) and $i \mathrm{e}^{i \theta_k}$ (along $\mathcal{P}\mathcal{C}_d$), respectively. Now, using vector addition rule, it can be written that $\mathcal{P}\mathcal{Q} = \mathcal{C}_d\mathcal{Q} - \mathcal{C}_d\mathcal{P} = (\mathcal{O}\mathcal{Q} - \mathcal{O}\mathcal{C}_d) - \mathcal{C}_d\mathcal{P}$ and $\mathcal{P}'\mathcal{Q} = \mathcal{C}_d\mathcal{Q} - \mathcal{C}_d\mathcal{P}' = (\mathcal{O}\mathcal{Q} - \mathcal{O}\mathcal{C}_d) - \mathcal{C}_d\mathcal{P}'$. Consequently, the error $e_k$ can be expressed as: 
\begin{equation}\label{error_primary}
e_k = (r_k - c_d) \pm i {r_d}{\rm e}^{i\theta_k},
\end{equation}
where $+$ (resp., $-$) sign corresponds to the robot's motion in the anticlockwise (resp., the clockwise) direction. 

\begin{figure}[t!]
	\centering{
		\includegraphics[width=6cm]{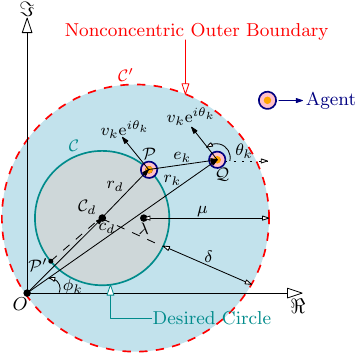}}
	\caption{The $k^\text{th}$ agent following desired circular orbit $\mathcal{C}$ with motion restricted within non-concentric outer boundary $\mathcal{C'}$.}
	\label{fig_agents_motion}
\end{figure}

Without loss of generality, and by applying a suitable coordinate transformation, we represent the two circles $\mathcal{C}$ and $\mathcal{C}'$ in Fig.~\ref{fig_agents_motion} within the framework of Lemma~\ref{lem_mapping_of_circles} as $\mathcal{C}: |z| = 1$ and $\mathcal{C}': |z - \lambda| = \mu$, respectively, where we assume $c_d = 0 + i0$ and $r_d = 1$. Under these conditions, \eqref{error_primary} can be expressed as:
\begin{equation}\label{error_actual_plane}
e_k = r_k  \pm i {\rm e}^{i\theta_k}, \ \forall k.
\end{equation}  
Further, we make the following practical assumption in our analysis, which ensures that there is sufficient separation between the two circular boundaries so that the agents, owing to their finite shape and size, can perform safe maneuvers. 
\begin{assumption}\label{assumption_cicular_boundaries} 
The non-concentric circles $\mathcal{C}$ and $\mathcal{C}'$ in Fig.~\ref{fig_agents_motion} are non-touching and non-intersecting. 
\end{assumption} 
This assumption is inherent in Lemma~\ref{lem_mapping_of_circles}. As illustrated in Fig.~\ref{lambda_mu_rel1}, for $\mathcal{C}'$ to encircle $\mathcal{C}$ under this condition, it must hold that $\mu > 1 + |\lambda|$. Clearly, the equality $\mu = 1 + |\lambda|$ corresponds to the case where $\mathcal{C}$ and $\mathcal{C}'$ touch each other, which is excluded from our discussion. We now formally state the problem addressed in this paper.
\begin{problem}\label{problem_actual_plane}
Consider the circles $\mathcal{C}: |z| = 1$ and $\mathcal{C}': |z - \lambda| = \mu$, where $\lambda$ and $\mu$ are such that Assumption~\ref{assumption_cicular_boundaries} holds with $\mathcal{C}'$ encircling $\mathcal{C}$. Let the agents \eqref{model_actual_plane} be interacting over an undirected and connected graph $\mathcal{G}$ with Laplacian $\mathcal{L}$ and begin their motion with initial positions $r_k(0) \in  \mathcal{Z} \triangleq \{z_r \in \mathbb{C} \mid |z_r - \lambda| < \mu\}$ and speeds $v_k(0) > 0$ for all $k$. Design the control laws $u_k$ and $\omega_k$ for all $k$, such that the agents asymptotically stabilizes on the desired circular orbit $\mathcal{C}$ with their heading angles $\{\theta_k\}_{k=1}^{N}$ in M\"{o}bius phase-shift-coupled synchronization or balancing as in Definition~\ref{def_mobius_synchronization_balancing}, and their trajectories remains bounded within non-concentric boundary $\mathcal{C}'$ for all times. That is, $e_k(t) \to 0$ (equivalently, $|r_k(t)| \to 1$) for all $k$, and that either $\theta_1 + \chi_1 = \cdots = \theta_N + \chi_N$ or $\pmb{1}^\top_N {\rm e}^{i(\pmb{\theta} + \pmb{\chi})} = 0$, as $t \to \infty$, while satisfying $r_k(t) \in \mathcal{Z}$ for all $k$ and all $t \geq 0$.
\end{problem}

Our solution approach to Problem~\ref{problem_actual_plane} relies on transforming the agents' motion \eqref{model_actual_plane} in the original coordinates to the transformed coordinates by using M\"{o}bius transformation \eqref{mobius_transformation}. In this direction, we first establish the equivalence between the two agent models in both planes, followed by the problem formulation in the transformed plane, which is the topic of the next section.

\section{Model Equivalence and Problem Formulation in Transformed Plane}\label{section_4_model_equivalence}

\subsection{Model Equivalence in Two Planes}
Under the M\"{o}bius transformation \eqref{mobius_transformation} in Lemma~\ref{lem_mapping_of_circles} (where $\beta = 1/\alpha$), the position of the $k^\text{th}$ agent is mapped from the original plane to the transformed plane as:
\begin{equation}\label{position_transformed_plane}
	\rho_k = f(r_k) = \frac{r_k + \alpha}{r_k + \beta} = \frac{\alpha(r_k + \alpha)}{1 + \alpha r_k} \triangleq |\rho_k|\mathrm{e}^{i\psi_k}, \ \forall k,
\end{equation}
where $|\rho_k|$ is the magnitude and $\psi_k \in \mathbb{R}$ is the direction of the transformed position vector $\rho_k \in \mathbb{C}$. It is worth mentioning that \eqref{mobius_transformation} is applied only on the positions of the $k^\text{th}$ robot, and hence, $\dot{\rho} \neq f(\dot{r})$. To obtain the velocity $\dot{\rho}_k$ of the $k^\text{th}$ robot in the transformed plane, we need to differentiate \eqref{position_transformed_plane} and define other states. Because $f(r_k)$ is holomorphic, one can easily obtain the time-derivative of \eqref{position_transformed_plane} using the chain rule as $\dot{\rho}_k = (df_k/dr_k)\dot{r}_k$, where
\begin{equation}\label{position_derivative_mobius_mapping}
{df_k}/{dr_k} = {\alpha(1 - \alpha^2)}/{(1 + \alpha r_k)^2},
\end{equation}
is derived according to \eqref{mobius_transform_derivative}. In this way, we relate the $k^\text{th}$ agent's velocities in two planes as:
\begin{equation}
	\dot{\rho}_k = \left(\frac{df_k}{dr_k}\right)\dot{r}_k = \frac{\alpha(1 - \alpha^2)\dot{r}_k}{(1 + \alpha r_k)^2} \triangleq |\dot{\rho}_k|{\rm e}^{i\gamma_k},
\end{equation}
where, 
\begin{subequations}\label{s_k_gamma_k}
	\begin{align}
		\label{s_k}	|\dot{\rho}_k| & = \left|\frac{df_k}{dr_k}\right||\dot{r}_k| = \left|\frac{\alpha(1 - \alpha^2)}{(1 + \alpha r_k)^2}\right|v_k \triangleq s_k,\\
		\label{gamma_k}	\gamma_k & = \arg[\dot{r}_k] + \arg[df_k/dr_k] = \theta_k + \chi_k,
	\end{align}
\end{subequations}
using \eqref{model_actual_plane} and \eqref{chi_argument_form}. With these notations, let the $k^\text{th}$ agent in the transformed plane be modeled as:
\begin{equation}\label{model_transformed_plane}
	\dot{\rho}_k = s_k{\rm e}^{i\gamma_k}, \ \dot{s}_k = \nu_k, \ \dot{\gamma}_k = \Omega_k, \ k = 1, \ldots, N,
\end{equation}
where $s_k > 0$ denotes the forward speed, $\gamma_k \in \mathbb{R}$ is the heading angle, and $\nu_k \in \mathbb{R}$ and $\Omega_k \in \mathbb{R}$ are the speed and turn-rate controllers in the transformed plane, respectively. Please note that $s_k$ is positive by definition \eqref{s_k}. As modeled in \eqref{model_transformed_plane}, since $\dot{s}_k$ is governed by the controller $\nu_k$, $s_k$ may, in principle, become negative depending on $\nu_k$, which could create ambiguity with \eqref{s_k}. However, we shall prove later in Corollary~\ref{cor_positive_speeds} that there exist design parameters such that the proposed controllers $\nu_k$ and $\Omega_k$ in \eqref{model_transformed_plane} ensure that $v_k(t) > 0$ (provided $v_k(0) > 0$) and $s_k(t) > 0$ for all $k$ and $t \geq 0$. 

Please note that \eqref{model_transformed_plane} is derived by applying the M\"{o}bius transformation on the agent's model \eqref{model_actual_plane} in the original plane. However, it is possible to establish an analogous relation between the two models by using the idea of inverse M\"{o}bius transformation. Notably, both models have their own importance in controller design as examined in the subsequent discussion. From \eqref{position_transformed_plane}, one arrives at the inverse mapping: 
\begin{equation}\label{inverse_mobius_transformation}
	\rho_k \mapsto r_k = f^{-1}(\rho_k) = \frac{\alpha^2-\rho_k}{\alpha(\rho_k-1)}, \ \forall k,
\end{equation}
which provides the agents' position in the original plane based on their positions in the transformed plane. Further, analogous to the above discussion, it can be written from \eqref{inverse_mobius_transformation} that $\dot{r}_k = ({df^{-1}_k}/{d \rho_k})\dot{\rho}_k = |\dot{r}_k|{\rm e^{\theta_k}}$, where 
\begin{equation}\label{position_derivative_inverse_mobius_mapping}
{df^{-1}_k}/{d\rho_k} = {(1 - \alpha^2)}/{\alpha(\rho_k - 1)^2},
\end{equation}
and
\begin{subequations}\label{r_k_theta_k}
	\begin{align}
		\label{v_k}	|\dot{r}_k| & = \left|\frac{df^{-1}_k}{d \rho_k}\right||\dot{\rho}_k| = \left|\frac{(1 - \alpha^2)}{\alpha(\rho_k - 1)^2}\right|s_k = v_k,\\
		\label{theta_k}	\theta_k & = \arg[\dot{\rho}_k] + \arg[{df^{-1}_k}/{d \rho_k}] = \gamma_k + \zeta_k,
	\end{align}
\end{subequations}
where we denote by
\begin{equation}\label{zeta_argument_form}
	\zeta_k = \arg[{df^{-1}_k}/{d \rho_k}].
\end{equation}

One can further simplify \eqref{chi_argument_form} and \eqref{zeta_argument_form} as mentioned in Lemma~\ref{lem_chi_zeta} in the Appendix. Note that \eqref{position_transformed_plane} has a singularity at $r_k = -1/\alpha$, and \eqref{inverse_mobius_transformation} has a singularity at $\rho_k = 1$. However, we shall show in the below Lemma~\ref{lem_transformation_derivative_properties} that both \eqref{position_transformed_plane} and \eqref{inverse_mobius_transformation} are well-defined in our framework. 

\begin{lem}\label{lem_transformation_derivative_properties}
Under the conditions specified in Lemma~\ref{lem_mapping_of_circles}, the following statements are true:
\begin{itemize}[leftmargin=*]
	\item[(a)] The mappings \eqref{position_actual_plane} and \eqref{inverse_mobius_transformation} are well-defined. 
	\item[(b)] Their positional derivatives in \eqref{position_derivative_mobius_mapping} and \eqref{position_derivative_inverse_mobius_mapping} are non-zero, that is, ${df_k}/{dr_k} \neq 0$ and  ${df^{-1}_k}/{d\rho_k} \neq 0$ for all $k$.   
\end{itemize} 
\end{lem}

\begin{proof}
We prove these by contradiction. 
\begin{itemize}[leftmargin=*]
	\item[(a)] Note that \eqref{mobius_transformation} is well-defined since the singularity at $r_k = -1/\alpha = -\beta$ contradicts our hypothesis in Lemma~\ref{lem_mapping_of_circles}, which states that $z = -\beta$ does not lie either on $\mathcal{C}$ or $\mathcal{C}'$. Further, \eqref{inverse_mobius_transformation} is well-defined in the sense that, at singularity $\rho_k = 1$, it follows from \eqref{position_transformed_plane} that $1 + \alpha r_k = \alpha(r_k + \alpha) \implies \alpha^2 = 1 \implies \alpha = \pm 1$, which again contradicts our hypothesis in Lemma~\ref{lem_mapping_of_circles} (see Corollary~\ref{cor_quadractic_roots}).
	
	\item[(b)] From the proof of part (a), it is clear that \eqref{position_derivative_mobius_mapping} and \eqref{position_derivative_inverse_mobius_mapping} are well-defined. Further, it is straightforward to observe that \eqref{position_derivative_mobius_mapping} and \eqref{position_derivative_inverse_mobius_mapping} are zero only if $\alpha = \pm 1$, which is a contradiction according to Corollary~\ref{cor_quadractic_roots}.
\end{itemize} 
This concluded the proof. 
\end{proof}

\begin{figure}[t!]
	\centering{
	\includegraphics[width=7cm]{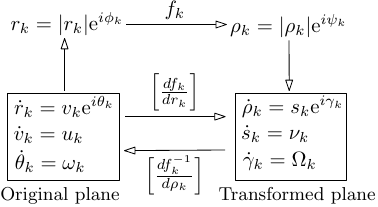}}
	\caption{Model equivalence in two planes.}
	\label{fig_model_mapping}
\end{figure}

A summary of model equivalence in both the planes is depicted in Fig.~\ref{fig_model_mapping}. Here, one may draw an analogy with the frequency response analysis for linear systems in classical control theory, where the output of the system, with respect to a sinusoidal input signal, is scaled and shifted in phase by the magnitude and argument of the system's transfer function, respectively. Similarly, in Fig.~\ref{fig_model_mapping}, $df_k/dr_k$ and $df^{-1}_k/d\rho_k$ may be regarded equivalent to a transfer function connecting two unicycle models in original and transformed planes such that the speed and heading angle of the $k^\text{th}$ agent in two planes is scaled and shifted in phase by the magnitude and argument of $df_k/dr_k$ and $df^{-1}_k/d\rho_k$, respectively, as evident by expressions \eqref{s_k_gamma_k} and \eqref{r_k_theta_k}. This is precisely the implication of the conformal property of the M\"{o}bius transformation \eqref{mobius_transformation}. Furthermore, the sense of rotation of the agents' velocity vectors in the transformed plane depends on the roots $\alpha_s$ and $\alpha_{\ell}$ of \eqref{mobius_roots}, as summarized in the below theorem from \cite{singh2025mobius}.

\begin{lem}[\hspace{-.1pt}Sense of rotation in two planes\cite{singh2025mobius}]\label{thm_velcoty_direction_transformed_plane}
	Consider the robot models \eqref{model_actual_plane} and \eqref{model_transformed_plane} in the actual and transformed planes, respectively. Under the M\"{o}bius transformation \eqref{mobius_transformation}, the sense of rotation of the $k^\text{th}$ agent's velocity vector in the two planes is the same if $\alpha = \alpha_s$, and it is opposite if $\alpha = \alpha_\ell$.
\end{lem}

We shall now formulate an equivalent problem in the next subsection where the non-concentric boundary constraints in Problem~\ref{problem_actual_plane} are mapped to uniform spatial constraints in the transformed plane keeping the same control objectives. This transformation facilitates controller design as we discuss in the upcoming discussion.    

\subsection{Equivalent Problem in Transformed Plane}
Let us now turn our focus to the transformed plane where both the circles $f(\mathcal{C})$ and $f(\mathcal{C}')$ are concentric as shown in Fig.~\ref{fig_mobius_mapping}. Since the desired circle $\mathcal{C}$ in original plane on which the agents are required to converge mapped to $f(\mathcal{C})$ in the transformed plane, we aim to stabilize agents \eqref{model_transformed_plane} on the circle $f(\mathcal{C})$ of radius $|\alpha|$. Following \eqref{error_actual_plane}, the position error in the transformed plane can be analogously written as: 
\begin{equation}\label{error_transformed_plane}
	\mathcal{E}_k = \rho_k \pm i |\alpha|\mathrm{e}^{i\gamma_k}=\rho_k+i\sigma \mathrm{e}^{i\gamma_k}, \ \forall k,
\end{equation}
where $\sigma \triangleq \pm |\alpha| \neq 0$, and $``+"$ and $``-"$ signs are analogously used as in \eqref{error_actual_plane}. Following Lemma~\ref{thm_velcoty_direction_transformed_plane}, this essentially implies that $\sigma = +|\alpha|$ if $\alpha = \alpha_s$, and $\sigma = -|\alpha|$ if $\alpha = \alpha_{\ell}$. In the transformed plane, our focus is to minimize error \eqref{error_transformed_plane}, which is now subject to uniform radial constraints between the two circles in the transformed plane and is given by (see Fig.~\ref{fig_mobius_mapping}):
\begin{equation}\label{delta_transformed}
	\delta_T = 
	\begin{cases}
		|(\lambda + \alpha)/\mu| - |\alpha|,  & \text{if } \alpha = \alpha_s\\
		|\alpha| - |(\lambda + \alpha)/\mu|,  & \text{if } \alpha = \alpha_{\ell}.
	\end{cases}	
\end{equation} 

Moreover, the phase-shifted order parameter \eqref{shifted_order_parameter} can be expressed solely in terms of the heading angles $\gamma_k$ in the transformed plane as: $q = (1/N)\sum_{k=1}^{N}{\rm e}^{i\gamma_k}$, using \eqref{gamma_k}. As a result, synchronization and balancing of $\theta_k + \chi_k$ in Problem~\ref{problem_actual_plane} is equivalent to synchronization and balancing of $\gamma_k$. Additionally, we also regulate the speed $s_k$ of the $k^\text{th}$ agent in the transformed plane to some constant value $s_d > 0$. As will be established later in Corollary~\ref{cor_positive_speeds}, $s_d$ is a design parameter and can be appropriately chosen to ensure that $s_k(t) > 0$ for all $k$ and $t \geq 0$, and so is $v_k(t) > 0$ for all $k$ and $t \geq 0$, provided $v_k(0) > 0$ for all $k$. Please note that the control over speed $s_k$ does not affect our objectives in Problem~\ref{problem_actual_plane} and is done for the simplicity of the controller design, as discussed in the next section. We assume that the interaction topology among agents is independent of the coordinate transformation and remains the same in both planes. With this, we now describe the equivalent problem in the transformed plane. 

\begin{problem}[Problem in transformed plane]\label{problem_transformed_plane}
	Consider the concentric circles $f(\mathcal{C}): |w| = |\alpha|$ and  $f(\mathcal{C}'): |w| = |(\lambda + \alpha)/\mu|$, obtained under the  M\"{o}bius transformation \eqref{mobius_transformation}, where $\lambda$ and $\mu$ are defined in Lemma~\ref{lem_mapping_of_circles} and $\alpha$ is the root of \eqref{mobius_roots}. Based on the choice of $\alpha$ as discussed in Lemma~\ref{lem_mapping_of_circles}, Problem~\ref{problem_actual_plane} is equivalent to the following two problems in the transformed plane:
	\begin{enumerate}[leftmargin=*]
		\item[\emph{(P1)}] Trajectory-Constraining Problem: If $\alpha = \alpha_s$. Design the controls $\nu_k$ and $\Omega_k$ for all $k$, such that the agents \eqref{model_transformed_plane} with initial positions $\rho_k(0) \in  \mathcal{W}_s \triangleq \{w_{\rho} \in \mathbb{C} \mid |w_{\rho}| < |(\lambda + \alpha_s)/\mu|\}, \forall k$, and interacting over an undirected and connected graph $\mathcal{G}$ with Laplacian $\mathcal{L}$, asymptotically converge on the circle $f_s(\mathcal{C})$ at constant speed $s_d$ in synchronization (resp., balancing) of phasors ${\rm e}^{i\gamma_k}$ with their trajectories remain bounded within the concentric circle $f_s(\mathcal{C}')$ at all times. That is, $\mathcal{E}_k(t) \to 0$ (equivalently, $|\rho_k(t)| \to |\alpha_s|$), $s_k \to s_d, \forall k$, with $\gamma_1 = \cdots = \gamma_N$ (resp., $\pmb{1}^\top_N \mathrm{e}^{i{\pmb{\gamma}}} = 0$), as $t \to \infty$, and $\rho_k(t) \in  \mathcal{W}_s, \forall k$ and $\forall t \geq 0$. 
		\item[\emph{(P2)}] Obstacle-Avoidance Problem: If $\alpha = \alpha_{\ell}$. Design the controls $\nu_k$ and $\Omega_k$ for all $k$, such that the agents \eqref{model_transformed_plane} with initial positions $\rho_k(0) \in  \mathcal{W}_{\ell} \triangleq \{w_{\rho} \in \mathbb{C} \mid |w_{\rho}| > |(\lambda + \alpha_{\ell})/\mu|\}, \forall k$, and interacting over an undirected and connected graph $\mathcal{G}$ with Laplacian $\mathcal{L}$, asymptotically converge on the circle $f_{\ell}(\mathcal{C})$ at constant speed $s_d$ in synchronization (resp., balancing) of phasors ${\rm e}^{i\gamma_k}$ with their trajectories remain exterior to the concentric circle $f_{\ell}(\mathcal{C}')$ at all times. That is, $\mathcal{E}_k(t) \to 0$ (equivalently, $|\rho_k(t)| \to |\alpha_{\ell}|$), $s_k \to s_d, \forall k$, with $\gamma_1 = \cdots = \gamma_N$ (resp., $\pmb{1}^\top_N \mathrm{e}^{i{\pmb{\gamma}}} = 0$), as $t \to \infty$, and $\rho_k(t) \in \mathcal{W}_{\ell}, \forall k$ and $\forall t \geq 0$. 
	\end{enumerate}
\end{problem}

A direct consequence of synchronization of agents around the common circular orbits results in an interesting fact that the synchronization of phasors ${\rm e}^{i\gamma_k}$ in the transformed plane is equivalent to the synchronization of phasors ${\rm e}^{i\theta_k}$ in the original plane, irrespective of M\"{o}bius-shift $\chi_k$ in \eqref{chi_argument_form}. We prove this result in the following lemma:

\begin{figure}[t!]
	\centering{\hspace*{-0.5cm}
	\includegraphics[width=8.5cm]{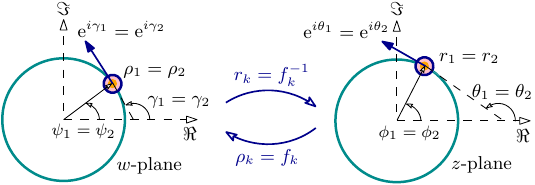}}
	\caption{Synchronization equivalence in both planes for $N=2$.}
	\label{fig_synchronization_equivalence}
\end{figure}

\begin{lem}\label{lem_synchronization_equivalence_both_planes}
Consider $N$ agents described by models \eqref{model_actual_plane} and \eqref{model_transformed_plane} in the original and transformed planes, respectively, and related via the M\"{o}bius transformation \eqref{mobius_transformation}. The synchronization of phasors ${\rm e}^{i\gamma_k}$ on the circle $f(\mathcal{C})$ in the transformed plane is equivalent to the synchronization of phasors ${\rm e}^{i\theta_k}$ on the desired circle $\mathcal{C}$ in the original plane for all $k$.
 \end{lem}

\begin{proof}
The proof is established using the logical chain: (synchronization in $w$-plane) $\implies$ (spatial synchronization) $\implies$ (synchronization in $z$-plane). Phase synchronization in $w$-plane implies that $\gamma_k(t) = \gamma^*(t), \forall k$ for some $\gamma^*(t) \in \mathbb{R}$. At the steady-state, when the agents evolve on $f(\mathcal{C})$, it follows from \eqref{error_transformed_plane} that $\mathcal{E}_k(t) = 0 \implies \rho_k(t) + i\sigma {\rm e}^{i\gamma_k(t)} = 0$. Since $\gamma_k(t) = \gamma^*(t), \forall k$ on $f(\mathcal{C})$, $\rho_k(t) = -i\sigma e^{i\gamma^*(t)} \triangleq \rho^*(t), \forall k$, which shows that all agents correspond to the same time-varying point $\rho^*(t)$ on $f(\mathcal{C})$, thereby establishing spatial synchronization in the $w$-plane. Next, since $f^{-1}$ is bijective, $\rho_k(t) \to \rho^*(t)$ implies, using \eqref{inverse_mobius_transformation}, that $r_k(t) = f^{-1}(\rho^*(t)) \triangleq r^*(t), \forall k$ in the $z$-plane. Moreover, since $\rho^*(t) \in f(\mathcal{C})$ and $f^{-1}$ maps $f(\mathcal{C})$ back to $\mathcal{C}$, it follows that $r^*(t) = f^{-1}(\rho^*(t)) \in \mathcal{C}$. Thus, all agents correspond to the same time-varying point $r^*(t)$ on $\mathcal{C}$, establishing spatial synchronization in the $z$-plane. Now, it remains to show that spatial coincidence implies synchronization of heading angles. Since the agents move on $\mathcal{C}$, their velocity vectors must be tangent to the circle. Consequently, the unit vectors ${\rm e}^{i\theta_k}$ and ${\rm e}^{i\phi_k}$, corresponding to the velocity and position of the $k^\text{th}$ agent on $\mathcal{C}$, must be perpendicular (see Fig.~\ref{fig_synchronization_equivalence} for illustration). This implies ${\rm e}^{i\theta_k} = {\rm e}^{i(\phi_k \pm \pi/2)}$, where the sign $\pm$ corresponds to the direction of motion. For the common spatial position $r^*(t) = {\rm e}^{i\phi^*(t)}$ on $\mathcal{C}$ for some $\phi^*(t) \in \mathbb{R}$, it follows that ${\rm e}^{i\theta_k(t)} = {\rm e}^{i(\phi^*(t) \pm \pi/2)}, \forall k$. Since all agents share the same spatial position and the same rotation convention (determined by the sign of $\sigma$ in \eqref{error_transformed_plane}), the tangent direction is identical for all agents. Therefore, $\theta_k = \theta_0 \mod(2\pi), \forall k$ for some $\theta_0 \in\mathbb{R}$. This completes the proof. 
\end{proof}  

As a consequence of Lemma~\ref{lem_synchronization_equivalence_both_planes}, it follows that $\chi_k = \chi_0$ and $\zeta_k = \zeta_0$ for some $\chi_0, \zeta_0 \in \mathbb{R}$ for all $k$, under synchronization (see Lemma~\ref{lem_chi_zeta} in the Appendix). It is important to note that, unlike synchronization, \emph{balancing} in the original plane generally corresponds to a spatially non-uniform distribution of agents. This arises due to the radial distortion induced by the M\"{o}bius transformation, which alters inter-agent spacing during the mapping between the two planes, even though the phase configuration remains balanced in the transformed domain. Next, we aim to solve Problem~\ref{problem_transformed_plane} by designing the controllers $\nu_k$ and $\Omega_k$ in the transformed plane, followed by control design in the original plane to solve actual Problem~\ref{problem_actual_plane}. The next section focuses on the former control design.

\section{Control Design in Transformed Plane}\label{section_5_control_transformed_plane}
In this section, we begin by introducing the potential functions corresponding to the individual control objectives in Problem~\ref{problem_transformed_plane}, and then proceed to design the controllers using a composite Lyapunov function framework.

\subsection{Control Objective-based Potential Functions}

\subsubsection{Achieving synchronization and balancing in transformed plane}
Let $\pmb{\gamma} = [\gamma_1, \ldots, \gamma_N]^\top \in \mathbb{R}^N$ be the heading vector in the transformed plane and denote $\mathrm{e}^{i\pmb{\gamma}} = [\mathrm{e}^{i\gamma_1}, \ldots, \mathrm{e}^{i\gamma_N}]^\top \in \mathbb{T}^N$. To achieve synchronization and balancing of phasors ${\rm e}^{i\gamma_k}$, we consider the following quadratic potential function: 
\begin{equation}\label{potential_synchro_balanc}
	\mathcal{U}(\pmb{\gamma})=(1/2)\langle \mathrm{e}^{i\pmb{\gamma}}, \mathcal{L}\mathrm{e}^{i\pmb{\gamma}}  \rangle,
\end{equation}
relying on the Laplacian $\mathcal{L}$ of the interaction graph $\mathcal{G}$ among agents in \eqref{model_transformed_plane}. Note that if $\mathrm{e}^{i\pmb{\gamma}} = \mathrm{e}^{i\gamma_0}\pmb{1}_N$ for some $\gamma_0 \in \mathbb{R}$, implying that $\mathcal{L}\mathrm{e}^{i\pmb{\gamma}} = \mathrm{e}^{i\gamma_0}\mathcal{L}\pmb{1}_N = 0$ for an undirected and connected graph $\mathcal{G}$. In other words, \eqref{potential_synchro_balanc} minimizes if phasors $\mathrm{e}^{i\gamma_k}$ synchronizes. Further, if the graph $\mathcal{G}$ is circulant, it follows from Lemma~\ref{lem_circulant_graphs} (in the Appendix) that (i) the Laplacian $\mathcal{L}$ has orthogonal basis vectors of the exponential form, (ii) the potential \eqref{potential_synchro_balanc} can be expressed as  $\mathcal{U}(\pmb{\gamma}) = (1/2)\langle \mathcal{F}^\ast {\rm e}^{i\pmb{\gamma}}, \Lambda \mathcal{F}^\ast {\rm e}^{i\pmb{\gamma}}\rangle$, which on substituting $\pmb{w} = \mathcal{F}^\ast {\rm e}^{i\pmb{\gamma}}$, implies $\mathcal{U}(\pmb{\gamma}) = (1/2)\langle \pmb{w}, \Lambda\pmb{w}\rangle = (1/2)\sum_{k=2}^{N}|w_k|^2\lambda_k$. Since $\mathcal{F}$ is unitary, $\|\pmb{w}\| = \|{\rm e}^{i\pmb{\gamma}}\| = \sqrt{N}$. Therefore, $ \mathcal{U}(\pmb{\gamma}) = (1/2)\langle \pmb{w}, \Lambda\pmb{w}\rangle \leq (N/2) \lambda_{\text{max}}$, where, $\lambda_{\text{max}}$ is the maximum eigenvalue of $\mathcal{L}$. In other words, $\mathcal{U}(\pmb{\gamma})$ is bounded by $(N/2) \lambda_{\text{max}}$ for a circulant graph $\mathcal{G}$, and the maximum value is achieved by selecting ${\rm e}^{i\pmb{\gamma}}$ as the eigenvector of $\mathcal{L}$, associated with $\lambda_{\max}$. According to point (i), since ${\rm e}^{i\pmb{\gamma}}$ is orthogonal to $\pmb{1}_N$, i.e., $\pmb{1}_N^\top{\rm e}^{i\pmb{\gamma}} = 0$, this corresponds to the balancing of ${\rm e}^{i\pmb{\gamma}}$. The time-derivative of \eqref{potential_synchro_balanc}, along \eqref{model_transformed_plane}, is given by
\begin{equation}\label{pf_sb1}
	\dot{\mathcal{U}} = \sum\nolimits_{k = 1}^{N}(\partial\mathcal{U}/\partial\gamma_k)\dot{\gamma}_k = \sum\nolimits_{k=1}^{N} (\partial\mathcal{U}/\partial\gamma_k)\Omega_k,
\end{equation}
where the gradient of $\mathcal{U}$ is obtained as:
\begin{equation}\label{gradient}
	\partial\mathcal{U}/\partial\gamma_k = \langle  i\mathrm{e}^{i\gamma_k}, \mathcal{L}_k \mathrm{e}^{i\pmb{\gamma}}\rangle = -\sum\nolimits_{j\in \mathcal{N}_k} \sin (\gamma_j-\gamma_k),
\end{equation}
with $\mathcal{L}_k$ being the $k^\text{th}$ row of the Laplacian $\mathcal{L}$. Substituting \eqref{gradient} into \eqref{pf_sb1}, we get
\begin{equation} \label{potential_synchronization_balancing_time_derivative}
	\dot{\mathcal{U}} = \sum\nolimits_{k=1}^{N}\langle  i\mathrm{e}^{i\gamma_k}, \mathcal{L}_k\mathrm{e}^{i\pmb{\gamma}}\rangle\Omega_k.
\end{equation}
It is worth noting that the critical points of \eqref{potential_synchro_balanc} (i.e., where \eqref{gradient} equates to zero) characterize the synchronization and balancing of $\pmb{\gamma}$ under certain conditions on the graph $\mathcal{G}$; please refer to Lemma~\ref{lem_critical_points} in the Appendix. For an undirected and connected graph $\mathcal{G}$, since $\mathcal{L} = \mathcal{B}\mathcal{B}^\top$ (where $\mathcal{B}$ is the incidence matrix of $\mathcal{G}$), we can write $\langle \mathrm{e}^{i\pmb{\gamma}}, \mathcal{B}\mathcal{B}^\top \mathrm{e}^{i\pmb{\gamma}} \rangle = \langle  \mathcal{B}^\top\mathrm{e}^{i\pmb{\gamma}}, \mathcal{B}^\top\mathrm{e}^{i\pmb{\gamma}}\rangle$. Using this, one can obtain from \eqref{gradient} that
\begin{align}\label{ortho_prop}
\sum_{k=1}^{N}\frac{\partial\mathcal{U}}{\partial\gamma_k}  =\langle i\mathrm{e}^{i\pmb{\gamma}}, \mathcal{B}\mathcal{B}^\top \mathrm{e}^{i\pmb{\gamma}} \rangle = \langle i \mathcal{B}^\top \mathrm{e}^{i\pmb{\gamma}}, \mathcal{B}^\top \mathrm{e}^{i\pmb{\gamma}}\rangle = 0,
\end{align}
implying that $\langle \nabla_{\pmb{\gamma}}\mathcal{U}, \pmb{1}_N \rangle = 0$, i.e., $\nabla_{\pmb{\gamma}}\mathcal{U} \perp \pmb{1}_N$. For an undirected and connected graph $\mathcal{G}$, we can also derive the upper bound on \eqref{potential_synchro_balanc} as follows: $0.5\langle \mathrm{e}^{i\pmb{\gamma}}, \mathcal{L}\mathrm{e}^{i\pmb{\gamma}}  \rangle = 0.5\langle \mathrm{e}^{i\pmb{\gamma}}, \mathcal{B}\mathcal{B}^\top \mathrm{e}^{i\pmb{\gamma}} \rangle = 0.5\langle \mathcal{B}^\top \mathrm{e}^{i\pmb{\gamma}}, \mathcal{B}^\top \mathrm{e}^{i\pmb{\gamma}} \rangle = 0.5\sum_{\{j, k\} \in \mathbb{E}} | \mathrm{e}^{i\gamma_j} - \mathrm{e}^{i\gamma_k}  |^2 {\leq} 0.5\sum_{\{j, k\} \in \mathbb{E}} (|\mathrm{e}^{i\gamma_j}| + |\mathrm{e}^{i\gamma_k}|)^2 = 2|\mathbb{E}|$, where $|\mathbb{E}|$ is the cardinality of the edge set $\mathbb{E}$.

\subsubsection{Constrained circular motion in transformed plane}
To stabilize the circular motion while enforcing uniform boundary constraints in the transformed plane, we minimize error \eqref{error_transformed_plane} by employing the following barrier Lyapunov function (BLF):
\begin{equation}\label{blf}
\mathcal{S}(\pmb{\mathcal{E}}) = \mathcal{S}(\pmb{r}, \pmb{\gamma}) =  \frac{1}{2}\sum_{k=1}^{N}\ln\left[ \frac{\delta_T^2}{\delta_T^2-|\mathcal{E}_k|^2} \right],
\end{equation}
where ``$\ln$" denotes the natural logarithm and $\pmb{\mathcal{E}} = [\mathcal{E}_1,\ldots,\mathcal{E}_N]^\top \in \mathbb{C}^N$ is the error vector. Note that \eqref{blf} is positive definite and $\mathcal{C}^1$ in the domain $|\mathcal{E}_k| < \delta_T, \forall k$ \cite{tee2009barrier}, and vanishes iff $\pmb{\mathcal{E}} = \pmb{0}_N$. The time-derivative of \eqref{blf}, along \eqref{model_transformed_plane}, is $\dot{\mathcal{S}} = (1/2) \sum_{k=1}^{N} (d/dt)|\mathcal{E}_k|^2/(\delta_T^2-|\mathcal{E}_k|^2)$, where $(1/2)(d/dt)|\mathcal{E}_k|^2 = \langle \mathcal{E}_k,\dot{ \mathcal{E}_k}  \rangle$. Differentiating \eqref{error_transformed_plane}, yields
\begin{equation}\label{error_transformed_plane_derivative}	
\dot{\mathcal{E}}_k = \dot{\rho}_k - \sigma{\rm e}^{i\gamma_k}\dot{\gamma}_k = {\rm e}^{i\gamma_k}(s_k - \sigma\Omega_k),
\end{equation}
substituting which, we have $(1/2)(d/dt)|\mathcal{E}_k|^2  = \langle \rho_k + i\sigma{\mathrm{e}}^{i\gamma_k},  {\mathrm{e}}^{i\gamma_k} \rangle (s_k - \sigma\Omega_k) = [\langle \rho_k,  {\mathrm{e}}^{i\gamma_k} \rangle + \sigma\langle i {\mathrm{e}}^{i\gamma_k},  {\mathrm{e}}^{i\gamma_k} \rangle](s_k - \sigma\Omega_k)$, using linearity of inner product. Since $\langle i {\mathrm{e}}^{i\gamma_k}, {\mathrm{e}}^{i\gamma_k} \rangle = 0$, we have $(1/2)(d/dt)|\mathcal{E}_k|^2 = \langle \rho_k,  {\mathrm{e}}^{i\gamma_k} \rangle (s_k - \sigma\Omega_k)$. Substituting this into the preceding expression, yields
\begin{equation}\label{blf_time_derivative}
\dot{\mathcal{S}} =  \sum_{k=1}^{N} \frac{\langle \rho_k,\mathrm{e}^{i\gamma_k} \rangle}{\delta_T^2-|\mathcal{E}_k|^2} (s_k - \sigma\Omega_k).
\end{equation}

\subsubsection{Achieving desired linear speed in transformed plane}
To fulfill the requirement that $s_k \to s_d$ for all $k$ at the steady-sate as in Problem~\ref{problem_transformed_plane}, we define the speed error as:
\begin{equation}\label{speed_error}
	\tilde{s}_k \triangleq s_k - s_d.
\end{equation}
and consider the following BLF for all $k$:
\begin{equation}\label{blf_speed}
	\mathcal{H}({\pmb{s}}) = \frac{1}{2} \sum_{k=1}^{N} \ln \left[\frac{\delta_S^2}{\delta_S^2 - \tilde{s}_k^2}\right],
\end{equation}
where $\delta_S > 0$ is a constant design parameter and $\pmb{s} = [s_1, \cdots, s_N]^\top \in \mathbb{R}^N$ is the speed vector. Analogous to \eqref{blf}, the function \eqref{blf_speed} is positive definite and $\mathcal{C}^1$ in the domain $|\tilde{s}_k| < \delta_S, \forall k$, and becomes zero iff $\pmb{s} = s_d\pmb{1}_N$, as desired. The time-derivative of \eqref{blf_speed}, along model \eqref{model_transformed_plane}, is given by 
\begin{equation}\label{H_time_derivative}
	\dot{\mathcal{H}} = \sum_{k=1}^{N}  \frac{\tilde{s}_k\dot{s}_k}{\delta_S^2 - \tilde{s}_k^2} = \sum_{k=1}^{N}  \frac{\tilde{s}_k}{\delta_S^2 - \tilde{s}_k^2}\nu_k.
\end{equation}
The expressions \eqref{potential_synchronization_balancing_time_derivative}, \eqref{blf_time_derivative} and \eqref{H_time_derivative} will be used in the next section where we discuss the control design. 

\subsection{Controller Design}
Based on the potentials \eqref{potential_synchro_balanc}, \eqref{blf}, and \eqref{blf_speed}, we now state the following theorem, which proposes $\nu_k$ and $\Omega_k$.

\begin{thm}\label{thm_control_transformed_plane}
	Consider $N$ agents described by the model \eqref{model_transformed_plane} and interacting over an undirected and connected graph $\mathcal{G}$ with Laplacian matrix $\mathcal{L}$. Assume that agents' initial states $[\rho_k(0), s_k(0), \gamma_k(0)]^\top$ belong to the set $\mathcal{W}_{\delta} \triangleq \{[\pmb{\rho}, \pmb{s}, \pmb{\gamma}]^\top \in \mathbb{C}^N \times \mathbb{R}^N \times \mathbb{T}^N \mid |\mathcal{E}_k| < \delta_T, |\tilde{s}_k| < \delta_S, \forall k \}$, where $\delta_T > 0$ and $\delta_S > 0$ are defined in \eqref{delta_transformed} and \eqref{blf_speed}, respectively. Let the agents be governed by the controllers
	\begin{align}
		\label{velocity_control_transformed_plane}	\nu_k & =  -(\delta_S^2 - \tilde{s}_k^2) \left[\kappa_1 \frac{\langle \rho_k,\mathrm{e}^{i\gamma_k}  \rangle}{\delta_T^2-|\mathcal{E}_k|^2} + \kappa_2\tilde{s}_k \right],\\
		\label{curvature_control_transformed_plane}	\Omega_k & = \frac{1}{\sigma}\left[s_d + \kappa_1 \frac{\langle \rho_k,\mathrm{e}^{i\gamma_k} \rangle}{\delta_T^2-|\mathcal{E}_k|^2} + \frac{\mathcal{K}}{\sigma} \langle  i\mathrm{e}^{i\gamma_k}, \mathcal{L}_k\mathrm{e}^{i\pmb{\gamma}}\rangle \right].
	\end{align}
	Then, the following properties hold:
	\begin{itemize}
		\item[(a)] If $\kappa_1 > 0$, $\kappa_2 > 0$, and $\mathcal{K} < 0$, all the agents converge to the circle $f(\mathcal{C})$, and the phasors ${\rm e}^{i\pmb{\gamma}}$ synchronize within $\mathcal{W}_{\delta}$.
	    \item[(b)] Additionally, if $\mathcal{G}$ is circulant and $\kappa_1 > 0$, $\kappa_2 > 0$, and $\mathcal{K} > 0$, all the agents converge to the circle $f(\mathcal{C})$, and the phasors ${\rm e}^{i\pmb{\gamma}}$ achieve balancing within $\mathcal{W}_{\delta}$. 
		\item[(c)] In both cases above, the agents' speeds remain bounded as $s_k(t) \in (s_d - \delta_S, s_d + \delta_S)$ for all $k$ and $t \geq 0$. Moreover, depending on the roots of \eqref{mobius_roots}, the agents' trajectories remain bounded within regions given by
		\begin{equation}\label{rho_bounds}
		\hspace*{-0.5cm}|\rho_k(t)| \in  
		\begin{dcases*}
			\left(2|\alpha| - |(\lambda + \alpha)/\mu|,  |(\lambda + \alpha)/\mu|\right), & if $\alpha = \alpha_s$ \\
			\left(|(\lambda + \alpha)/\mu|,  2|\alpha| - |(\lambda + \alpha)/\mu|\right), & if $\alpha = \alpha_{\ell}$, 
		\end{dcases*}
		\end{equation}
		for all $k$ and $t \geq 0$.
		\end{itemize}
\end{thm}

\begin{proof}
(a) Consider the following composite potential function
\begin{equation}\label{V_s}
	\mathcal{V}_s(\pmb{\rho}, \pmb{s}, \pmb{\gamma}) =  \kappa_1\mathcal{S}(\pmb{\mathcal{E}}) - \mathcal{K}\mathcal{U}(\pmb{\gamma}) + \mathcal{H}{(\pmb{s})}; \ \kappa_1 > 0, \mathcal{K} < 0,
\end{equation}
which is positive-definite and bounded from below by zero, and hence, is a valid Lyapunov function candidate in the set $\mathcal{W}_{\delta}$. The time-derivative of \eqref{V_s}, along \eqref{model_transformed_plane}, is $\dot{\mathcal{V}}_s = \kappa_1\dot{\mathcal{S}} - \sigma\mathcal{K}{\dot{\mathcal{U}}} + \dot{\mathcal{H}}$. Using \eqref{potential_synchronization_balancing_time_derivative}, \eqref{blf_time_derivative}, and  \eqref{H_time_derivative}, we have $\dot{\mathcal{V}}_s = \sum_{k=1}^{N} \left[\kappa_1 \frac{\langle \rho_k,\mathrm{e}^{i\gamma_k} \rangle}{\delta_T^2-|\mathcal{E}_k|^2} (s_k-\sigma\Omega_k) - \mathcal{K}\langle  i\mathrm{e}^{i\gamma_k}, \mathcal{L}_k\mathrm{e}^{i\pmb{\gamma}}\rangle\Omega_k \right] + \sum_{k=1}^{N} \frac{\tilde{s}_k}{\delta_S^2 - \tilde{s}_k^2}\nu_k$. 
Further, substituting $\tilde{s}_k$ and $\nu_k$ from \eqref{speed_error} and \eqref{velocity_control_transformed_plane}, respectively, yields
\begin{align*}
	\dot{\mathcal{V}}_s & = \sum_{k=1}^{N} \left[ \kappa_1 \frac{\langle \rho_k,\mathrm{e}^{i\gamma_k} \rangle}{\delta_T^2-|\mathcal{E}_k|^2} (s_k-\sigma\Omega_k) - \mathcal{K}\langle  i\mathrm{e}^{i\gamma_k}, \mathcal{L}_k\mathrm{e}^{i\pmb{\gamma}}\rangle\Omega_k \right]\\
	& \qquad + \sum_{k=1}^{N} (s_k - s_d) \left[-\kappa_1 \frac{\langle \rho_k,\mathrm{e}^{i\gamma_k}  \rangle}{\delta_T^2-|\mathcal{E}_k|^2} - \kappa_2(s_k - s_d) \right]\\ 
	& = \sum_{k=1}^{N} \left[ \kappa_1 \frac{\langle \rho_k,\mathrm{e}^{i\gamma_k} \rangle}{\delta_T^2-|\mathcal{E}_k|^2} \left(s_d - \sigma\Omega_k\right) - \frac{\mathcal{K}}{\sigma}\langle  i\mathrm{e}^{i\gamma_k}, \mathcal{L}_k\mathrm{e}^{i\pmb{\gamma}}\rangle (\sigma\Omega_k) \right]\\
	& \qquad - \sum_{k=1}^{N} \kappa_2\tilde{s}_k^2.
\end{align*}
Using the orthogonality \eqref{ortho_prop} and knowing that $\sigma, s_d$ and $\mathcal{K}$ are constants, one can observe that $s_d(\mathcal{K}/\sigma) \sum_{k=1}^{N} \langle  i\mathrm{e}^{i\gamma_k}, \mathcal{L}_k\mathrm{e}^{i\pmb{\gamma}}\rangle = 0$. This term can be added to the above expression to yield \begin{align*}
	\nonumber\dot{\mathcal{V}}_s & {=} \sum_{k=1}^{N} \left[ \kappa_1 \frac{\langle \rho_k,\mathrm{e}^{i\gamma_k} \rangle}{\delta_T^2-|\mathcal{E}_k|^2} {+}	\frac{\mathcal{K}}{\sigma}\langle  i\mathrm{e}^{i\gamma_k}, \mathcal{L}_k\mathrm{e}^{i\pmb{\gamma}}\rangle \right]\left(s_d - \sigma\Omega_k\right)\\
	& \qquad - \sum_{k=1}^{N} \kappa_2\tilde{s}_k^2.
\end{align*} 
Now, substituting $\Omega_k$ from \eqref{curvature_control_transformed_plane}, we have
\begin{equation}\label{Vs_derivative_final}
\dot{\mathcal{V}}_s = -\sum_{k=1}^{N} \left[ \left\{\kappa_1 \frac{\langle \rho_k,\mathrm{e}^{i\gamma_k} \rangle}{\delta_T^2-|\mathcal{E}_k|^2} + \frac{\mathcal{K}}{\sigma}\langle  i\mathrm{e}^{i\gamma_k}, \mathcal{L}_k\mathrm{e}^{i\pmb{\gamma}}\rangle\right\}^2 + \kappa_2\tilde{s}_k^2\right],
\end{equation}
implying that $\dot{\mathcal{V}}_s$ is negative semi-definite, and hence, $\mathcal{V}_s(\pmb{\rho}(t), \pmb{s}(t), \pmb{\gamma}(t)) \leq \mathcal{V}_s(\pmb{\rho}(0),\pmb{s}(0),\pmb{\gamma}(0))$ for all $t \geq 0$, in the set $\mathcal{W}_{\delta}$. Since $\mathcal{U}(\pmb{\gamma}) \leq 2|\mathcal{E}|$, and so is $\mathcal{U}(\pmb{\gamma}(0)) \leq 2|\mathcal{E}|$. From \eqref{blf}, it can be observed that $\mathcal{S}(\pmb{\rho}(0),\pmb{\gamma}(0))$ is finite and positive for any initial condition in the set $|\mathcal{E}_k| < \delta_T, \forall k$, and it is undefined only if $|\mathcal{E}_k| = \delta_T$ for any $k$.  Also, it follows from \eqref{blf_speed} that $\mathcal{H}(\pmb{s}(0))$ is finite and positive for any initial condition in the set $|\tilde{s}_k| < \delta_S, \forall k$, and it is undefined only if $|\tilde{s}_k| = \delta_S$ for any $k$. Consequently, it follows from \eqref{V_s} that $\mathcal{V}_s(\pmb{\rho}(0),\pmb{s}(0),\pmb{\gamma}(0))$ is finite and positive for any initial condition in $\mathcal{W}_{\delta}$, since $\kappa_1, \mathcal{K}$ and $\sigma$ are finite. In other words, this implies that there exists a positive constant $\kappa$ such that $\mathcal{V}_s(\pmb{\rho}, \pmb{s}, \pmb{\gamma}) \leq \kappa$ for $t \geq 0$ in $\mathcal{W}_{\delta}$. 

Note that the set  $\Gamma \triangleq \{[\pmb{\rho}, \pmb{s}, \pmb{\gamma}]^\top \in \mathcal{W}_{\delta} \mid \mathcal{V}_s(\pmb{\rho},\pmb{s},\pmb{\gamma}) \leq \kappa \} \subset \mathcal{W}_{\delta}$ is compact and positively invariant, since $\mathcal{V}_s(\pmb{\rho}, \pmb{s}, \pmb{\gamma})$ is positive definite and continuously differentiable, and $\dot{\mathcal{V}}_s \leq 0$, along the solutions of \eqref{model_transformed_plane}, in $\mathcal{W}_{\delta}$. From LaSalle's invariance principle \cite[Theorem 4.4, pg. 128]{khalil2002nonlinear}, it follows that solutions of \eqref{model_transformed_plane}, under controllers \eqref{velocity_control_transformed_plane} and \eqref{curvature_control_transformed_plane}, converge to the largest invariant set $\Delta_s$ contained in the set $\Delta \subset \Gamma$, where $\dot{\mathcal{V}}_s = 0$. That is,
\begin{equation}
	\Delta \triangleq \{[\pmb{\rho}, \pmb{s}, \pmb{\gamma}]^\top \in \mathcal{W}_{\delta} \mid \dot{\mathcal{V}}_s = 0\},
\end{equation}
which, using \eqref{Vs_derivative_final}, implies that
\begin{subequations}
	\begin{align}
		\label{inv_seta} \kappa_1 \frac{\langle \rho_k,\mathrm{e}^{i\gamma_k} \rangle}{\delta_T^2-|\mathcal{E}_k|^2} + \frac{\mathcal{K}}{\sigma} \langle i\mathrm{e}^{i\gamma_k},  \mathcal{L}_k\mathrm{e}^{i\pmb{\gamma}}\rangle & =0 \\
		\label{inv_setb} \tilde{s}_k &= 0,
	\end{align}
\end{subequations}
for all $k$ in $\Delta$. From \eqref{speed_error}, $\tilde{s}_k = 0 \implies s_k  = s_d, \forall k$, and hence, $\nu_k = \dot{s}_k = 0, \forall k$ in $\Delta$. Substituting these into \eqref{velocity_control_transformed_plane}, yield  
\begin{equation}\label{inv_setc}
	-\kappa_1\delta_S^2[\langle \rho_k,\mathrm{e}^{i\gamma_k} \rangle/(\delta_T^2-|\mathcal{E}_k|^2)] = 0, \forall k,	
\end{equation}
in $\Delta$. From \eqref{inv_setc}, $\langle \rho_k,\mathrm{e}^{i\gamma_k} \rangle = 0, \forall k$ in $\Delta$, implying that $\rho_k \perp \mathrm{e}^{i\gamma_k}, \forall k$ in $\Delta$. Also, it follows by substituting  \eqref{inv_seta} into \eqref{curvature_control_transformed_plane} that $\Omega_k = s_d/\sigma, \forall k$ in $\Delta$. Therefore, it can be concluded that the agents \eqref{model_transformed_plane} converge to the circle $f(\mathcal{C})$ and move with constant linear and angular speeds $s_d$ and $s_d/\sigma$ for all $k$, respectively. 

Further, upon substituting \eqref{inv_setc} into \eqref{inv_seta}, we have that $\langle  i\mathrm{e}^{i\gamma_k},  \mathcal{L}_k\mathrm{e}^{i\pmb{\gamma}}\rangle = 0, \forall k$ in $\Delta$. Exploiting \eqref{gradient}, this implies that $\pmb{\gamma}$ is the critical point of $\mathcal{U}(\pmb{\gamma})$ in $\Delta$. According to Lemma~\ref{lem_critical_points} in Appendix, this implies that $\mathrm{e}^{i\pmb{\gamma}}$ is the eigenvector of $\mathcal{L}$, i.e., $\mathcal{L}\mathrm{e}^{i\pmb{\gamma}} = \lambda \mathrm{e}^{i\pmb{\gamma}}$, using which, $\langle  i\mathrm{e}^{i\gamma_k},  \mathcal{L}_k\mathrm{e}^{i\pmb{\gamma}}\rangle=\lambda \langle  i\mathrm{e}^{i\gamma_k},  \mathrm{e}^{i\gamma_k}\rangle = 0$, where $\lambda \in \mathbb{R}$ is the eigenvalue of $\mathcal{L}$. For an undirected and connected graph $\mathcal{G}$, since $\pmb{1}_N$ spans the kernel of $\mathcal{L}$ corresponding to $\lambda = 0$, ${\rm e}^{i\pmb{\gamma}} = {\rm e}^{i\gamma_0}\pmb{1}_N$ for some $\gamma_0 \in \mathbb{R}$ (i.e., synchronization). While the other eigenvectors ${\rm e}^{i\pmb{\gamma}}$ satisfy $\pmb{1}^\top_N {\rm e}^{i\pmb{\gamma}} = 0$ (i.e., balancing), as they are mutually orthogonal for symmetric $\mathcal{L}$. For example, if $N = 2$, \eqref{gradient} is zero for $\gamma_j = \gamma_k$ or $\gamma_j = \pi \pm \gamma_k$. However, for $\gamma_j \neq \gamma_k$ (in general for $\pmb{1}^\top_N {\rm e}^{i\pmb{\gamma}} = 0$), it is easy to check that the potential $\mathcal{U}(\pmb{\gamma})$ is non-zero, and hence, $\mathcal{V}_s(\pmb{\rho}, \pmb{s}, \pmb{\gamma})$ attains a constant value. However, since $\dot{\mathcal{V}}_s \leq 0$ in $\mathcal{W}_{\delta}$, the solution trajectories of \eqref{model_transformed_plane} cannot stay in any local minimum of $\mathcal{U}(\pmb{\gamma})$, and approach to the largest invariant set $\Delta_s \subset \Delta$ as $t \to \infty$, that is, the phasors ${\rm e}^{i\pmb{\gamma}}$ synchronize at the steady-state. \par

(b) To prove the second statement, we consider the following joint potential function:
\begin{equation}\label{V_b}
	\mathcal{V}_b(\pmb{\rho},\pmb{s},\pmb{\gamma}) =  \kappa_1\mathcal{S}(\pmb{\mathcal{E}}) + \mathcal{K}\left[\frac{N\lambda_{\max}}{2} - \mathcal{U}(\pmb{\gamma}) \right] + \mathcal{H}(\pmb{s}),
\end{equation}
where $\kappa_1 > 0$ and $\mathcal{K} > 0$. Note that \eqref{V_b} is positive and bounded from below by zero for the circulant graph $\mathcal{G}$ (see Lemma~\ref{lem_critical_points} in Appendix). It is easy to check that the time-derivative of \eqref{V_b}, along model \eqref{model_transformed_plane}, and under the controls \eqref{velocity_control_transformed_plane} and \eqref{curvature_control_transformed_plane}, results in $\dot{\mathcal{V}}_b = \dot{\mathcal{V}}_s \leq 0$. Similar to the above analysis, the condition $\dot{\mathcal{V}}_b = 0$ implies that $\lambda \langle  i\mathrm{e}^{i\gamma_k},  \mathrm{e}^{i\gamma_k}\rangle = 0$. When $\lambda = 0$, it follows that ${\rm e}^{i\pmb{\gamma}} = {\rm e}^{i\gamma_0}\pmb{1}_N$ for some $\gamma_0 \in \mathbb{R}$, which corresponds to synchronization. However, for such configurations, the term $(N\lambda_{\max}/2) - \mathcal{U}(\pmb{\gamma})$ in \eqref{V_b} is strictly positive, and attains a constant value $N\lambda_{\max}/2$. Since $\dot{\mathcal{V}}_b \leq 0$, these configurations correspond to unstable equilibria within $\Delta$. On the other hand, for $\lambda \neq 0$, the corresponding eigenvectors of a circulant graph Laplacian satisfy $\pmb{1}^\top_N {\rm e}^{i\pmb{\gamma}} = 0$ (see Lemma~\ref{lem_circulant_graphs}), which corresponds to balancing according to Definition~\ref{def_mobius_synchronization_balancing}. In this case, $(N\lambda_{\max}/2) - \mathcal{U}(\pmb{\gamma})$ attains its minimum value of zero. Therefore, the trajectories converge to the balancing manifold within the largest invariant set $\Delta_b \subset \Delta$ as $t \to \infty$ and the agents move on the circle $f(\mathcal{C})$ within the set $\mathcal{W}_{\delta}$, as proved in the part (a). \par 
  
(c) Since $\dot{\mathcal{V}}_s = \dot{\mathcal{V}}_b \leq 0$ in $\mathcal{W}_{\delta}$, it immediately follows from Lemma~\ref{lem_blf_convergence} (in the Appendix) that $|\tilde{s}_k(t)| < \delta_S$ and $|\mathcal{E}_k(t)| < \delta_T$ for all $k$ and $t \geq 0$. Now, substituting for $\tilde{s}_k$ from \eqref{speed_error} and $\mathcal{E}_k$ from \eqref{error_transformed_plane}, along with $\delta_T$ from \eqref{delta_transformed}, and simplifying the modulus, the required bounds on $s_k(t)$ and $\rho_k(t)$ follows. For more details about the latter proof, please refer to \cite{singh2025mobius}.  
\end{proof}

\begin{remark}
From the proof of Theorem~\ref{thm_control_transformed_plane}, it follows that the invariant sets $\Delta_s$ and $\Delta_b$ satisfy $\Delta_s \cap \Delta_b = \emptyset$ and $\Delta_s \cup \Delta_b = \Delta$. Unlike \cite{hegde2023synchronization}, the analysis of the invariant set $\Delta$ in this paper is fundamentally different due to the higher dimensionality of the agent model \eqref{model_transformed_plane} and the stabilizing control laws employed. Furthermore, as established in Theorem~\ref{thm_control_transformed_plane}, since $|\mathcal{E}_k(t)| < \delta_T$ and $|\tilde{s}_k(t)| < \delta_S$ for all $k$ and $t \geq 0$, the controllers \eqref{velocity_control_transformed_plane} and \eqref{curvature_control_transformed_plane} remain bounded along any solution trajectory of \eqref{model_transformed_plane}.
\end{remark}

Please note that $s_d$ and $\delta_S$ are two design parameters in the proposed controllers \eqref{velocity_control_transformed_plane} and \eqref{curvature_control_transformed_plane}. By choosing these parameters appropriately, the prior assumption about the positivity of the speeds $s_k$ and $v_k$ in both the planes can be guaranteed. This is established in the below corollary to Theorem~\ref{thm_control_transformed_plane}. 

\begin{cor}\label{cor_positive_speeds}
Under the conditions specified in Theorem~\ref{thm_control_transformed_plane}, if the design parameters $s_d$ and $\delta_S$ are chosen such that $s_d \geq \delta_S$, the controllers \eqref{velocity_control_transformed_plane} and \eqref{curvature_control_transformed_plane} ensure that the agents move with positive speeds in both the planes, that is, $s_k(t) > 0$ and $v_k(t) > 0$ for all $k$ and $t \geq 0$, provided $v_k(0) > 0$ for all $k$.
\end{cor}

\begin{proof}
According to Theorem~\ref{thm_control_transformed_plane}, since $s_d - \delta_S < s_k(t) < s_d + \delta_S$ for all $k$ and $t \geq 0$ under controllers \eqref{velocity_control_transformed_plane} and \eqref{curvature_control_transformed_plane}, it holds that $0 < s_k(t) < s_d + \delta_S$ for all $k$ and $t \geq 0$, if $s_d \geq \delta_S$. Moreover, since both $s_k \in \mathcal{C}^1$ and $v_k \in \mathcal{C}^1$ for all $k$, it readily follows either from the relations \eqref{s_k} or \eqref{v_k} that, if $v_k(0) > 0$, $v_k(t) > 0$ for all $k$ and $t > 0$, since $|df_k/dr_k| \neq 0$ and $|df^{-1}_k/d \rho_k| \neq 0$ according to Lemma~\ref{lem_transformation_derivative_properties} (i.e., $v_k(t)$ can never cross zero and remains positive if it initially positive).
\end{proof}

We shall now show how the controllers \eqref{velocity_control_transformed_plane} and \eqref{curvature_control_transformed_plane} solve the actual Problem~\ref{problem_actual_plane} in the original plane. 

\begin{thm}\label{thm_problem_1_solution}
Under the conditions stated in Theorem~\ref{thm_control_transformed_plane}, with the design parameters $s_d$ and $\delta_S$ chosen according to Corollary~\ref{cor_positive_speeds}, Problem~\ref{problem_actual_plane} is solved in the following sense:  
\begin{itemize}[leftmargin=*]
	\item[(a)] The agents \eqref{model_actual_plane} in the original plane converge to the desired circular orbit $\mathcal{C}$, i.e., $|r_k(t)| = 1, \ \forall k$, as $t \to \infty$, while their positions remain within the set $|r_k(t) - \lambda| < \mu$ for all $k$ and $t \geq 0$, where $\lambda$ and $\mu$ are as defined in Lemma~\ref{lem_mapping_of_circles}.
	\item[(b)] At steady state, the agents' heading angles $\{\theta_k\}_{k=1}^{N}$ achieve M\"{o}bius phase-shift coupled synchronization or balancing in the sense of Definition~\ref{def_mobius_synchronization_balancing}, under the same interaction graph $\mathcal{G}$ with Laplacian $\mathcal{L}$ and control gains $\kappa_1, \kappa_2, \mathcal{K}$ as in Theorem~\ref{thm_control_transformed_plane}. Specifically:  
	\begin{itemize}
		\item if $\mathcal{G}$ is undirected and connected, with $\kappa_1 > 0$, $\kappa_2 > 0$, and $\mathcal{K} < 0$, then $\theta_k = \theta_0$ for all $k$, for some $\theta_0 \in \mathbb{R}$;  
		\item if $\mathcal{G}$ is additionally circulant, with $\kappa_1 > 0$, $\kappa_2 > 0$, and $\mathcal{K} > 0$, then $\pmb{1}^\top_N {\rm e}^{i(\pmb{\theta} + \pmb{\chi})} = 0$, where $\chi_k$ is defined in \eqref{chi_argument_form}.
	\end{itemize}
\end{itemize}
\end{thm}

\begin{proof}
For the proof of part (a), please refer to \cite[pg. 10]{singh2025mobius}. The proof of part (b) directly follows by noticing that synchronization and balancing of $\gamma_k$ in the transformed plane corresponds to the synchronization and balancing of $\theta_k + \chi_k$ in the original plane, according to \eqref{gamma_k}. Further, using Lemma~\ref{lem_synchronization_equivalence_both_planes}, it can be concluded that synchronization of $\gamma_k$ is equivalent to that of $\theta_k$ in the original plane at the steady-state. Therefore, the results follow for the corresponding setting of control gains as in Theorem~\ref{thm_control_transformed_plane}. Hence, proved.      
\end{proof}

Note that Theorem~\ref{thm_control_transformed_plane} designs the control laws in the transformed plane. It remains to obtain the controllers $u_k$ and $\omega_k$ in \eqref{model_actual_plane}, applied in the original plane. We establish this connection in the next section, followed by an analysis of the post-design signals based on potential functions \eqref{V_s} and \eqref{V_b}.

\section{Control Laws in Original Plane and Post-Design Signal Analysis}\label{section_6_control_original_plane}

\subsection{Control Laws in Original Plane}

Before preceding with the main results, we state the following lemma from \cite{singh2025mobius}:

\begin{lem}[\hspace{-.1pt}\cite{singh2025mobius}]\label{lem_chi_zeta_dot}
	The time-derivatives of $\chi_k$ in \eqref{chi_argument_form} and $\zeta_k$ in \eqref{zeta_argument_form}, in terms of their respective coordinate frames, are:
	\begin{align}
	\label{chi_dot}	\dot{\chi}_k &= -2\alpha v_k \left[ \frac{\sin\theta_k +\alpha |r_k| \sin(\theta_k-\phi_k)}{|1+\alpha r_k|^2} \right],\\
	\label{zeta_dot} \dot{\zeta}_k &= -\left(\frac{2\alpha v_k}{1-\alpha^2}\right) [|\rho_k|\sin(\gamma_k-\psi_k) - \sin\gamma_k].
	\end{align}
\end{lem}

Next, we first express the controllers $u_k$ and $\omega_k$ in terms the transformed plane parameters in the following theorem. 

\begin{thm}\label{thm_actual_controllers_transformed_plane}
Consider agent models \eqref{model_actual_plane} and \eqref{model_transformed_plane} and let the conditions in Theorem~\ref{thm_control_transformed_plane} hold with the design parameters $s_d$ and $\delta_S$ chosen according to Corollary~\ref{cor_positive_speeds}. The controllers $[u_k, \omega_k]^\top$ are related to the controllers $[\nu_k, \Omega_k]^\top$ solely in terms of the transformed plane parameters as follows: 
\begin{align}
\label{u_k_transformed_plane} u_k &= \left|\frac{(1 - \alpha^2)}{\alpha}\right|\left[\frac{|\rho_k - 1|^2\nu_k - 2s_k \langle \rho_k-1, \dot{\rho}_k \rangle}{|\rho_k - 1|^4}\right],\\
\label{omega_k_transformed_plane} \omega_k &= \Omega_k - \left( \frac{2\alpha v_k}{1-\alpha^2} \right) [|\rho_k|\sin(\gamma_k-\psi_k)-\sin\gamma_k],
\end{align}
where $\nu_k$ and $\Omega_k$ are given by \eqref{velocity_control_transformed_plane} and \eqref{curvature_control_transformed_plane}, respectively. 
\end{thm}

\begin{proof}
The time-derivative of \eqref{v_k}, yields
	\begin{align*}
		\dot{v}_k &= \left|\frac{(1 - \alpha^2)}{\alpha}\right| \frac{d}{dt}\left(\frac{s_k}{|\rho_k - 1|^2}\right)\\
		&= \left|\frac{(1 - \alpha^2)}{\alpha}\right|\left[\frac{|\rho_k - 1|^2\dot{s}_k - 2s_k \langle \rho_k-1, \dot{\rho}_k \rangle}{|\rho_k - 1|^4}\right].
	\end{align*}
Now, substituting for $\dot{v}_k$ and $\dot{s}_k$ from \eqref{model_actual_plane} and \eqref{model_transformed_plane}, respectively, we get the required result in \eqref{u_k_transformed_plane}. Further, the time-derivative of \eqref{theta_k} implies $\dot{\theta}_k = \dot{\gamma}_k + \dot{\zeta}_k$, where, using \eqref{model_actual_plane}, \eqref{model_transformed_plane} and \eqref{zeta_dot}, we get the required result in \eqref{omega_k_transformed_plane}. Hence, proved.  	
\end{proof}

Please note that the controllers $[\nu_k, \Omega_k]^\top$ in Theorem~\ref{thm_control_transformed_plane} are formulated entirely in terms of the transformed plane parameters, as are the controllers \eqref{u_k_transformed_plane} and \eqref{omega_k_transformed_plane} in Theorem~\ref{thm_actual_controllers_transformed_plane}. However, for practical implementation in the original plane, the controllers $[u_k, \omega_k]^\top$ must be expressed exclusively in terms of the original coordinates. To this end, we first rewrite \eqref{velocity_control_transformed_plane} and \eqref{curvature_control_transformed_plane} in terms of the original plane parameters by analyzing their individual terms. Please note that 
\begin{equation}\label{speed_error_original_plane}
\tilde{s}_k = |\alpha(1 - \alpha^2)/(1 + \alpha r_k)^2|v_k - s_d \triangleq h_k(r_k),
\end{equation}
using \eqref{speed_error} and \eqref{s_k}. Further, using \eqref{mobius_transformation} and \eqref{gamma_k}, the inner product $\langle \rho_k, \mathrm{e}^{i\gamma_k} \rangle$ can be expressed as
\begin{equation}
\langle \rho_k, \mathrm{e}^{i\gamma_k} \rangle = \left \langle \frac{\alpha(r_k + \alpha)}{1 + \alpha r_k}, \mathrm{e}^{i(\theta_k + \chi_k)} \right \rangle \triangleq m_k(r_k, \theta_k), 
\end{equation}
where $\chi_k(\cdot)$ is a function of $r_k$ (see \eqref{chi} in the Appendix). The transformed error \eqref{error_transformed_plane} can be expressed as 
\begin{equation}\label{error_original_plane}
\mathcal{E}_k = \frac{\alpha(r_k + \alpha)}{1 + \alpha r_k} + i\sigma \mathrm{e}^{i(\theta_k + \chi_k)} \triangleq n_k(r_k, \theta_k),
\end{equation}
using \eqref{mobius_transformation} and \eqref{gamma_k}. Further, using $\mathrm{e}^{i\pmb{\gamma}} = [\mathrm{e}^{i\gamma_1}, \cdots, \mathrm{e}^{i\gamma_N}]^\top \in \mathbb{T}^N$, we can write $\langle  i\mathrm{e}^{i\gamma_k}, \mathcal{L}_k\mathrm{e}^{i\pmb{\gamma}}\rangle = \langle  i\mathrm{e}^{i\gamma_k}, \sum_{j\in \mathcal{N}_k} \mathcal{L}_{k_j} \mathrm{e}^{i\gamma_j}\rangle$, where $\mathcal{L}_{k_j}$ represents the $j^\text{th}$ element of the $k^\text{th}$ row of the Laplacian $\mathcal{L}$. Further, exploiting \eqref{gamma_k}, we have
\begin{align}
\nonumber	& \langle  i\mathrm{e}^{i\gamma_k}, \mathcal{L}_k\mathrm{e}^{i\pmb{\gamma}}\rangle = \sum_{j\in \mathcal{N}_k} \langle  i\mathrm{e}^{i(\theta_k + \chi_k)},  \mathcal{L}_{k_j} \mathrm{e}^{i(\theta_j + \chi_j)}\rangle \\
& = - \sum_{j\in \mathcal{N}_k} \sin ((\theta_j - \theta_k) + (\chi_j - \chi_k)) \triangleq \tau_k(\pmb{r}, \pmb{\theta}). 
\end{align} 
With these simplifications and notations, the controller \eqref{velocity_control_transformed_plane} and \eqref{curvature_control_transformed_plane} can be expressed in terms of the original plane parameters as:
\begin{align}
	\label{nu_original_plane}	\nu_k & =  -[\delta_S^2 - h_k^2(r_k)] \left[\kappa_1\frac{m_k(r_k, \theta_k)}{\delta_T^2-|n_k(r_k, \theta_k)|^2} + \kappa_2 h_k(r_k) \right],\\
	\label{Omega_original_plane}	\Omega_k & = \frac{1}{\sigma} \left[s_d + \kappa_1 \frac{ m_k(r_k, \theta_k)}{\delta_T^2-|n_k(r_k, \theta_k)|^2} + \frac{\mathcal{K}}{\sigma} \tau_k(\pmb{r}, \pmb{\theta}) \right].
\end{align}
We now state the following result analogous to Theorem~\ref{thm_actual_controllers_transformed_plane}.
\begin{thm}\label{thm_actual_controllers_original_plane}
Consider agent models \eqref{model_actual_plane} and \eqref{model_transformed_plane} and let the conditions in Theorem~\ref{thm_control_transformed_plane} hold with the design parameters $s_d$ and $\delta_S$ chosen according to Corollary~\ref{cor_positive_speeds}. The controllers $[u_k, \omega_k]^\top$ are related to the controllers $[\nu_k, \Omega_k]^\top$ solely in terms of the original plane parameters as follows:
	\begin{align}
		\label{u_k_original_plane} u_k &= \left|\frac{(1 + \alpha r_k)^2}{\alpha(1 - \alpha^2)}\right|\nu_k + 2\alpha v_k\frac{\langle 1 + \alpha r_k, \dot{r}_k \rangle}{|1 + \alpha r_k|^2},\\
		\label{omega_k_original_plane} \omega_k &= \Omega_k +  2\alpha v_k \left[ \frac{\sin\theta_k +\alpha |r_k| \sin(\theta_k-\phi_k)}{|1+\alpha r_k|^2} \right],
	\end{align}
	where $\nu_k$ and $\Omega_k$ are given by \eqref{nu_original_plane} and \eqref{Omega_original_plane}, respectively. 
\end{thm} 

\begin{proof}
Using \eqref{s_k}, we have $\dot{v}_k = (1/|\alpha(1 - \alpha^2)|)(d/dt)(|1 + \alpha r_k|^2s_k) = (1/|\alpha(1 - \alpha^2)|)(|1 + \alpha r_k|^2\dot{s}_k + 2\alpha s_k \langle 1 + \alpha r_k, \dot{r}_k \rangle)$. Now, substituting for $s_k$ from \eqref{s_k} and using \eqref{model_actual_plane} and \eqref{model_transformed_plane}, we get the desired result in \eqref{u_k_original_plane}. To prove second part, the time-derivative of \eqref{gamma_k} results in $\dot{\gamma}_k = \dot{\theta}_k + \dot{\chi}_k \implies \omega_k = \Omega_k - \dot{\chi}_k$ using \eqref{model_actual_plane} and \eqref{model_transformed_plane}. Now, substituting for $\dot{\chi}_k$ from \eqref{chi_dot} we get the required result in \eqref{omega_k_original_plane}. 
\end{proof}

\begin{remark}
Please note that the proposed control design in Theorem~\ref{thm_control_transformed_plane} is applicable only if $|\tilde{s}_k(0)| < \delta_S$ and $|\mathcal{E}_k(0)| < \delta_T$ for all $k$. These requirements specify certain conditions on the initial states $[{\rho}_k(0), {s}_k(0), {\gamma}_k(0)]^\top$ in the transformed plane and must be mapped back to the original plane in terms of the initial states $[{r}_k(0), {v}_k(0), {\theta}_k(0)]^\top$ for every $k$. Using \eqref{speed_error_original_plane} and \eqref{error_original_plane}, the requirements translate to the conditions $|h_k(r_k(0))| < \delta_S$ and $|n_k(r_k(0), \theta_k(0))| < \delta_T$ for each $k$, expressed in the original plane parameters. Once the initial positions $r_k(0)$ are chosen within the open set $\mathcal{Z}$ in Problem~\ref{problem_actual_plane}, the speeds $v_k(0)$ and heading angles $\theta_k(0)$ can then be selected to ensure these inequalities are satisfied for all $k$.
\end{remark}

\begin{remark}
Note that the controllers \eqref{u_k_original_plane} and \eqref{omega_k_original_plane} may be undefined if either $1 + \alpha r_k = 0$ or $\alpha = 1$. Specifically, \eqref{u_k_original_plane} is undefined if either of the previous relations holds, and \eqref{omega_k_original_plane} is undefined if $1 + \alpha r_k = 0$. However, none of these situations arises in our analysis as elaborated in Lemma~\ref{lem_transformation_derivative_properties}. From the perspective of closed-loop system under the proposed controllers, it follows from \eqref{position_transformed_plane} and \eqref{rho_bounds} in Theorem~\ref{thm_control_transformed_plane}(c) that $|\alpha(r_k(t) + \alpha)/(1 + \alpha r_k(t))| \in \left(2|\alpha| - |(\lambda + \alpha)/\mu|, |(\lambda + \alpha)/\mu|\right)$ if $\alpha = \alpha_s$, and $|\alpha(r_k(t) + \alpha)/(1 + \alpha r_k(t))| \in \left(|(\lambda + \alpha)/\mu|, 2|\alpha| - |(\lambda + \alpha)/\mu|\right)$ if $\alpha = \alpha_{\ell}$, for all $t \geq 0$. Please note that despite $1 + \alpha r_k(t)$ being in the denominator, the term $|\alpha(r_k(t) + \alpha)/(1 + \alpha r_k(t))|$ is bounded for all time $t \geq 0$. Consequently, it follows that $r_k \neq -(1/\alpha)$ otherwise the previous bounds are violated as numerator $r_k(t) + \alpha$ is finite even if $r_k = -(1/\alpha)$. Also, $\alpha \neq \pm 1$, otherwise the hypothesis of Lemma~\ref{lem_mapping_of_circles} does not hold. Therefore, \eqref{u_k_original_plane} and \eqref{omega_k_original_plane} are well-defined in our analysis. 
\end{remark}

\subsection{Post-Design Signal Analysis}
This subsection establishes tighter bounds on the closed-loop signals by leveraging the results of Theorem~\ref{thm_control_transformed_plane}. By replacing $\mathcal{E}_k, s_k, \gamma_k$ in terms of the original plane parameters, \eqref{V_s} and \eqref{V_b} can be expressed as $\mathcal{V}_s(\pmb{r}, \pmb{v}, \pmb{\theta})$ and $\mathcal{V}_b(\pmb{r}, \pmb{v}, \pmb{\theta})$, respectively. Accordingly, the initial values are given by $\mathcal{V}_s(\pmb{r}(0), \pmb{v}(0), \pmb{\theta}(0)) \triangleq \mathcal{V}_s(0)$ and $\mathcal{V}_b(\pmb{r}(0), \pmb{v}(0), \pmb{\theta}(0)) \triangleq \mathcal{V}_b(0)$. Define the constants $c_s = \sqrt{1-\mathrm{e}^{-({2\mathcal{V}_s(0)}/{\kappa_1})}}$ and $c_b =  \sqrt{1-\mathrm{e}^{({-2\mathcal{V}_b(0)}/{\kappa_1})}}$. Based on these, we introduce the following combination of the two radii, $|\alpha|$ and $\left|(\lambda + \alpha)/\mu\right|$, of the circles in the transformed plane: $\eta^c_{+} = (1 - c)|\alpha| + c|(\lambda + \alpha)/\mu|$ and $\eta^c_{-} = (1 + c)|\alpha| - c|(\lambda + \alpha)/\mu|$, where $c = c_s$ or $c = c_b$ depending on the context. We analogously define constants $\ell_s = \sqrt{1-\mathrm{e}^{-{2\mathcal{V}_s(0)}}}$ and $\ell_b = \sqrt{1-\mathrm{e}^{{-2\mathcal{V}_b(0)}}}$. We now state the following results based on \eqref{V_s} and \eqref{V_b}.  

 \begin{thm}[Synchronization]\label{thm_bounds_synchronization}
Let $N$ agents be described by the model \eqref{model_transformed_plane} and interacting over an undirected and connected graph $\mathcal{G}$ with Laplacian matrix $\mathcal{L}$. Assume that the agents are governed by the control laws \eqref{velocity_control_transformed_plane} and \eqref{curvature_control_transformed_plane} with  controller gains $\kappa_1 > 0$, $\kappa_2 > 0$, and $\mathcal{K} < 0$, and that their initial conditions lie in the set $\mathcal{W}_{\delta}$ as defined in Theorem~\ref{thm_control_transformed_plane}. Then, the following properties hold in both planes: 
\begin{itemize}[leftmargin=*]
	\item[(a)] In the transformed plane, the absolute values of error $\mathcal{E}_k(t)$, position $\rho_k(t)$, speed $s_k(t)$ and the summation of squared relative phasors $\mathrm{e}^{i\gamma_j(t)} - \mathrm{e}^{i\gamma_k(t)}$ belong to the following compact sets for all $k$ and $t \geq 0$: 
	\begin{align}
		\label{error_bound_transformed_plane_sync}	& |\mathcal{E}_k(t)|  \leq \delta_Tc_s,\\
		\label{position_bound_transformed_plane_sync}	& |\rho_k(t)| \in 
		\begin{cases}
			\left[\eta^{c_s}_{-}, \eta^{c_s}_{+}\right], & \text{if}~ \alpha = \alpha_s\\
			\left[\eta^{c_s}_{+}, \eta^{c_s}_{-}\right], & \text{if}~ \alpha = \alpha_{\ell}
		\end{cases},\\
		\label{speed_bound_transformed_plane_sync}	& s_k(t) \in [s_d - \delta_S\ell_s, s_d + \delta_S\ell_s],\\
		\label{angle_bound_transformed_plane_sync}	& \hspace*{-20pt}\sum_{\{j,k\} \in \mathbb{E}} \left|\mathrm{e}^{i\gamma_j(t)} - \mathrm{e}^{i\gamma_k(t)} \right|^2  \in \left[0, \min \left\{\frac{-2\mathcal{V}_s(0)}{\mathcal{K}} ,4|\mathbb{E}|\right\} \right].
	\end{align}
	\item[(b)] In the original plane, the position $r_k(t)$ and the absolute value of speed $v_k(t)$ belong to the following compact sets for all $k$ and $t \geq 0$: 
	\begin{align}
		\label{position_bound_original_plane_sync} & \left|r_k(t) + \frac{\alpha^2 - (\eta^{c_s}_{+})^2}{\alpha(1 - (\eta^{c_s}_{+})^2)}\right| \leq \left|\frac{\eta^{c_s}_{+}(1 - \alpha^2)}{\alpha(1 - (\eta^{c_s}_{+})^2)}\right|,\\
		\label{speed_bound_original_plane_sync} & \hspace*{-0.27cm} v_k(t) \in \left[(s_d - \delta_S\ell_s) \left|\frac{df_k}{dr_k}\right|^{-1}, (s_d + \delta_S\ell_s)\left|\frac{df_k}{dr_k}\right|^{-1}\right],
	\end{align}
	irrespective of the roots $\alpha_s$ or $\alpha_{\ell}$. 
\end{itemize}
\end{thm}

\begin{proof}
	(a) From Theorem~\ref{thm_control_transformed_plane}, since $\dot{\mathcal{V}}_s \leq 0$ for the initial conditions in the set $\mathcal{W}_{\delta}$, $\mathcal{V}_s(\pmb{\rho}(t),\pmb{s}(t),\pmb{\gamma}(t)) \leq \mathcal{V}_s(0)$ for all $t \geq 0$ in $\mathcal{W}_{\delta}$. By applying the preceding inequality for individual terms in \eqref{V_s}, one can obtain the desired bounds since all three terms in \eqref{V_s} are positive. Considering the first term $\kappa_1\mathcal{S}(\pmb{\mathcal{E}})$ in \eqref{V_s} and using \eqref{blf}, it can be concluded that $(\kappa_1/2)\sum_{k=1}^{N} \ln[{\delta_T^2}/(\delta_T^2-|\mathcal{E}_k(t)|^2)] \leq \mathcal{V}_s(0) \implies (\kappa_1/2) \ln[\delta_T^2/(\delta_T^2-|\mathcal{E}_k(t)|^2)] \leq \mathcal{V}_s(0) \implies {\delta_T^2}/(\delta_T^2-|\mathcal{E}_k(t)|^2) \leq \mathrm{e}^{{2\mathcal{V}_s(0)}/{\kappa_1}}, \forall k, \forall t \geq 0$. Since $|\mathcal{E}_k(t)| < \delta_T$ as proved in Theorem~\ref{thm_control_transformed_plane}, $\delta_T^2-|\mathcal{E}_k(t)|^2 > 0$, and hence, the preceding inequality can be simplified as $|\mathcal{E}_k(t)| \leq  \delta_T\sqrt{1-\mathrm{e}^{-({2\mathcal{V}_s(0)}/{\kappa_1})}} \implies |\mathcal{E}_k(t)| \leq \delta_T c_s, \forall k, \forall t\geq 0$, where $c_s$ is defined above. This proves \eqref{error_bound_transformed_plane_sync}. Further, exploiting \eqref{error_transformed_plane}, we can rewrite the previous inequality as $|\rho_k \pm i |\alpha|\mathrm{e}^{i\gamma_k}| \leq \delta_Tc_s, \forall k, \forall t\geq 0$. Moreover, using the fact $\Big{|}|w_1| - |w_2|\Big{|} \leq |w_1 + w_2|$ for some $w_1 \in \mathbb{C}, w_2 \in \mathbb{C}$, we can write $\pm(|\rho_k|-|\alpha|) < \delta_Tc_s, \forall k, \forall t \geq 0$. Furthermore, considering $+$ and $-$ sign gives $|\rho_k| < \delta_Tc_s+|\alpha|$ and $|\alpha|-\delta_Tc_s< |\rho_k|$ respectively. Consequently, we get $|\alpha|-\delta_Tc_s < |\rho_k(t)| <  \delta_Tc_s+|\alpha|, \forall k, \forall t \geq 0$. Moreover, substitution for $\delta_T$ from \eqref{delta_transformed} gives the desired result \eqref{position_bound_transformed_plane_sync}. \par 
	
	A similar argument for the third term $\mathcal{H}(\pmb{s})$ in \eqref{V_s}, this implies using \eqref{blf_speed} that $({1}/{2}) \sum_{k=1}^{N} \ln [\delta_S^2/(\delta_S^2 - \tilde{s}_k^2)] \leq \mathcal{V}_s(0) \implies \ln [\delta_S^2/(\delta_S^2 - \tilde{s}_k^2)] \leq 2\mathcal{V}_s(0), \forall k, \forall t \geq 0$. Using the similar analysis as carried out for the first term, it can be concluded that $|\tilde{s}_k(t)| \leq \delta_S\sqrt{1-\mathrm{e}^{-2\mathcal{V}_s(0)}} \implies |\tilde{s}_k(t)| \leq \delta_S\ell_s, \forall k, \forall t \geq 0$ ($\ell_s$ is defined above). Now, using \eqref{speed_error}, one can rewrite the preceding inequality as $s_d - \delta_S\ell_s \leq s_k \leq s_d + \delta_S\ell_s, \forall k, \forall t \geq 0$. Further, using a similar argument for the second term $\mathcal{U}(\pmb{\gamma})$ in \eqref{V_s}, it can be concluded using \eqref{potential_synchro_balanc} that $\langle \mathrm{e}^{i\pmb{\gamma}}, \mathcal{L}\mathrm{e}^{i\pmb{\gamma}} \rangle \leq -2\mathcal{V}_s(0)/\mathcal{K}$ since $\mathcal{K} < 0$. It can be inferred that $\langle \mathrm{e}^{i\pmb{\gamma}}, \mathcal{L}\mathrm{e}^{i\pmb{\gamma}} \rangle = \sum_{\{j,k\}\in \mathbb{E}} \left|\mathrm{e}^{i\gamma_j} - \mathrm{e}^{i\gamma_k} \right|^2  \leq 4|\mathbb{E}|$. Consequently, the required bound in \eqref{angle_bound_transformed_plane_sync} is obtained. \par
	
	(b) The proof of \eqref{position_bound_original_plane_sync} follows similar to \cite{singh2025mobius} and omitted for brevity. For proving \eqref{speed_bound_original_plane_sync}, we substitute $s_k$ from \eqref{s_k} in \eqref{speed_bound_transformed_plane_sync}, and using \eqref{position_derivative_mobius_mapping}, get the desired result. 
\end{proof}

\begin{figure}[t!]
	\centering
	\subfigure[Topology]{\includegraphics[width=2.7cm]{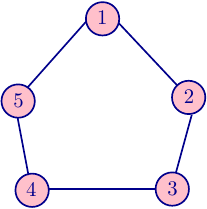}\label{nt}}\hspace{0.5cm}
	\subfigure[Laplacian Matrix]{\includegraphics[width=4.5cm]{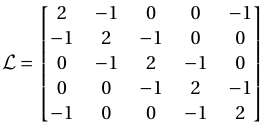}\label{lm}}
	\caption{Topology and Laplacian for $N = 5$ agents.}
	\label{fig_network topology}	
\end{figure}

\begin{table*}[t!]
	\centering
	\caption {Agents' corresponding initial conditions in both the planes.}\label{table_initial_condition}
	\begin{tabular}{c c c c c | c c c c c}
		\hline
		\multicolumn{5}{c}{Original Plane} & \multicolumn{5}{c}{Transformed Plane} \\
		\hline
		$r_k(0)$ & $|r_k(0)|$ & $v_k(0)$ & $\theta_k(0)$ & $\phi_k(0)$ & $\rho_k(0)$ & $|\rho_k(0)|$ & $s_k(0)$ & $\gamma_k(0)$ & $\psi_k(0)$\\ 
		\hline
		$~~1.25 + i 0.00$	& 1.25 & 0.4 & $~~90^\circ$ & $0^\circ$       & $~~0.54+i0.00$ & 0.54 & 0.056 & $90.00^\circ$ &  $0^\circ$\\ 
		$-0.20 - i 1.15$	& 1.17 & 0.1 & $-15^\circ$  & $-99.84^\circ$  & $~~0.41-i0.38$ & 0.56 & 0.032 & $50.28^\circ$ & $-42.74^\circ$\\
		$-0.65 -i 0.90$	    & 1.11 & 0.1 & $-40^\circ$  & $-125.83^\circ$ & $0.23-i0.51$  & 0.56 & 0.057 & $27.38^\circ$ & $-65.76^\circ$\\
		$~~0.75 - i 1.10$	& 1.33 & 0.2 & $~~25^\circ$ & $-55.71^\circ$  & $~~0.53-i0.19$ & 0.56 & 0.034 & $68.60^\circ$ &  $-19.53^\circ$\\ 
		$~~1.00 - i 0.75$	& 1.25 & 0.5 & $~~50^\circ$ & $-36.87^\circ$  & $~~0.53 - i0.12$ & 0.54 & 0.078 & $78.07^\circ$ & $-12.52^\circ$\\ 
		\hline
	\end{tabular}
\end{table*} 

\begin{thm}[Balancing]\label{thm_bounds_balancing}
Let $N$ agents be described by the model \eqref{model_transformed_plane} and interacting over an undirected and connected graph $\mathcal{G}$ with Laplacian matrix $\mathcal{L}$. Assume that the agents are governed by the control laws \eqref{velocity_control_transformed_plane} and \eqref{curvature_control_transformed_plane} with  controller gains $\kappa_1 > 0$, $\kappa_2 > 0$, and $\mathcal{K} > 0$, and that their initial conditions lie in the set $\mathcal{W}_{\delta}$ as defined in Theorem~\ref{thm_control_transformed_plane}. Then, the following properties hold in both planes: 
\begin{itemize}[leftmargin=*]
	\item[(a)] In the transformed plane, the absolute values of error $\mathcal{E}_k(t)$, position $\rho_k(t)$, speed $s_k(t)$ and the summation of squared relative phasors $\mathrm{e}^{i\gamma_j(t)} - \mathrm{e}^{i\gamma_k(t)}$ belong to the following compact sets for all $k$ and $t \geq 0$: 
	\begin{align}
		\label{error_bound_transformed_plane_balanc}	& |\mathcal{E}_k(t)|  \leq \delta_Tc_b,\\
		\label{position_bound_transformed_plane_balanc}	& |\rho_k(t)| \in 
		\begin{cases}
			\left[\eta^{c_b}_{-}, \eta^{c_b}_{+}\right], & \text{if}~ \alpha = \alpha_s\\
			\left[\eta^{c_b}_{+}, \eta^{c_b}_{-}\right], & \text{if}~ \alpha = \alpha_{\ell}
		\end{cases},\\
		\label{speed_bound_transformed_plane_balanc}	& s_k(t) \in [s_d - \delta_S\ell_b, s_d + \delta_S\ell_b],\\
		\label{angle_bound_transformed_plane_balanc}	& \hspace*{-5pt}\sum_{\{j,k\} \in \mathbb{E}} \nonumber \nonumber \left|\mathrm{e}^{i\gamma_j(t)} - \mathrm{e}^{i\gamma_k(t)} \right|^2  \in \\
		& \left[\max \left\{0, \left( N\lambda_{\max} - \left(\frac{2\mathcal{V}_b(0)}{\mathcal{K}}\right)\right)\right\}, N\lambda_{\max} \right].
	\end{align}
	\item[(b)] In the original plane, the position $r_k(t)$ and the absolute value of speed $v_k(t)$ belong to the following compact sets for all $k$ and $t \geq 0$: 
	\begin{align}
		\label{position_bound_original_plane_balanc} & \left|r_k(t) + \frac{\alpha^2 - (\eta^{c_b}_{+})^2}{\alpha(1 - (\eta^{c_b}_{+})^2)}\right| \leq \left|\frac{\eta^{c_b}_{+}(1 - \alpha^2)}{\alpha(1 - (\eta^{c_b}_{+})^2)}\right|,\\
		\label{speed_bound_original_plane_balanc} & \hspace*{-0.27cm} v_k(t) \in \left[(s_d - \delta_S\ell_b)\left|\frac{df_k}{dr_k}\right|^{-1}, (s_d + \delta_S\ell_b)\left|\frac{df_k}{dr_k}\right|^{-1}\right],
	\end{align}
	irrespective of the roots $\alpha_s$ or $\alpha_{\ell}$. 
\end{itemize}
\end{thm}

\begin{proof}
Please refer to \cite{hegde2023synchronization} for the proof of \eqref{angle_bound_transformed_plane_balanc}. The remainder of the proof follows along similar lines as in Theorem~\ref{thm_bounds_synchronization}. 
\end{proof}

\section{Simulations and Experimental Validation}\label{section_7_simulations_and_experiments}

\subsection{Numerical Simulations}
Consider the setting of non-concentric circles as described in Example~\ref{example_1} where $\lambda = 0.5$ and $\mu = \sqrt{5/2}$. Let $N=5$ agents interact according to Fig.~\ref{fig_network topology}, with their initial positions, speeds, and heading angles in both planes specified in Table~\ref{table_initial_condition}. All parameters are given in SI units. According to Corollary~\ref{cor_positive_speeds}, the design parameters are selected as $s_d = \delta_S = 0.06$. Further, $\delta_T$ for the smaller root $\alpha_s = 0.5$ in \eqref{delta_transformed} is given by $\delta_T = 0.1325$ (the case for larger root $\alpha_{\ell}$ follows analogously and is omitted for brevity). As summarized in Table~\ref{table_parameter_comparision}, the feasibility requirements for the initial conditions are satisfied in both planes: $|\mathcal{E}_k(0)| \leq \delta_T$ and $|\tilde{s}_k(0)| < \delta_S$ in the transformed plane, and $|n_k(r_k(0), \theta_k(0))| < \delta_T$ and $|h_k(r_k(0))| < \delta_S$ in the original plane, for all $k=1,\ldots,5$. 

We implemented the controllers in \eqref{u_k_original_plane} and \eqref{omega_k_original_plane} using MATLAB’s \texttt{ode45} solver, with gains $\kappa_1=0.004$, $\kappa_2=10$, and $\mathcal{K}=0.04$ for M\"{o}bius phase-shift-coupled balancing (respectively, $\mathcal{K}=-0.04$ for synchronization). These gains were selected empirically to illustrate representative transient responses; however, the qualitative convergence properties hold for a range of gain values satisfying the sign conditions in Theorem~\ref{thm_control_transformed_plane}. As shown in Figs.~\ref{plot1} and \ref{plot2}, all agents converge to the desired circular orbit $\mathcal{C}$ in the original plane while remaining confined within the nonconcentric outer boundary $\mathcal{C}'$ and attaining the corresponding phase patterns. The position errors converge to zero, while the controls $u_k$ and $\omega_k$ exhibit the associated balancing and synchronization patterns. Figures~\ref{plot3} and \ref{plot4} show the corresponding behavior in the transformed plane, where the two circles are concentric. As the position errors converge to zero, the velocity controls $\nu_k$ in \eqref{velocity_control_transformed_plane} approach zero, whereas the turn-rate controls $\Omega_k$ in \eqref{curvature_control_transformed_plane} converge to the desired value $\Omega_k=s_d/|\alpha|=0.06/0.5=0.12$. The results in Figs.~\ref{plot2} and \ref{plot4} also validate the synchronization equivalence established in Lemma~\ref{lem_synchronization_equivalence_both_planes}. Finally, Fig.~\ref{plot5} confirms that the agents' speeds remain positive and bounded in both planes.

\begin{remark}
Through simulations, we observed that the phase balancing is generally not guaranteed when the graph $\mathcal{G}$ is not circulant. In such cases (with $\mathcal{K} > 0$), the agents still converge to the desired circular orbit $\mathcal{C}$ and achieve a certain phase arrangement. However, it is not ensured that $\pmb{1}_N^\top {\rm e}^{i\pmb{\gamma}} = 0$, and therefore the resulting configuration cannot be classified as balanced according to Definition~\ref{def_mobius_synchronization_balancing}. 
\end{remark}

\subsection{Experimental Validation}
We conducted experiments with four Khepera IV ground robots in a motion capture (MoCap) environment equipped with overhead cameras and infrared sensors. The MoCap system provides real-time position and heading information of each robot, which is synchronized and processed on a track manager. For closed-loop control, the Robot Operating System (ROS) was used to transmit control commands from the track manager to the robots via a WiFi channel. For the experiments, we select the first four initial conditions in Table~\ref{table_initial_condition}, where the speeds are in accordance with the maximum hardware limit of $0.814$ m/s, while maintaining the same control gains as above in phase balancing. The experimental results in Fig.~\ref{screenshot_trajectory} show a top-down snapshot of the system's configuration at $t = 100$ seconds. The recorded trajectories obtained from the track-manager are plotted in Fig.~\ref{qtm_plot}, capturing the real-time motion of each robot. The corresponding acceleration and turn rate control inputs are shown in Fig.~\ref{hardware_dfv_ac} and Fig.~\ref{hardware_dfv_tc}, validating the reliability of the proposed control laws in hardware. A video recording of the experiment is also provided at \url{https://youtu.be/FkHR-GATFrA}, where the robots can be seen naturally slowing down in tighter regions and moving faster in open areas.

\begin{table}[t!]
	\centering
	\caption{Initial condition feasibility analysis.}\label{table_parameter_comparision}	
	\begin{tabular}{c c | c c}
		\hline
		\multicolumn{2}{c}{Transformed Plane} & \multicolumn{2}{c}{Original Plane} \\
		\hline
		$|\mathcal{E}_k(0)|$ & $|\tilde{s}_k(0)|$ & $|n_k(r_k(0),\theta_k(0))|$ & $|h_k(r_k(0))|$ \\ 
		\hline
		0.0385 & 0.004 &	0.0383 & 0.0032 \\ 
		0.0639 & 0.028 &	0.0641 & 0.0272 \\
		0.0688 & 0.003 &	0.0686 & 0.0030 \\
		0.0645 & 0.026 &	0.0643 & 0.0258 \\ 
		0.0426 & 0.018 &	0.0428 & 0.0184 \\ 
		\hline
	\end{tabular}
\end{table}

\begin{figure*}[t!]
	\centering{
		\subfigure[Trajectory]{\includegraphics[width=3.4cm]{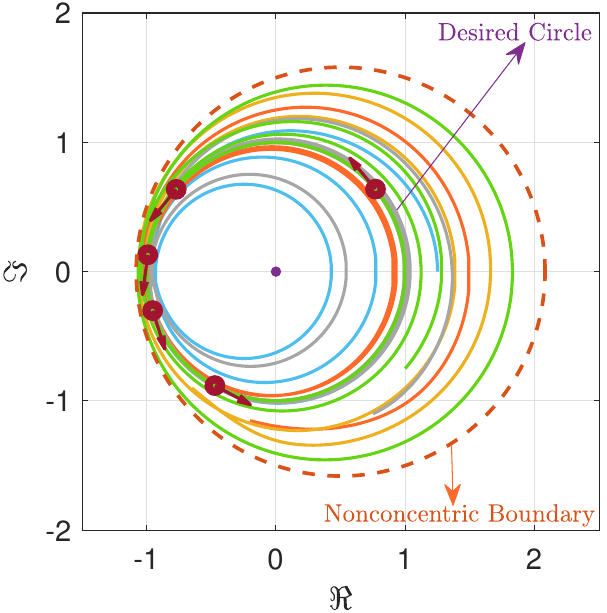}\label{trj_1}}\hspace*{0.2cm}
		\subfigure[Absolute error $|e_k|$]{\includegraphics[width=4.825cm]{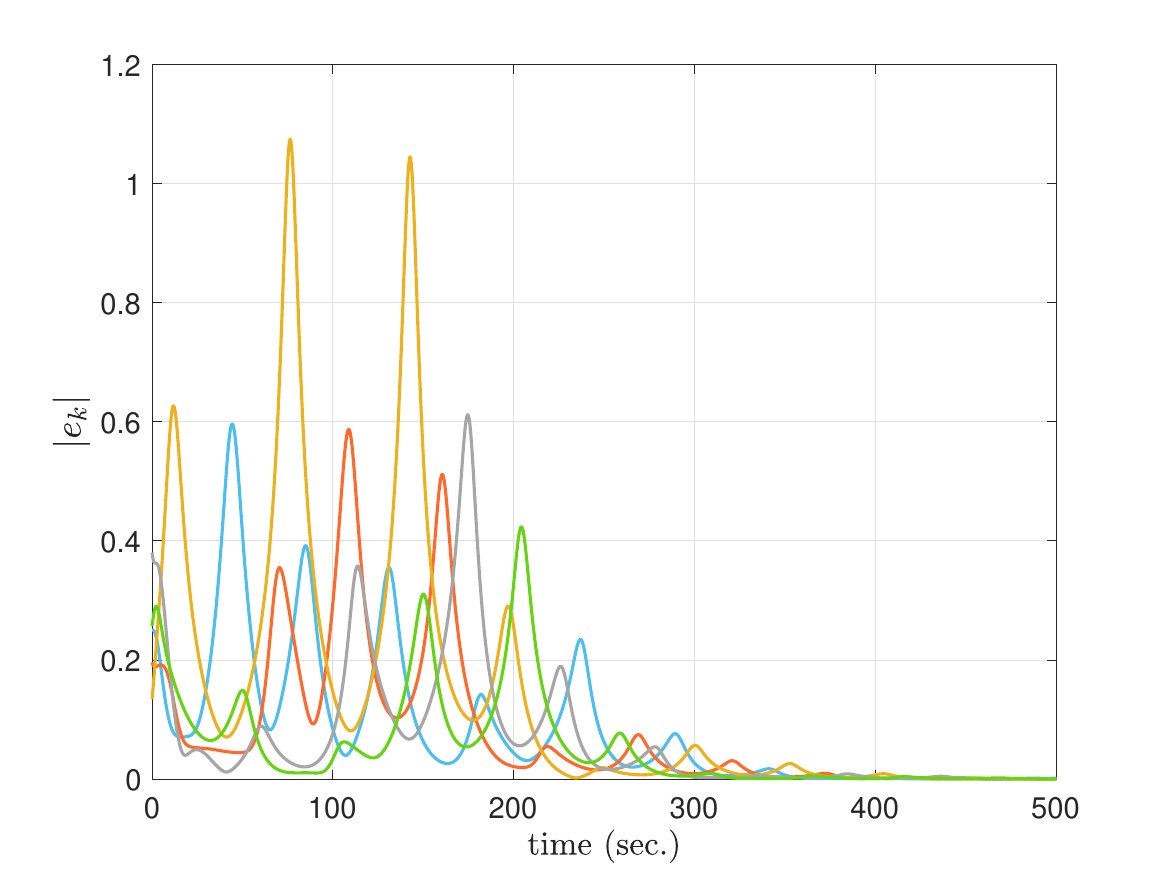}\label{error_1}}\hspace*{-0.25cm}
		\subfigure[$u_k$ with $\alpha=\alpha_s$]{\includegraphics[width=4.825cm]{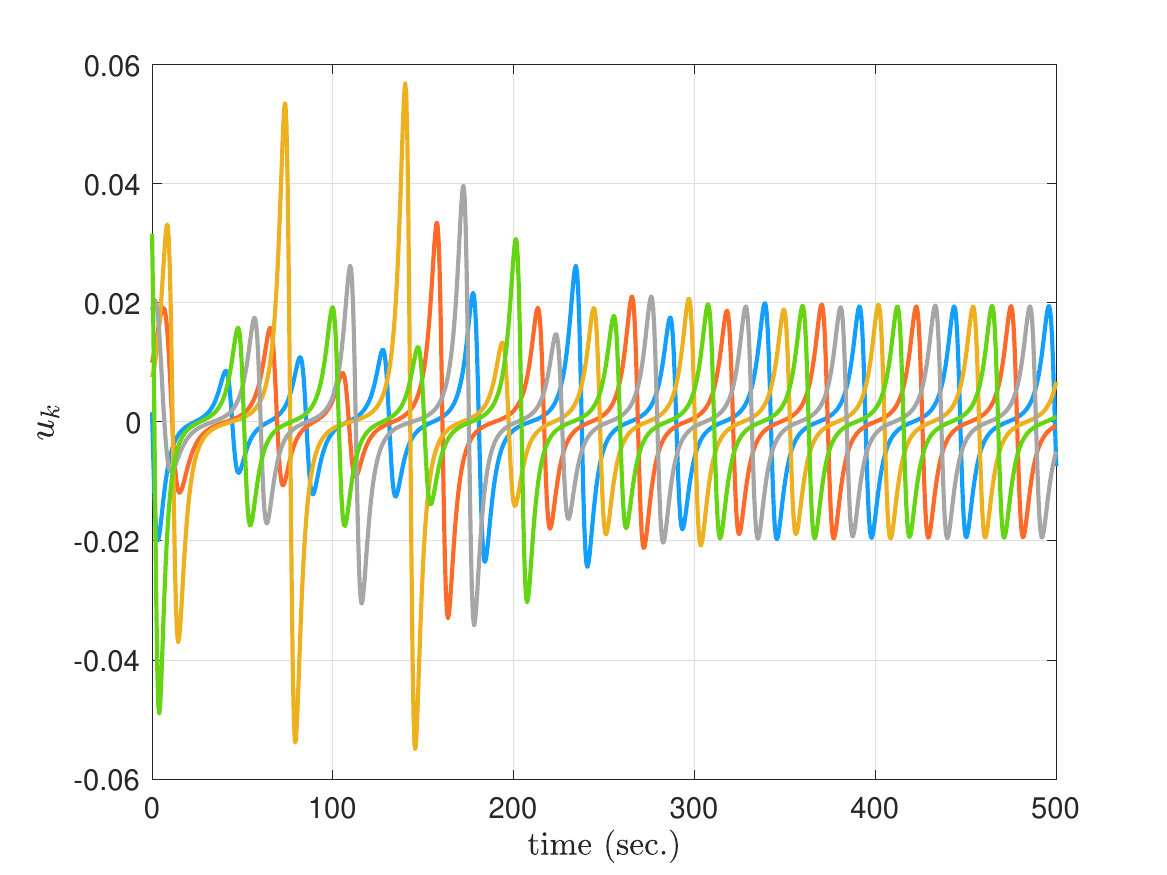}\label{uk_alpha_s}}\hspace*{-0.25cm}
		\subfigure[$\omega_k$ with $\alpha=\alpha_s$]{\includegraphics[width=4.825cm]{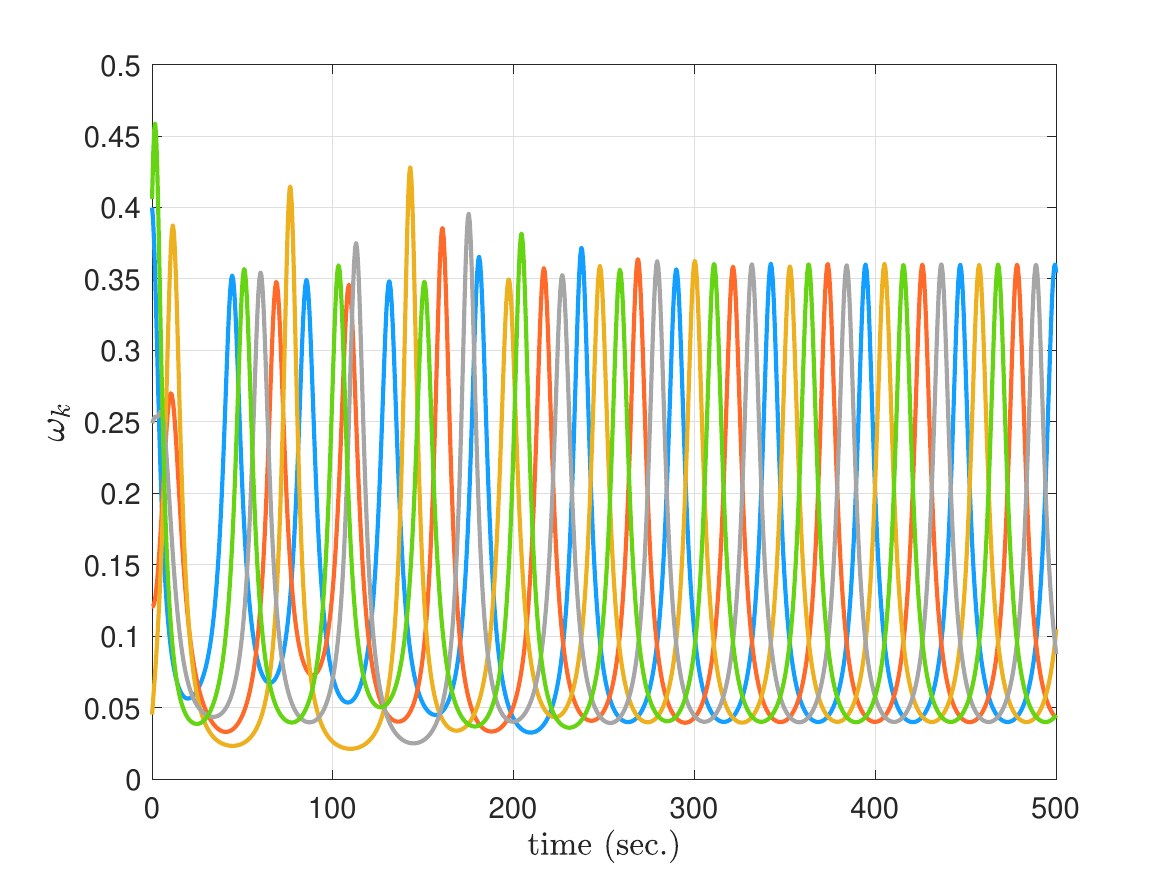}\label{omega_alpha_s}}}
	\caption{Agents with M\"{o}bius phase-shift-coupled balancing in the original plane under controllers \eqref{u_k_original_plane} and \eqref{omega_k_original_plane}.}
	\label{plot1}
\end{figure*}

\begin{figure*}[t!]
	\centering{
		\subfigure[Trajectory]{\includegraphics[width=3.4cm]{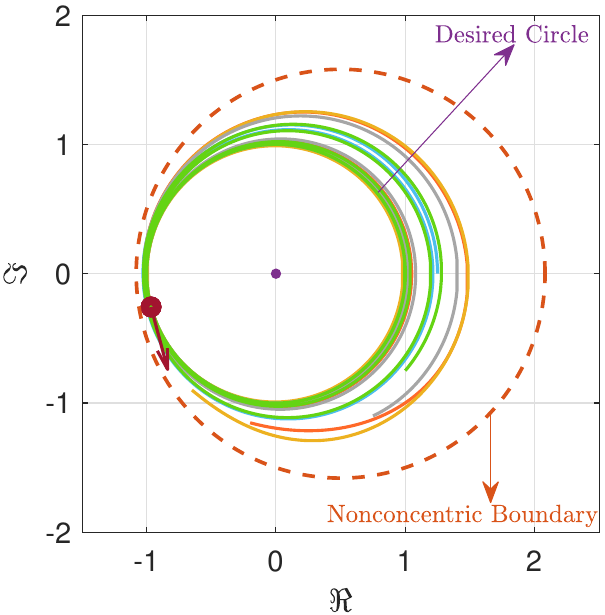}\label{trj_2}}\hspace*{0.2cm}
		\subfigure[Absolute error $|e_k|$]{\includegraphics[width=4.85cm]{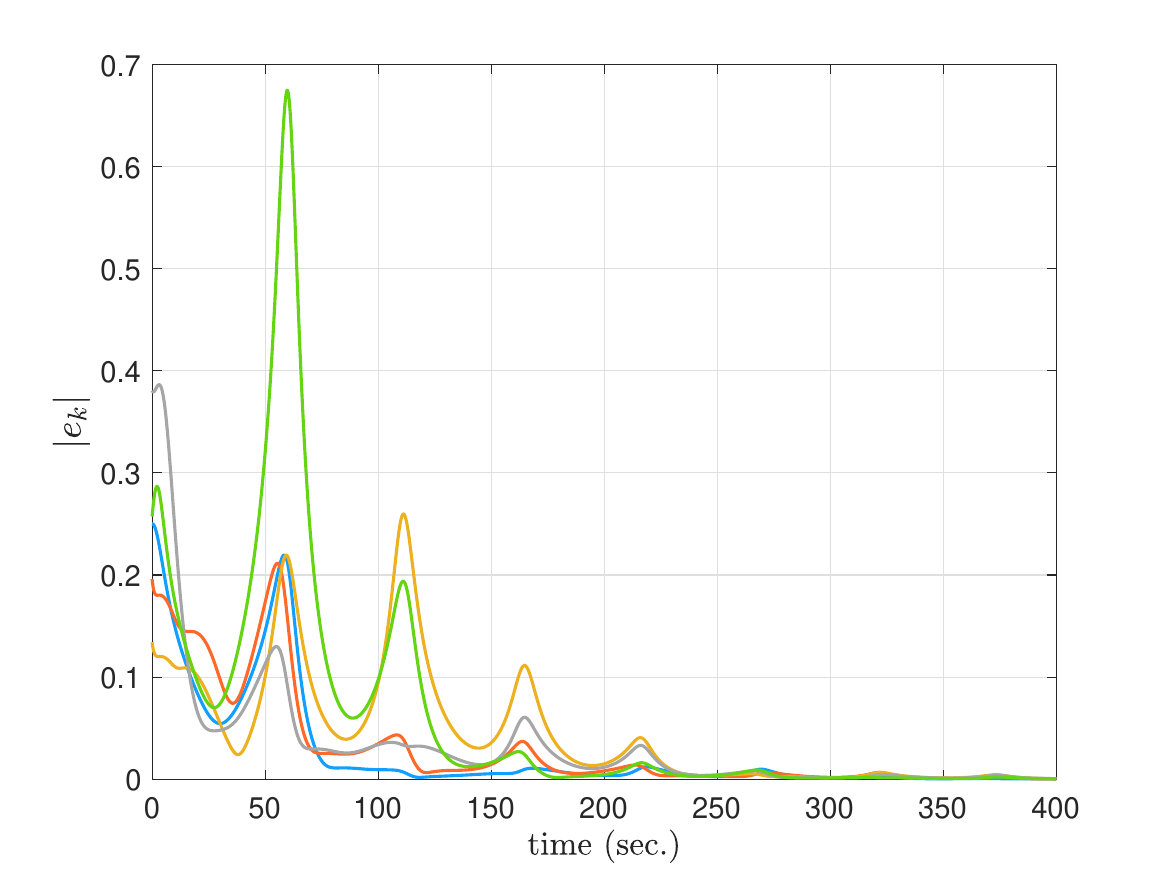}\label{error_2}}\hspace*{-0.3cm}
		\subfigure[$u_k$ with $\alpha=\alpha_s$]{\includegraphics[width=4.85cm]{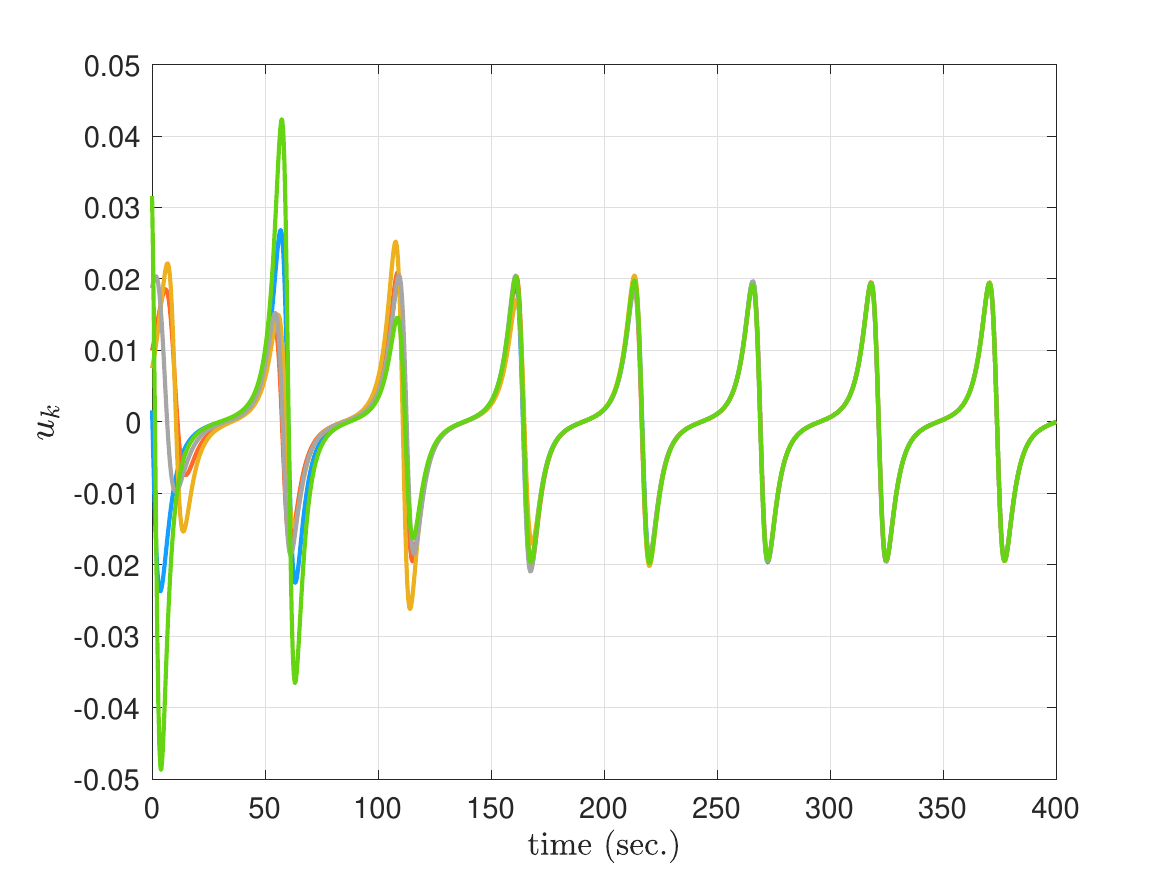}\label{uk_alpha_s2}}\hspace*{-0.3cm}
		\subfigure[$\omega_k$ with $\alpha=\alpha_s$]{\includegraphics[width=4.85cm]{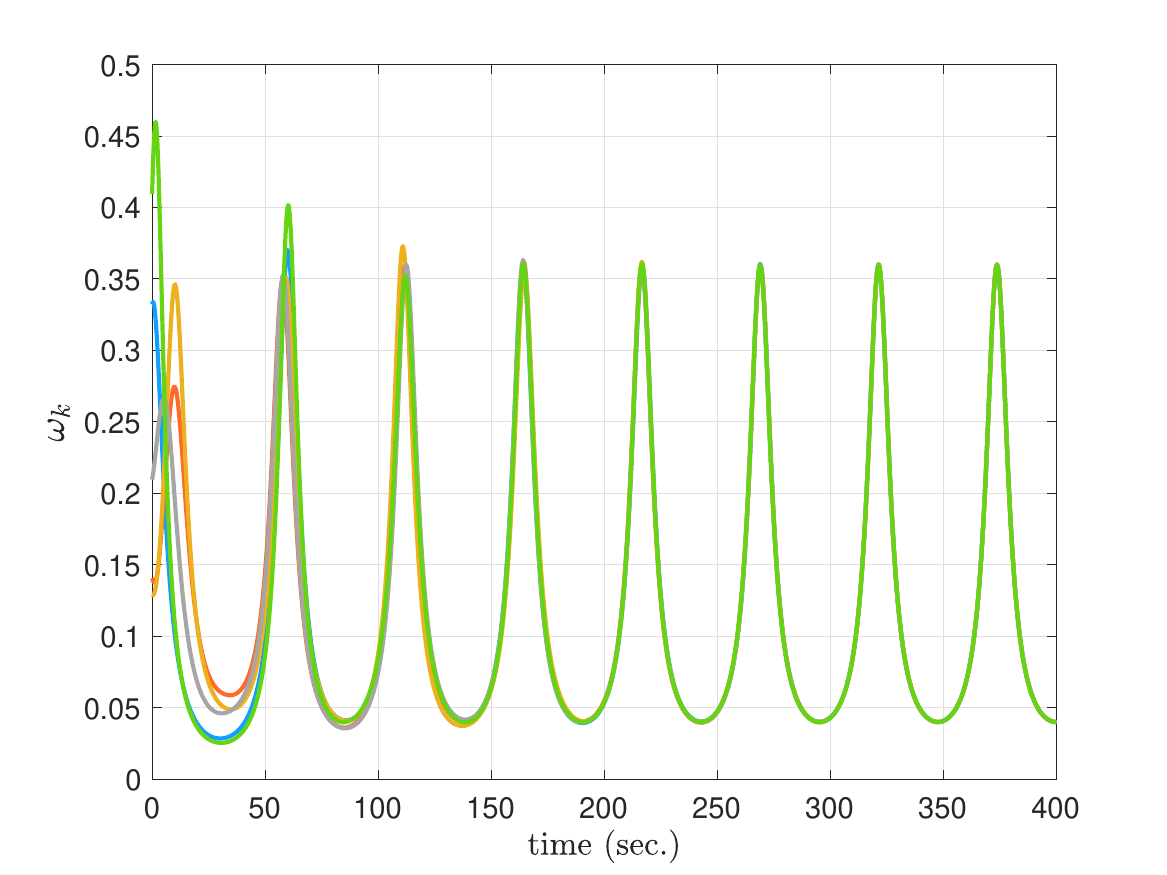}\label{omega_alpha_s2}}}
	\caption{Agents with M\"{o}bius phase-shift-coupled synchronization in the original plane under controllers \eqref{u_k_original_plane} and \eqref{omega_k_original_plane}.}
	\label{plot2}
\end{figure*}

\begin{figure*}[t!]
	\centering{
		\subfigure[Trajectory]{\includegraphics[width=3.25cm]{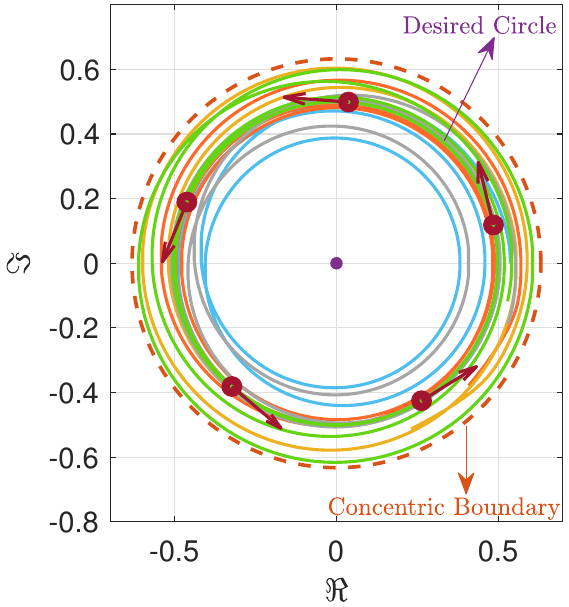}\label{tp_trjb}}\hspace*{0.2cm}
		\subfigure[Absolute error $|\mathcal{E}_k|$]{\includegraphics[width=4.875cm]{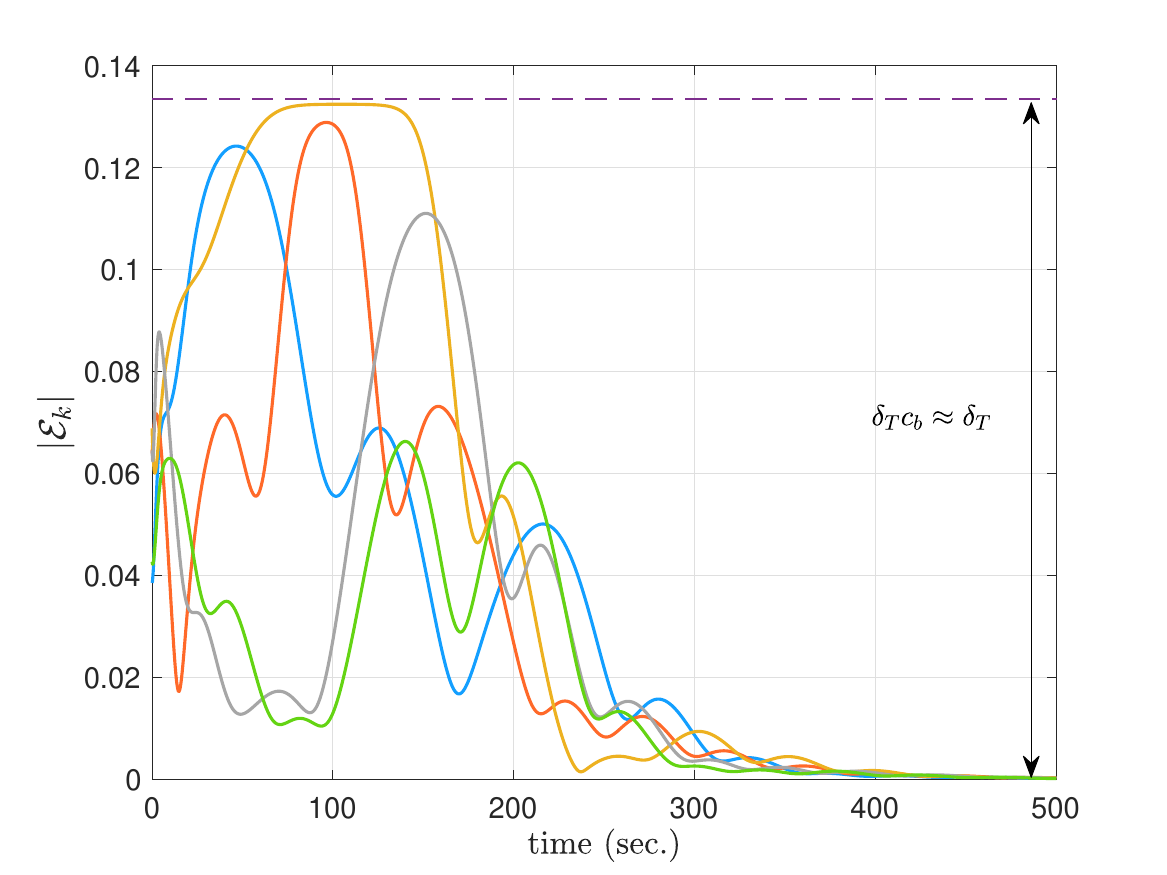}\label{tp_errorb}}\hspace*{-0.4cm}
		\subfigure[$\nu_k$ with $\alpha=\alpha_s$]{\includegraphics[width=4.875cm]{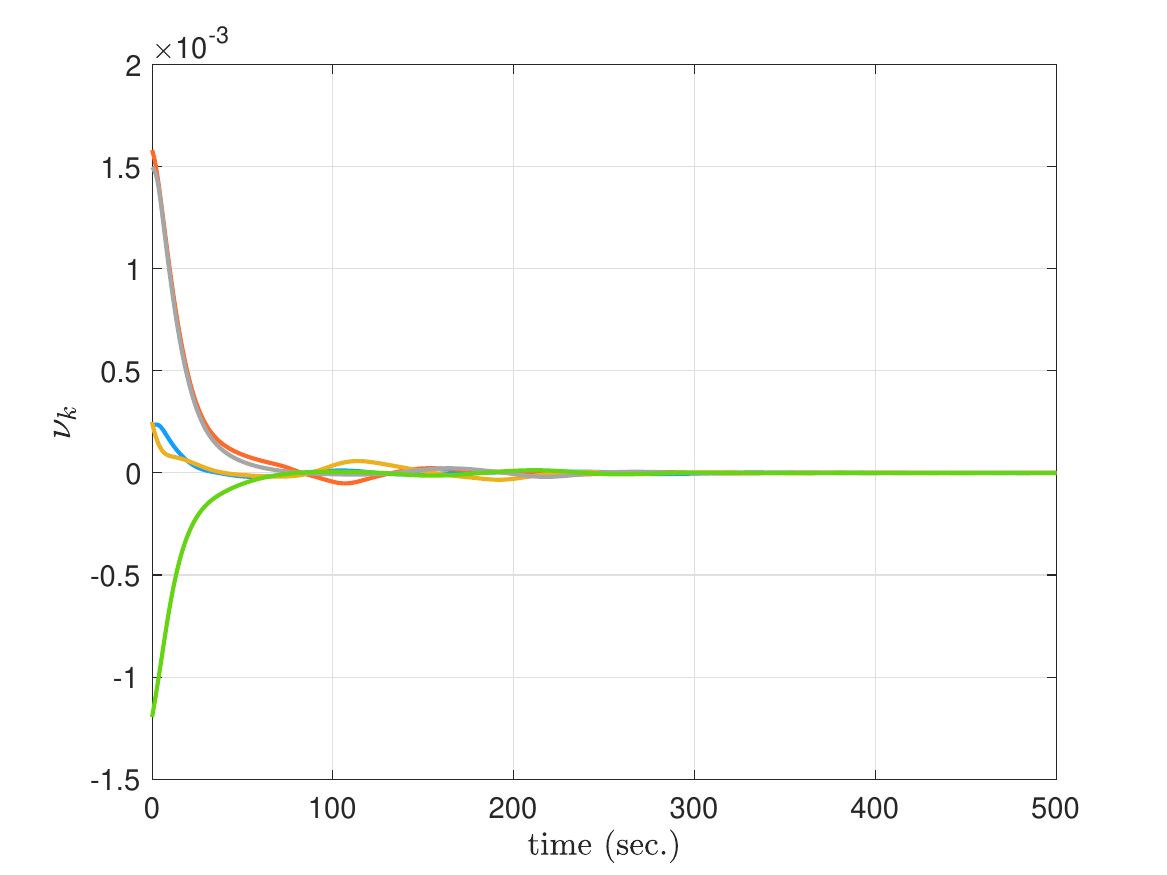}\label{tp_acc_controlb}}\hspace*{-0.4cm}
		\subfigure[$\Omega_k$ with $\alpha=\alpha_s$]{\includegraphics[width=4.875cm]{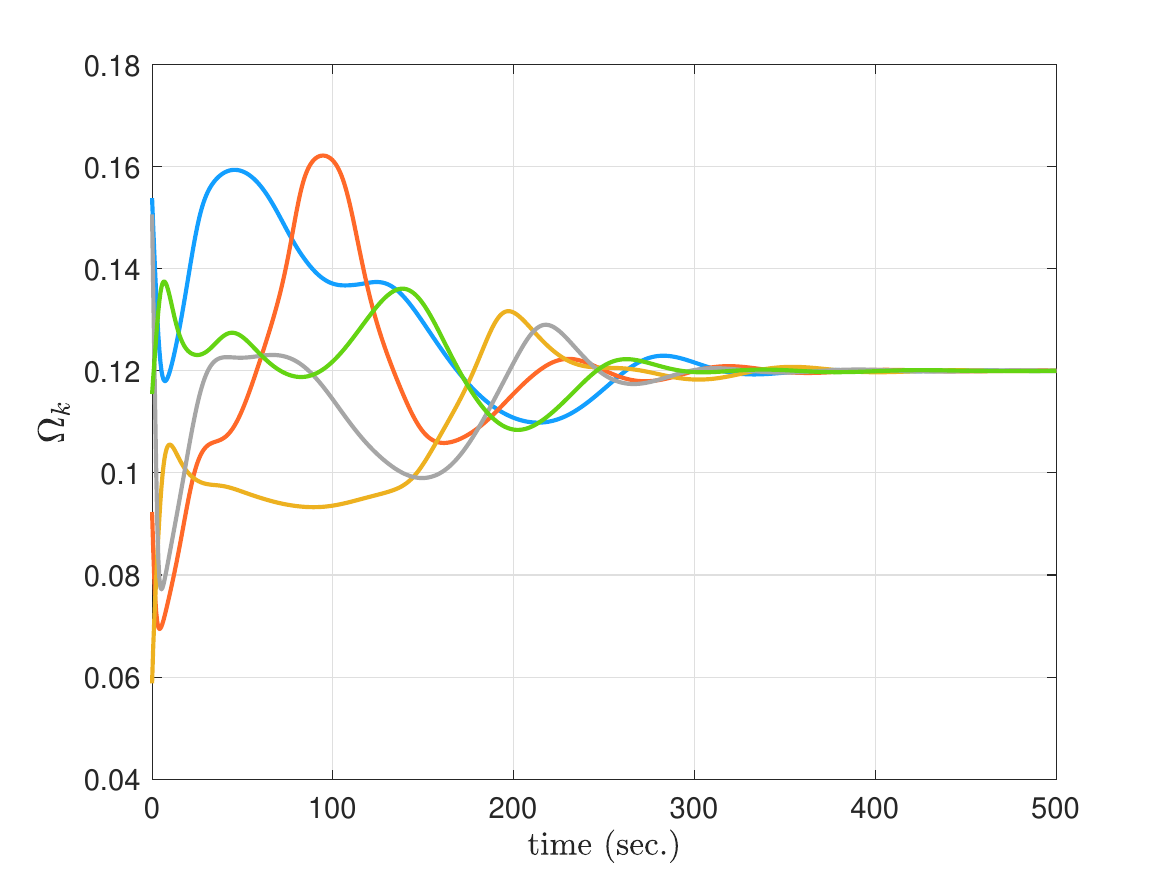}\label{tp_curv_controlb}}}
	\caption{Agents with phase balancing in the transformed plane under controllers \eqref{velocity_control_transformed_plane} and \eqref{curvature_control_transformed_plane}.}
	\label{plot3}
\end{figure*}

\begin{figure*}[t!]
	\centering{
		\subfigure[Trajectory]{\includegraphics[width=3.25cm]{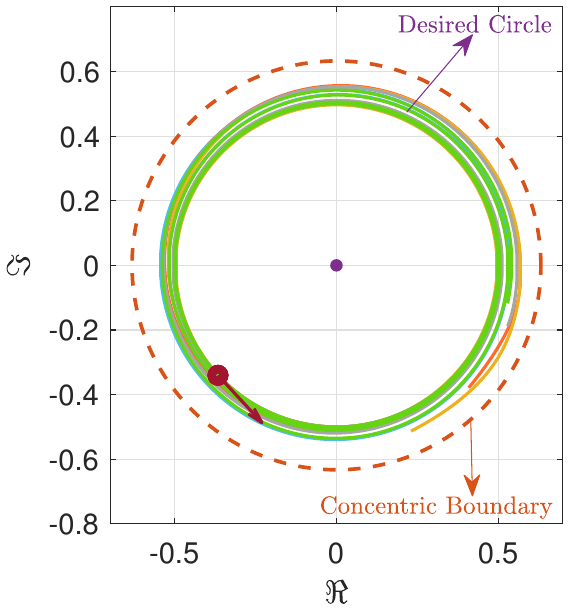}\label{tp_trjs}}\hspace*{0.2cm}
		\subfigure[Absolute error $|\mathcal{E}_k|$]{\includegraphics[width=4.9cm]{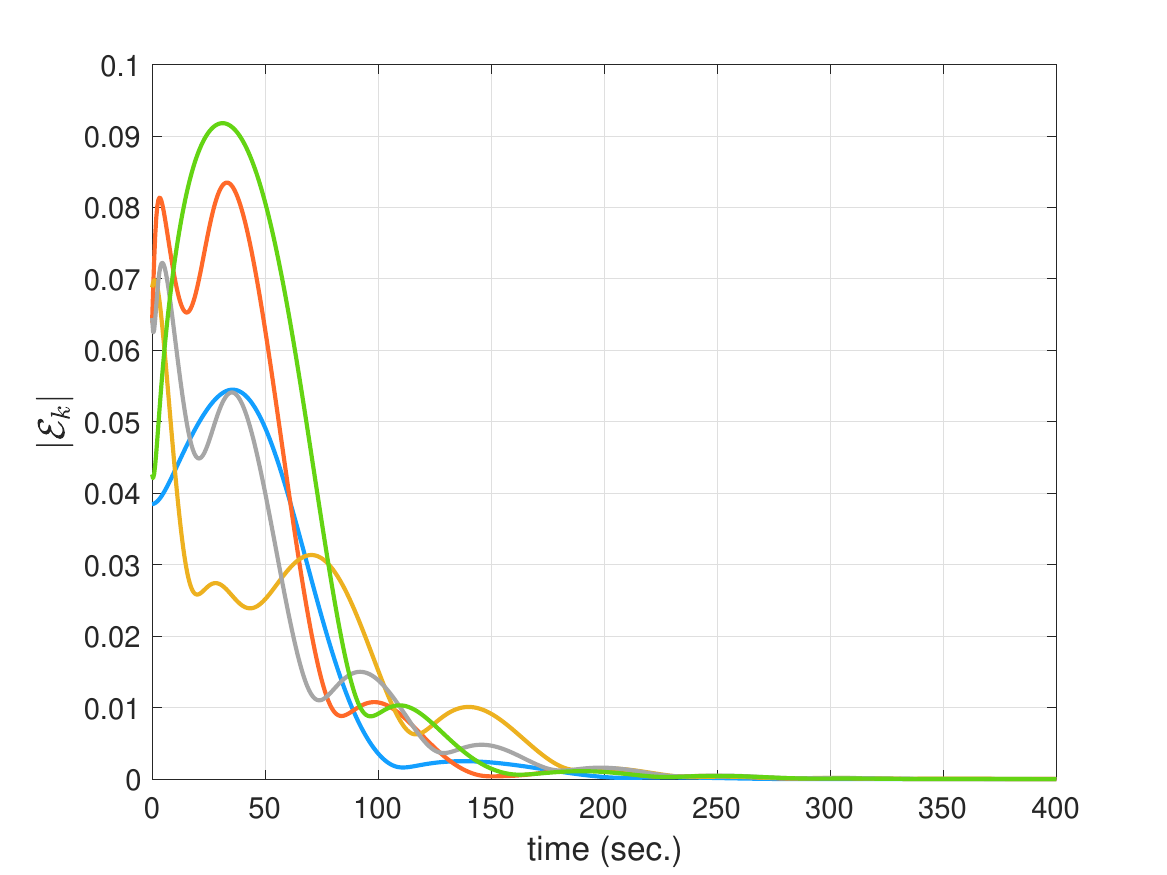}\label{tp_errors}}\hspace*{-0.4cm}
		\subfigure[$\nu_k$ with $\alpha=\alpha_s$]{\includegraphics[width=4.9cm]{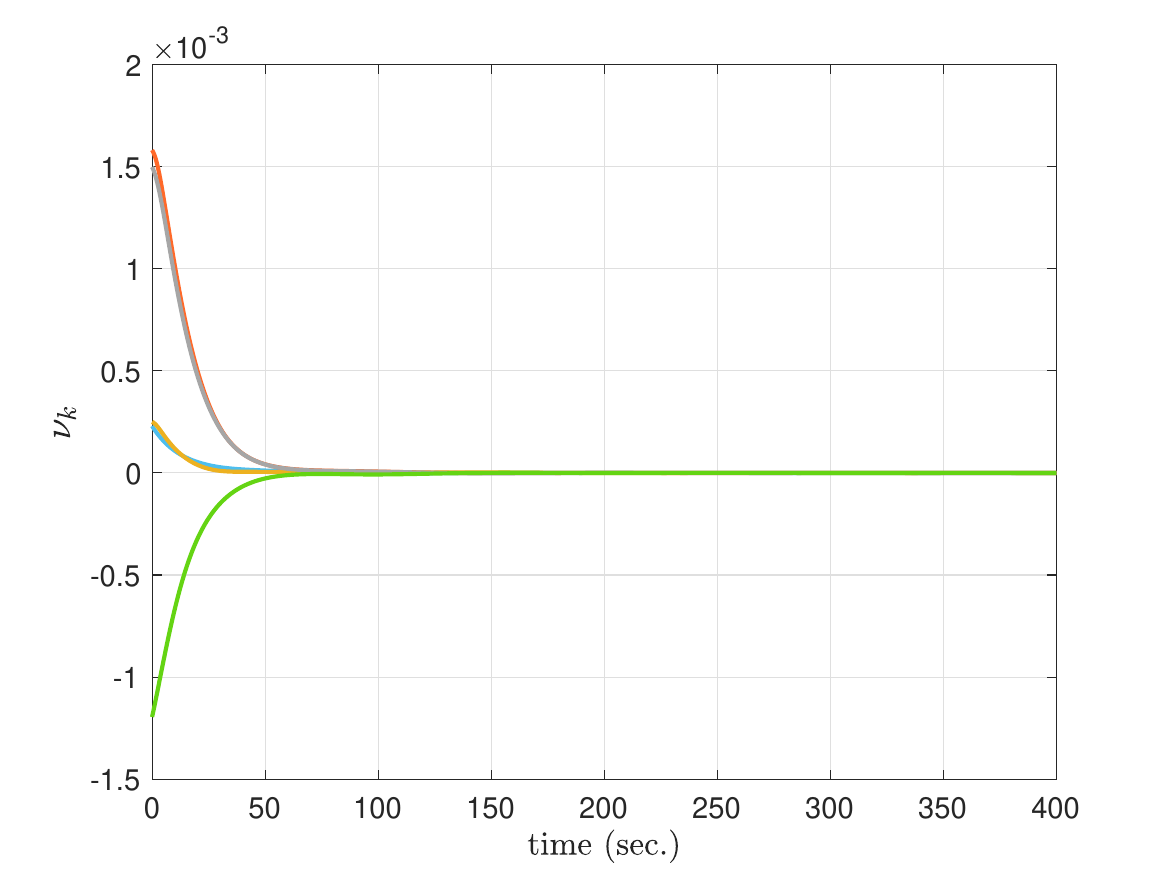}\label{tp_acc_controls}}\hspace*{-0.4cm}
		\subfigure[$\Omega_k$ with $\alpha=\alpha_s$]{\includegraphics[width=4.9cm]{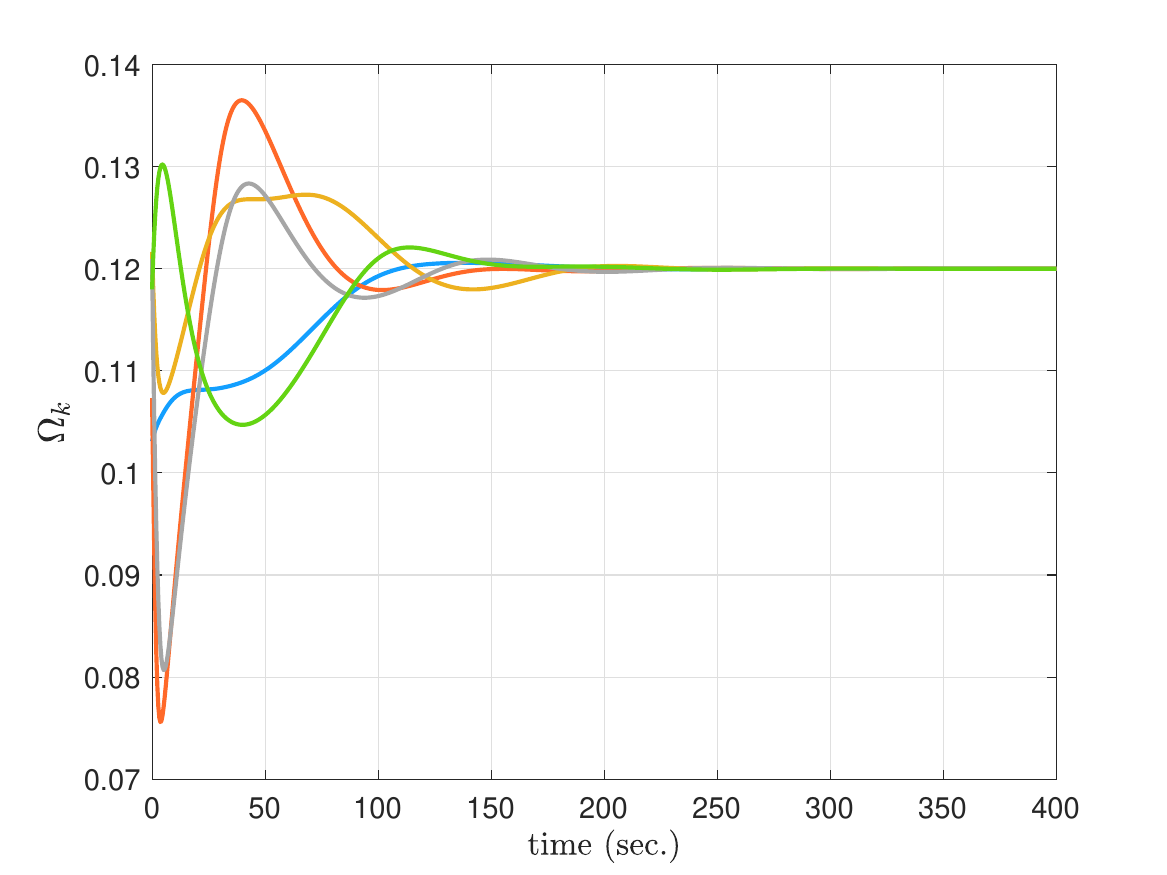}\label{tp_curv_controls}}}
	\caption{Agents with phase synchronization in the transformed plane under controllers \eqref{velocity_control_transformed_plane} and \eqref{curvature_control_transformed_plane}.}
	\label{plot4}
\end{figure*}

\begin{figure*}[t!]
	\centering{
		\subfigure[$v_k$-phase-shift synchronization]{\includegraphics[width=4.8cm]{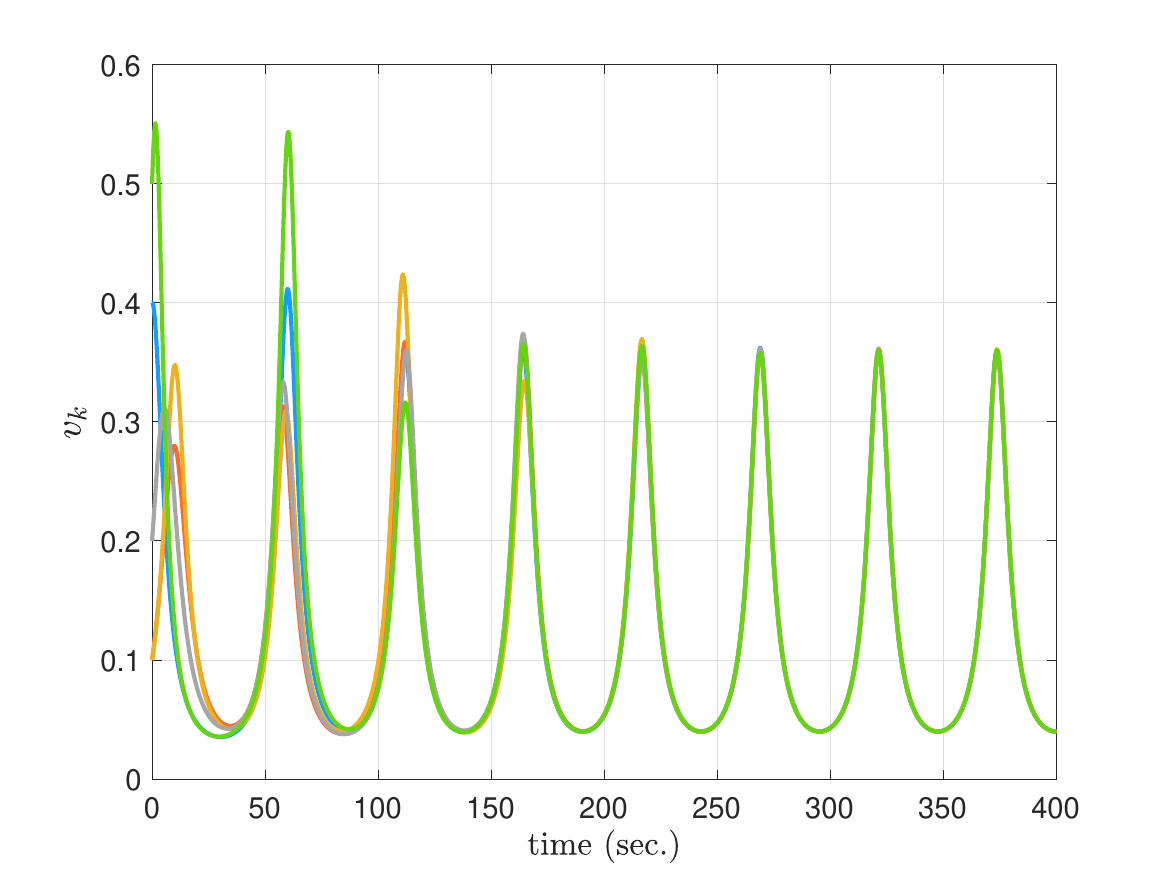}\label{tp_sync_gama}}\hspace*{-0.5cm}
		\subfigure[$v_k$-phase-shift balancing]{\includegraphics[width=4.8cm]{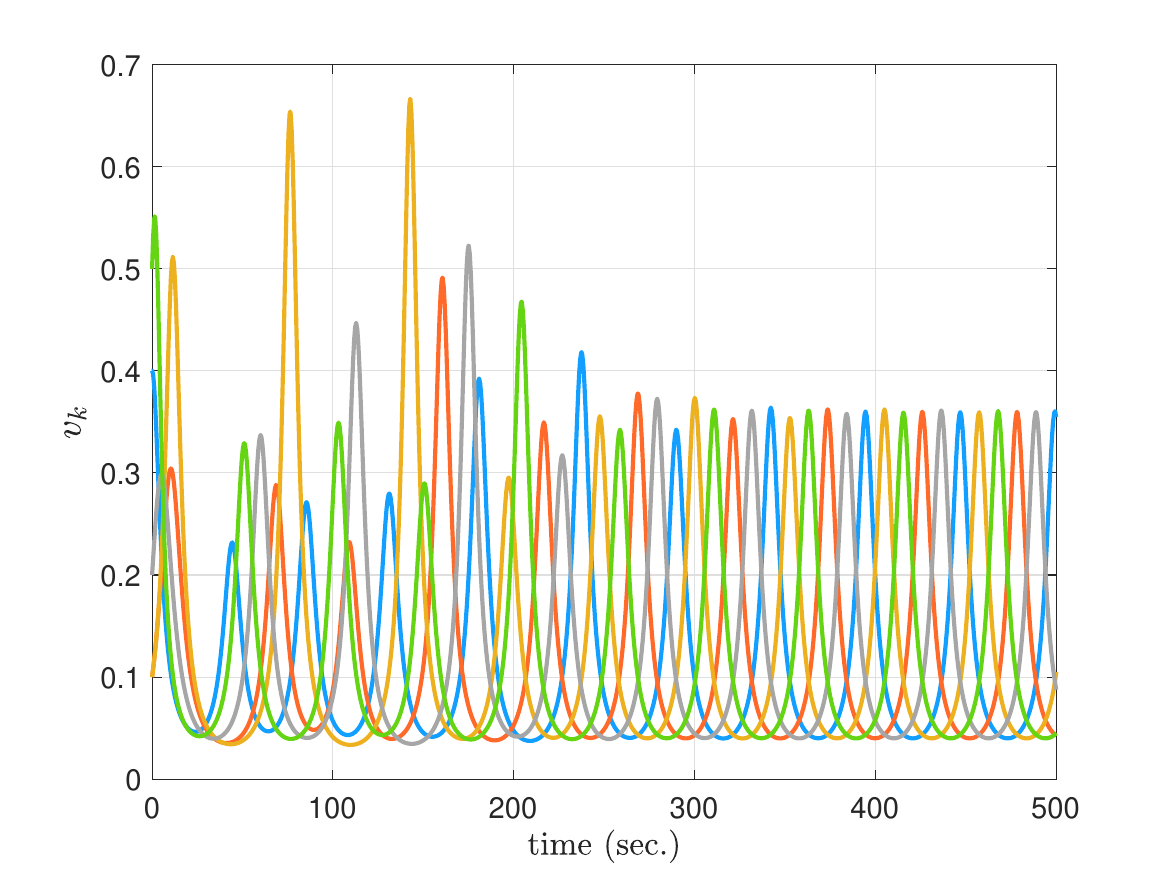}\label{tp_balan_gama}}\hspace*{-0.5cm}
		\subfigure[$s_k$-phase synchronization]{\includegraphics[width=4.8cm]{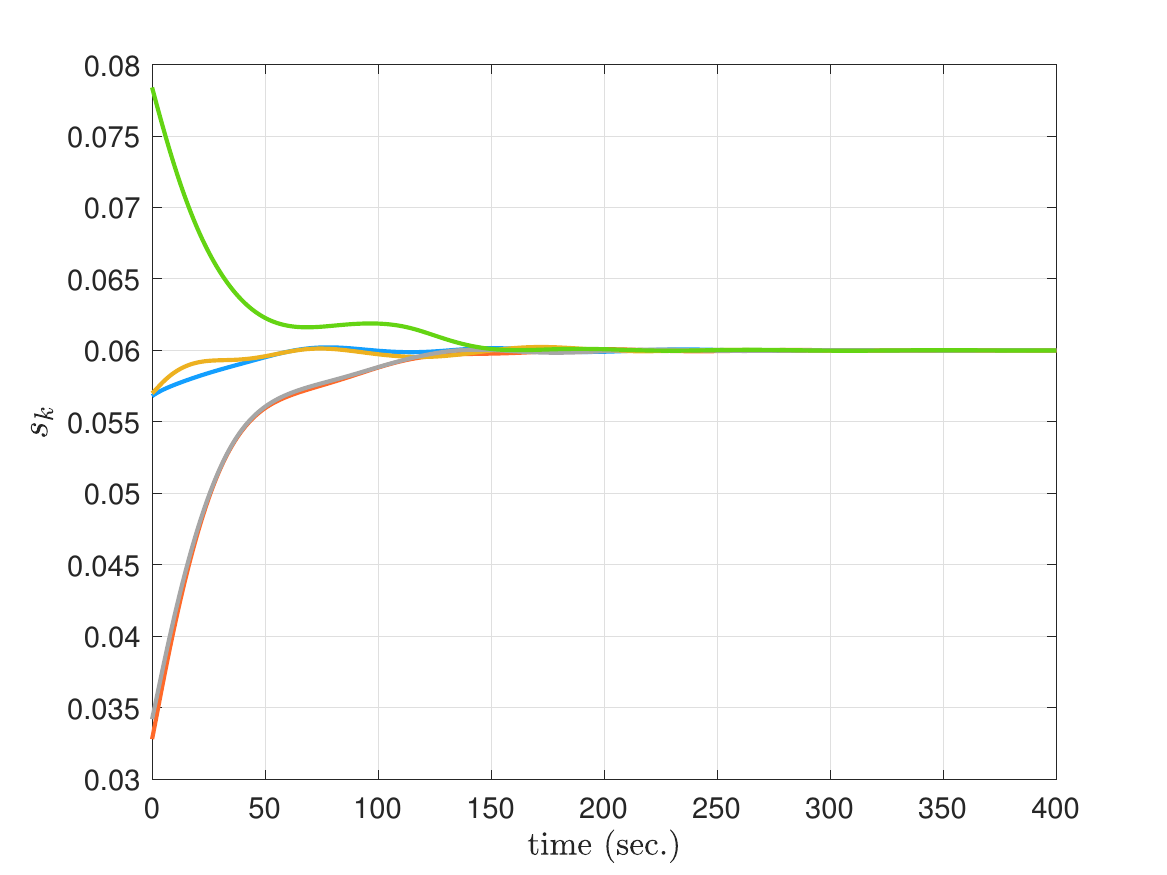}\label{ap_sync_theta}}\hspace*{-0.5cm}
		\subfigure[$s_k$-phase balancing]{\includegraphics[width=4.8cm]{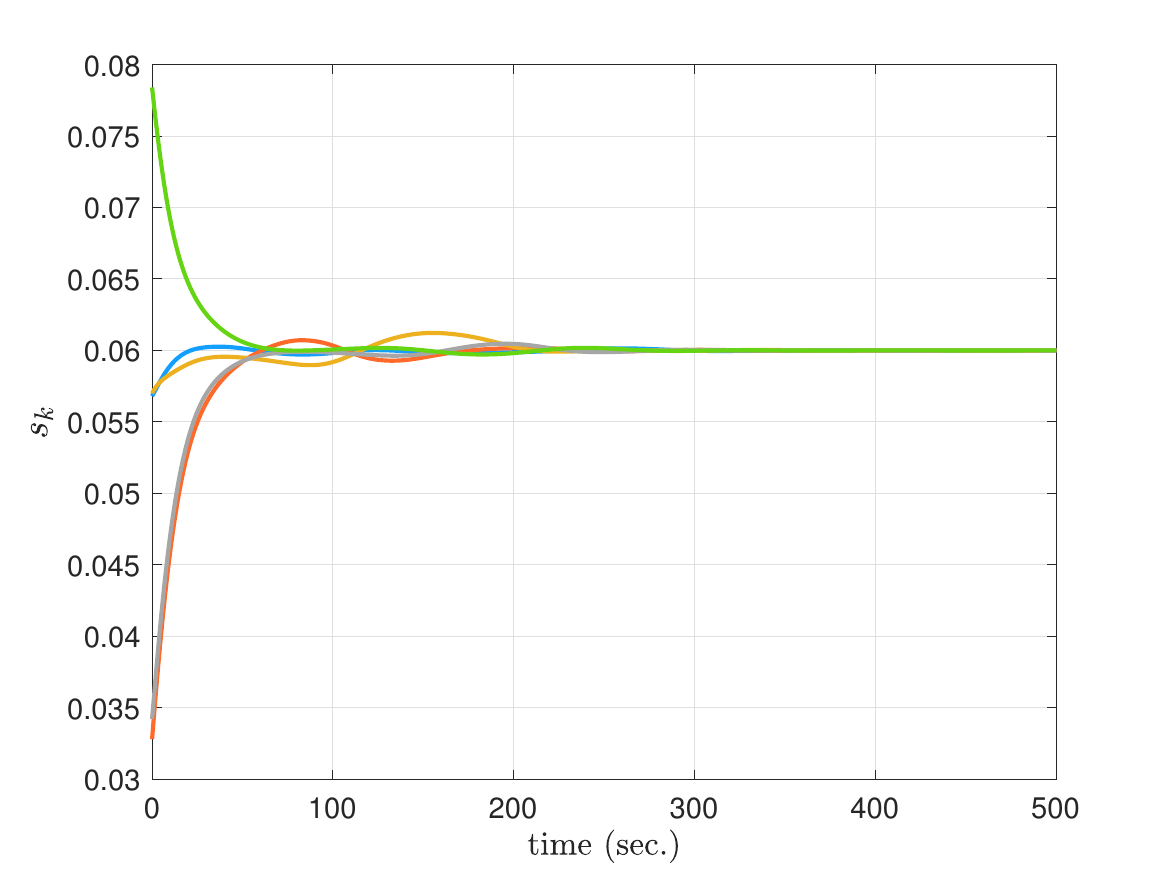}\label{ap_dfv_theta}}}
	\caption{Agents' speed in original and transformed planes.}
	\label{plot5}
\end{figure*}

\begin{figure*}[t!]
	\centering{
		\subfigure[Agents' positions at $t=100$ sec.]{\includegraphics[width=3.5cm]{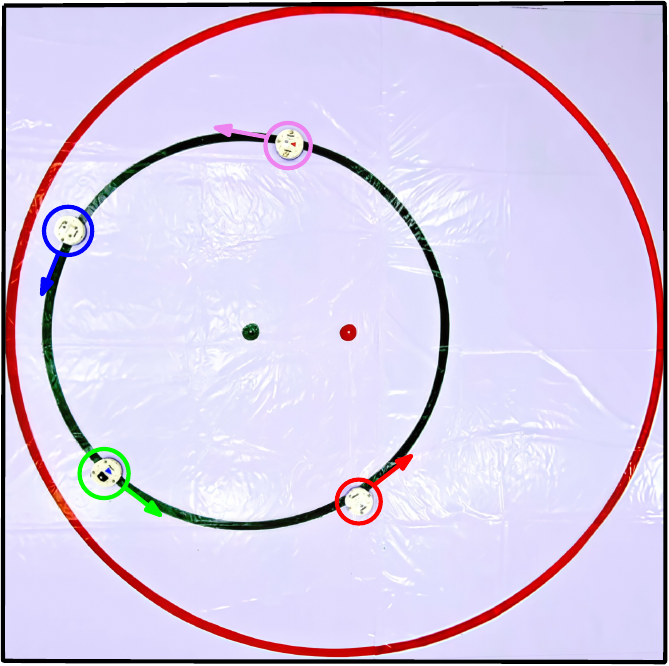}\label{screenshot_trajectory}}\hspace*{0.2cm}
		\subfigure[Trajectory in track manager]{\includegraphics[width=3.7cm]{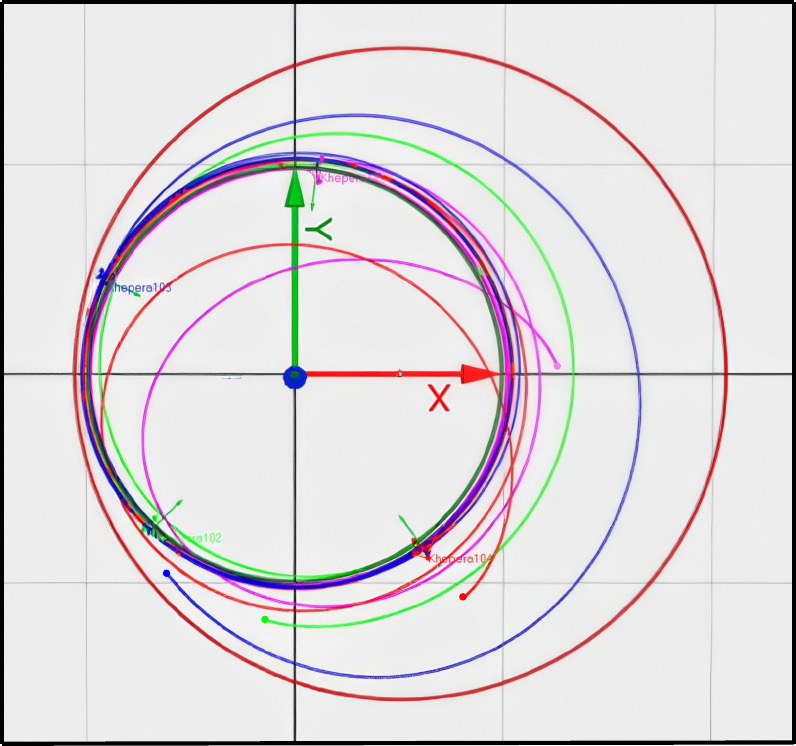}\label{qtm_plot}}
		\subfigure[$u_k$ with $\alpha=\alpha_s$]{\includegraphics[width=5.2cm]{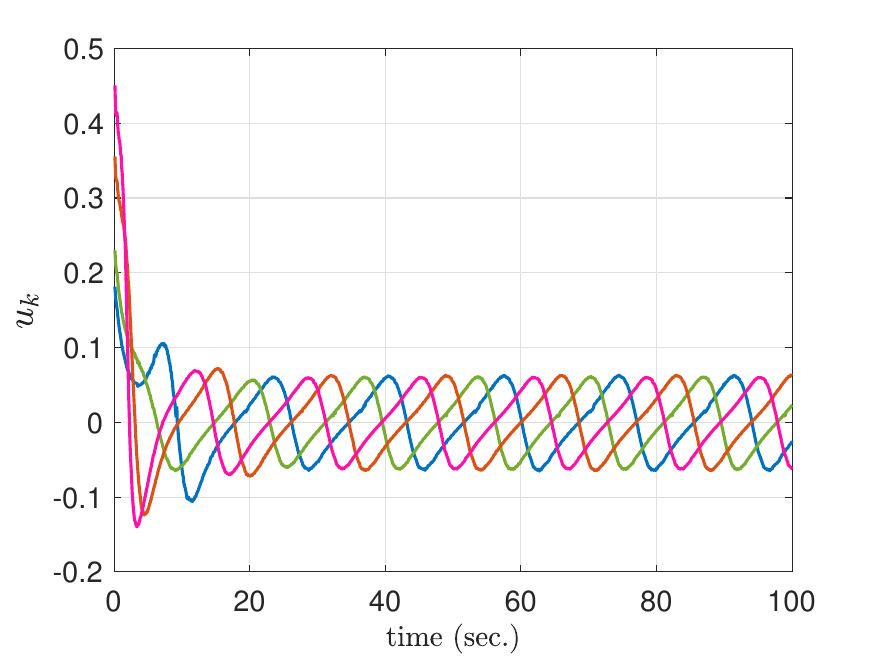}\label{hardware_dfv_ac}}\hspace*{-0.4cm}
		\subfigure[$\omega_k$ with $\alpha=\alpha_s$]{\includegraphics[width=5.2cm]{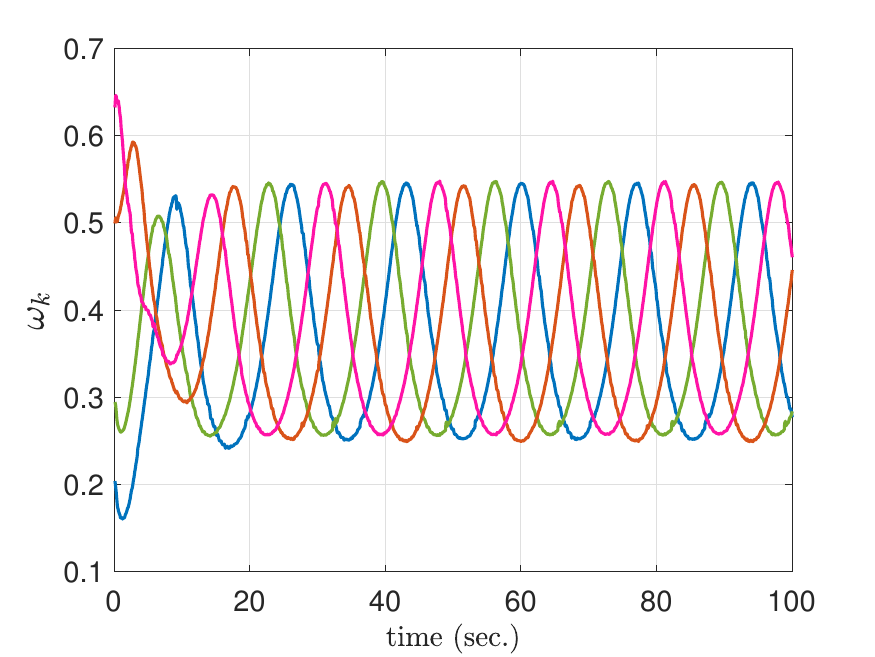}\label{hardware_dfv_tc}}}
	\caption{Experimental results with Khepera IV robots in M\"{o}bius phase-shift-coupled balancing in the original plane.}
	\label{plot6}
\end{figure*}

\begin{remark}
Please note that the controllers \eqref{velocity_control_transformed_plane} and \eqref{curvature_control_transformed_plane} involve three gain parameters: $\kappa_1 > 0$, $\kappa_2 > 0$ and $\mathcal{K}$. Here, $\mathcal{K} > 0$ corresponds to balancing and $\mathcal{K} < 0$ corresponds to synchronization. The gain $\kappa_1$ governs the convergence rate of the positional error $\mathcal{E}_k$, while $\kappa_2$ regulates the speed error $\tilde{s}_k$, determining how quickly the agents' speeds converge to the desired value $s_d$ in the transformed plane. In principle, any choice of these gains (with the appropriate sign for $\mathcal{K}$) ensures applicability of the proposed controllers. However, in practice, there is a trade-off: larger gains lead to faster convergence but may result in higher control effort, and hence must be selected carefully to respect the actuation limits of the robotic platform. From an implementation perspective, these gains are set empirically in a simulation environment to obtain the expected performance characteristics (such as transient and steady-state performance) and control profile, before implementation on the real robotic system.
\end{remark}

\section{Conclusion and Future Remarks}\label{section_8_conclusion_and_future_remarks}
We investigated collective circular motion control of unicycle agents under M\"{o}bius phase-shift-coupled synchronization and balancing, ensuring that their trajectories remain bounded within a non-concentric circular boundary. Our solution methodology relied on formulating an equivalent problem in the transformed plane via the M\"{o}bius transformation, and subsequently deriving the controllers applied to the agents in the original plane. The equivalence of the agents' models across the two planes revealed the conformal property of the M\"{o}bius transformation and highlighted a resemblance to transfer functions in frequency-domain analysis from classical control. Moreover, we established that the synchronization property remains invariant under the M\"{o}bius transformation. Depending on the root of \eqref{mobius_roots}, the equivalent Problem~\ref{problem_transformed_plane} can be interpreted either as a trajectory-constraining problem or a obstacle-avoidance problem, both of which ultimately yield a solution to the original Problem~\ref{problem_actual_plane}. The control analysis relied fundamentally on the design parameters $s_d$ and $\delta_S$, which guarantee forward motion in both planes, and on the sign of the gain term $\mathcal{K}$ in \eqref{u_k_original_plane} and \eqref{omega_k_original_plane}, which determined the emergent phase pattern. The orthogonality property \eqref{ortho_prop} was instrumental in establishing the negative semi-definiteness of the composite Lyapunov function, and hence asymptotic stability. Simulations and experiments demonstrated the effectiveness of the proposed approach.

The present work does not address inter-agent collision avoidance, which remains an important direction for future research. Other worthwhile extensions include establishing exponential convergence at a prescribed rate; selecting control gains to satisfy desired safety margins and actuator limits; realizing more general phase patterns beyond those induced by the M\"{o}bius transformation; extending the framework to three-dimensional settings; incorporating dynamic interaction topologies; and analyzing robustness to measurement noise and model uncertainties.


\appendix

\section*{Supporting Definitions and Results}

\begin{lem}[\hspace{-.1pt}\cite{singh2025mobius}]\label{lem_chi_zeta}
	The angles $\chi_k$ in \eqref{chi_argument_form} and $\zeta_k$ in \eqref{zeta_argument_form} are:
	\begin{align}
		\label{chi}	\chi_k ~ & \hspace*{-0.2cm} = -\arctan\left[\frac{2\alpha|r_k|\sin\phi_k + \alpha^2|r_k|^2\sin2\phi_k}{1 + 2\alpha|r_k|\cos\phi_k + \alpha^2|r_k|^2 \cos2\phi_k}\right],\\
		\label{zeta}	\zeta_k &= -\arctan\left[\frac{- 2|\rho_k|\sin\psi_k + |\rho_k|^2\sin2\psi_k}{1 - 2|\rho_k|\cos\psi_k + |\rho_k|^2\cos2\psi_k}\right].
	\end{align}
\end{lem}

\begin{lem}[Convergence with BLF \hspace{-.1pt}\cite{tee2009barrier}]\label{lem_blf_convergence}
For any positive constant $\varrho$, let $\Xi \triangleq \{\xi \in \mathbb{R} \mid -\varrho < |\xi| < \varrho\} \subset \mathbb{R}$ and $\mathcal{N} \triangleq \mathbb{R}^\ell  \times \Xi \subset \mathbb{R}^{\ell+1}$ be open sets. Consider the system $\dot{\pmb{\varphi}} = \pmb{h}(t, \pmb{\varphi})$, where, $\pmb{\varphi} \triangleq [\pmb{w}, \xi]^\top \in \mathcal{N}$, and $\pmb{h} : \mathbb{R}_+ \times \mathcal{N} \to \mathbb{R}^{\ell+1}$ is piecewise continuous in $t$ and locally Lipschitz in $\pmb{\varphi}$, uniformly in $t$, on $\mathbb{R}_+ \times \mathcal{N}$. Suppose there exist functions $U: \mathbb{R}^\ell \times \mathbb{R}_{+} \to\mathbb{R}_+$ and $V_1: \Xi \to \mathbb{R}_+$, continuously differentiable and positive definite in their respective domains, such that $ V_1(\xi) \to \infty \ \text{as} \ |\xi| \to \varrho$ and $\varpi_1(\|\pmb{w}\|) \leq U(\pmb{w}, t) \leq \varpi_2(\|\pmb{w}\|)$, where $\varpi_1$ and $\varpi_2$ are class $\mathcal{K}_\infty$ functions. Let $V(\pmb{\varphi}) \triangleq V_1(\xi) + U(\pmb{w}, t)$, and $\xi(0) \in \Xi$. If it holds that $\dot{V} = (\nabla V)^\top{\pmb{h}} \leq 0$, in the set $\xi \in \Xi$, then $\xi(t) \in\Xi, \forall t \in [0, \infty)$.
\end{lem} 

\begin{lem}[\hspace{-.1pt}Spectrum of circulant graphs {\cite[pg. 34]{gray2006toeplitz}}]\label{lem_circulant_graphs}
	Let $\mathcal{L}$ be the Laplacian of an undirected circulant graph $\mathcal{G}$ with $N$ vertices. Define $\Phi_k \triangleq (k-1)({2\pi}/{N})$, for $k = 1, \ldots, N$. Then, the vectors $\pmb{f}^{(\ell)} \triangleq {\mathrm{e}}^{i(\ell-1)\pmb{\Phi}}, \ell = 1, \ldots, N$, form a basis of $N$ orthogonal eigenvectors of $\mathcal{L}$. The unitary matrix $\mathcal{F}$, whose columns are the $N$ (normalized) eigenvectors $(1/\sqrt{N}) \pmb{f}^{(\ell)}$, diagonalizes $\mathcal{L}$, that is, $\mathcal{L} = \mathcal{F} \Lambda \mathcal{F}^\star$, where $\Lambda \triangleq \text{diag}\{0, \lambda_2, \ldots, \lambda_N\} \succeq 0$ is the (real) diagonal matrix of the eigenvalues of $\mathcal{L}$, and $\mathcal{F}^\star$ denotes the conjugate transpose of $\mathcal{F}$.	
\end{lem}

\begin{lem}[\hspace{-.1pt}Critical points of $\mathcal{U}(\pmb{\gamma})$ in \eqref{potential_synchro_balanc} {\cite[pg.~161]{jain2018collective}}]\label{lem_critical_points}
Let $\mathcal{L}$ be the Laplacian of an undirected and connected graph $\mathcal{G}$ with $N$ vertices. Consider the Laplacian-based potential function $\mathcal{U}(\pmb{\gamma})$ defined by \eqref{potential_synchro_balanc}. If ${\rm e}^{i \pmb{\gamma}}$ is an eigenvector of $\mathcal{L}$, then $\pmb{\gamma}$ is a critical point of $\mathcal{U}(\pmb{\gamma})$, and $\pmb{\gamma}$ is either synchronized or balanced. The potential $\mathcal{U}(\pmb{\gamma})$ reaches its global minimum if and only if $\pmb{\gamma}$ is synchronized. If $\mathcal{G}$ is circulant, then $\mathcal{U}(\pmb{\gamma})$ attains its global maximum in a balanced phase arrangement of ${\rm e}^{i \pmb{\gamma}}$ (i.e., $\pmb{1}^\top_N{\rm e}^{i \pmb{\gamma}} = 0$), with $\mathcal{U}(\pmb{\gamma})$ being upper bounded by $(N/2)\lambda_{\max}$ where $\lambda_{\max}$ is the maximum eigenvalue of $\mathcal{L}$.   	
\end{lem}


\bibliographystyle{IEEEtran}
\bibliography{References}

\end{document}